\documentclass[pra,twocolumn,floatfix,superscriptaddress,notitlepage,longbibliography]{revtex4-1}

\usepackage{appendix}

\usepackage{amsmath}
\usepackage{amsfonts}
\usepackage{amssymb}
\usepackage{amsthm}

\usepackage[normalem]{ulem}
\usepackage{graphicx, color, graphpap}
\usepackage{enumitem}
\usepackage{multirow}
\usepackage[colorlinks=true,citecolor=blue,linkcolor=magenta]{hyperref}
\usepackage[T1]{fontenc}
\usepackage{dsfont}
\usepackage{verbatim}
\usepackage{mathtools}
\usepackage{titlesec}
\usepackage{float}
\usepackage{nicefrac}
\usepackage{tikz}
\usepackage{apptools}
\usepackage[linesnumbered,ruled,vlined]{algorithm2e}
\SetKwInput{kwInit}{Init}

\long\def\ca#1\cb{} 

\newcommand{\braket}[2]{\langle #1 \hspace{1pt} | \hspace{1pt} #2 \rangle}

\newcommand{\ketbra}[2]{| \hspace{1pt} #1 \rangle \langle #2 \hspace{1pt} |}

\newcommand{\norm}[2][]{#1| \! #1| #2 #1| \! #1|}

\newcommand{\ket}[1]{|#1\rangle}               
\newcommand{\bra}[1]{\langle #1|}              

\DeclareMathOperator*{\E}{{\mathbb{E}}}

\newcommand{\poly}{\operatorname{poly}}

\newcommand{\EC}{\mathcal{E}}

\newcommand{\OC}{\mathcal{O}}

\newcommand{\Tr}{{\rm Tr}}

\renewcommand{\geq}{\geqslant}
\renewcommand{\leq}{\leqslant}

\renewcommand{\vec}[1]{\boldsymbol{#1}}  

{}
{}
\theoremstyle{definition}
\newtheorem{theorem}{Theorem}
\newtheorem{lemma}{Lemma}
\newtheorem{corollary}{Corollary}

\newtheorem{definition}{Definition}
\newtheorem{remark}{Remark}
\newtheorem*{thm-mother}{Theorem \ref{thm-mother}}

\AtAppendix{\counterwithin{theorem}{section}}
\AtAppendix{\counterwithin{corollary}{section}}
\AtAppendix{\counterwithin{lemma}{section}}
\AtAppendix{\counterwithin{proposition}{section}}
\AtAppendix{\counterwithin{definition}{section}}
\AtAppendix{\counterwithin{remark}{section}}
\AtAppendix{\counterwithin{example}{section}}
\AtAppendix{\counterwithin{property}{section}}
\AtAppendix{\counterwithin{problem}{section}}

\newcommand{\stkout}[1]{\ifmmode\text{\sout{\ensuremath{#1}}}\else\sout{#1}\fi}
\newif\ifverbose
\verbosetrue

\begin{document}

\title{Generalization in quantum machine learning from few training data}

\author{Matthias C.~Caro}
\email{caro@ma.tum.de}
\affiliation{Department of Mathematics, Technical University of Munich, Garching, Germany}
\affiliation{Munich Center for Quantum Science and Technology (MCQST), Munich, Germany}

\author{Hsin-Yuan Huang}
\affiliation{Institute for Quantum Information and Matter, Caltech, Pasadena, CA, USA}
\affiliation{Department of Computing and Mathematical Sciences, Caltech, Pasadena, CA, USA}

\author{M.~Cerezo}
\affiliation{Information Sciences, Los Alamos National Laboratory, Los Alamos, NM 87545, USA}
\affiliation{Center for Nonlinear Studies, Los Alamos National Laboratory, Los Alamos, NM 87545, USA
}

\author{Kunal~Sharma} 
\affiliation{Joint Center for Quantum Information and Computer Science, University of Maryland, College Park, Maryland 20742, USA}

\author{Andrew Sornborger} 
\affiliation{Information Sciences, Los Alamos National Laboratory, Los Alamos, NM 87545, USA}
\affiliation{Quantum Science Center, Oak Ridge, TN 37931, USA}

\author{Lukasz Cincio}
\affiliation{Theoretical Division, Los Alamos National Laboratory, Los Alamos, NM 87545, USA}

\author{Patrick J.~Coles}
\affiliation{Theoretical Division, Los Alamos National Laboratory, Los Alamos, NM 87545, USA}

\begin{abstract}
Modern quantum machine learning (QML) methods involve variationally optimizing a parameterized quantum circuit on a training data set, and subsequently making predictions on a testing data set (i.e., generalizing).
In this work, we provide a comprehensive study of generalization performance in QML after training on a limited number~$N$ of training data points.
We show that the generalization error of a quantum machine learning model with $T$ trainable gates scales at worst as $\sqrt{T/N}$. When only $K \ll T$ gates have undergone substantial change in the optimization process, we prove that the generalization error improves to $\sqrt{K / N}$.
Our results imply that the compiling of unitaries into a polynomial number of native gates, a crucial application for the quantum computing industry that typically uses exponential-size training data, can be sped up significantly.
We also show that classification of quantum states across a phase transition with a quantum convolutional neural network requires only a very small training data set.
Other potential applications include learning quantum error correcting codes or quantum dynamical simulation.
Our work injects new hope into the field of QML, as good generalization is guaranteed from few training data.  
\end{abstract}

\maketitle

\section{Introduction}


The ultimate goal of machine learning (ML) is to make accurate predictions on unseen data. This is known as generalization, and significant effort has been expended to understand the generalization capabilities of classical ML models. For example, theoretical results have been formulated as upper bounds on the generalization error as a function of the training data size and the model complexity~\cite{vapnik1971uniform, pollard1984convergence, gine1984some, dudley1999uniform, bartlett2002rademacher}. Such bounds provide guidance as to how much training data is required and/or sufficient to achieve accurate generalization.


Quantum machine learning (QML) is an emerging field that has generated great excitement~\cite{biamonte2017quantum,schuld2015introduction,schuld2014quest,dunjko2018machine}. Modern QML typically involves training a parameterized quantum circuit in order to analyze either classical or quantum data sets~\cite{cerezo2020variationalreview,havlivcek2019supervised,farhi2018classification,romero2017quantum,wan2017quantum,larocca2021theory,schatzki2021entangled}. Early results indicate that, for classical data analysis, QML models may offer some advantage over classical models under certain circumstances~\cite{huang2021power,abbas2020power,liu2021rigorous}. It has also been proven that QML models can provide an exponential advantage in sample complexity for analyzing quantum data \cite{huang2021information, aharonov2021quantum}.


However, little is known about the conditions needed for accurate generalization in QML. Significant progress has been made in understanding the trainability of QML models~\cite{mcclean2018barren,cerezo2020cost,cerezo2020impact,arrasmith2020effect,holmes2021connecting,pesah2020absence,volkoff2021large,sharma2020trainability,holmes2020barren,marrero2020entanglement,uvarov2020barren,patti2020entanglement,abbas2020power,wang2020noise,larocca2021diagnosing, thanaslip2021subtleties}, but trainability is a separate question from generalization~\cite{abbas2020power,banchi2021generalization,du2021efficient}. Overfitting of training data could be an issue for QML, just as it is for classical machine learning. Moreover, the training data size required for QML generalization has yet to be fully studied. Na\"ively, one could expect that an exponential number of training points are needed when training a function acting on an exponentially large Hilbert space. For instance, some studies have found that, exponentially in $n$, the number of qubits, large amounts of training data would be needed, assuming that one is trying to train an arbitrary unitary~\cite{poland2020no,sharma2020reformulation}. This is a concerning result, since it would imply exponential scaling of the resources required for QML, which is precisely what the field of quantum computation would like to avoid.


In practice, a more relevant scenario to consider instead of arbitrary unitaries is learning a unitary that can be represented by a polynomial-depth quantum circuit. This class of unitaries corresponds to those that can be efficiently implemented on a quantum computer, and it is exponentially smaller than that of arbitrary unitaries. More generally, one could consider a QML model with $T$ parameterized gates and relate the training data size $N$ needed for generalization to $T$. Even more general would be to consider generalization error a dynamic quantity that varies during the optimization.


In this work, we prove highly general theoretical bounds on the generalization error in variational QML: The generalization error is approximately upper bounded by $\sqrt{T/N}$. 
In our proofs, we first establish covering number bounds for the class of quantum operations that a variational QML model can implement. From these, we then derive generalization error bounds using the chaining technique for random processes.
A key implication of our results is that an efficiently implementable QML model, one such that \mbox{$T\in \mathcal{O}(\poly n )$}, only requires an efficient amount of training data, $N\in \mathcal{O}(\poly n )$, to obtain good generalization. This implication, by itself, will improve the efficiency guarantees of variational quantum algorithms~\cite{cerezo2020variationalreview,bharti2021noisy,endo2021hybrid} that employ training data, such as quantum autoencoders~\cite{romero2017quantum}, quantum generative adversarial networks~\cite{romero2021variational}, variational quantum error correction~\cite{johnson2017qvector,cong2019quantum}, variational quantum compiling~\cite{khatri2019quantum,sharma2019noise}, and variational dynamical simulation~\cite{cirstoiu2020variational,commeau2020variational,endo2020variational,li2017efficient}. It also yields improved efficiency guarantees for classical algorithms that simulate QML models.

We furthermore refine our bounds to account for the optimization process. We show that generalization improves if only some parameters have undergone substantial change during the optimization. Hence, even if we used a number of parameters $T$ larger than the training data size $N$, the QML model could still generalize well if only some of the parameters have changed significantly.
This suggests that QML researchers should be careful not to overtrain their models especially when the decrease in training error is insufficient.

To showcase our results, we consider quantum convolutional neural networks (QCNNs)~\cite{cong2019quantum,pesah2020absence}, a QML model that has received significant attention. QCNNs have only $T =  \mathcal{O}(\log n )$ parameters and yet they are capable of classifying quantum states into distinct phases. 
Our theory guarantees that QCNNs have good generalization error for quantum phase recognition with only polylogarithmic training resources, $N \in \mathcal{O}(\log^2 n )$.
We support this guarantee with a numerical demonstration, which suggests that even constant-size training data can suffice.


Finally, we highlight the task of quantum compiling, a crucial application for the quantum computing industry. State-of-the-art classical methods for approximate optimal compiling of unitaries often employ exponentially large training data sets~\cite{cincio2018learning,cincio2021machine,qFactor}. However, our work indicates that only polynomial-sized data sets are needed, suggesting that state-of-the-art compilers could be further improved. Indeed, we numerically demonstrate the surprisingly low data cost of compiling the quantum Fourier transform at relatively large scales.

\section{Results}\label{sec:results}

\subsection{Framework}\label{sec:Background}

Let us first outline our theoretical framework.
We consider a quantum machine learning model (QMLM) as being a parameterized quantum channel, i.e., a completely positive trace preserving (CPTP) map that is parameterized. We denote a QMLM as $\EC^{\mathrm{QMLM}}_{\vec{\alpha}}(\cdot)$ where $\vec{\alpha} = (\vec{\theta}, \vec{k})$ denotes the set of parameters, including continuous parameters $\vec{\theta}$ inside gates, as well as discrete parameters $\vec{k}$ that allow the gate structure to vary. We make no further assumptions on the form of the dependence of the CPTP map $\EC^{\mathrm{QMLM}}_{\vec{\alpha}}(\cdot)$ on the parameters $\vec{\alpha}$. During the training process, one would optimize the continuous parameters $\vec{\theta}$ and potentially also the structure $\vec{k}$ of the QMLM.


A QMLM takes input data in the form of quantum states. For classical data $x$, the input is first encoded in a quantum state via a map $x\mapsto\rho(x)$. 
This allows the data to be either classical or quantum in nature, since regardless it is eventually encoded in a quantum state. 
We assume that the data encoding is fixed in advance and not optimized over. We remark here that our results also apply for more general encoding strategies involving data re-uploading~\cite{perez2020data}, as we explain in Remark~\ref{rmk:data-reuploading}.


For the sake of generality, we allow the QMLM to act on a subsystem of the state $\rho(x)$. Hence, the output state can be written as $(\mathcal{E}^{\mathrm{QMLM}}_{\vec{\alpha}}\otimes \operatorname{id})(\rho(x))$. For a given data point $(x_i, y_i)$, we can write the loss function as
\begin{equation}
    \ell (\vec{\alpha};x_i,y_i)
    =\Tr\left[O^{\mathrm{loss}}_{x_i,y_i}(\mathcal{E}^{\mathrm{QMLM}}_{\vec{\alpha}}\otimes \operatorname{id})(\rho(x_i))\right],\label{eq:loss-function}
\end{equation}
for some Hermitian observable $O^{\mathrm{loss}}_{x_i, y_i}$. As is common in classical learning theory, the prediction error bounds will depend on the largest (absolute) value that the loss function can attain. In our case, we therefore assume $C_\mathrm{loss}:=\sup_{x,y}\norm{O^{\mathrm{loss}}_{x, y}}<\infty$, i.e., the spectral norm can be bounded uniformly over all possible loss observables.


In Eq.~\eqref{eq:loss-function}, we take the measurement to act on a single copy of the output of the QMLM $\mathcal{E}^{\mathrm{QMLM}}_{\vec{\alpha}}(\cdot)$ upon input of (a subsystem of) the data encoding state $\rho(x_i)$. At first this looks like a restriction. 
However, note that one can choose $\mathcal{E}^{\mathrm{QMLM}}_{\vec{\alpha}}(\cdot)$ to be a tensor product of multiple copies of a QMLM, each with the same parameter setting, applied to multiple copies of the input state. Hence our framework is general enough to allow for global measurements on multiple copies. In this addition to the aforementioned situation, we further study the case in which trainable gates are more generally reused.


For a training dataset $S=\{(x_i,y_i)\}_{i=1}^{N}$ of size $N$, the average loss for parameters $\vec{\alpha}$ on the training data is
\begin{equation}
    \hat{R}_S(\vec{\alpha})
    = \frac{1}{N} \sum_{i=1}^N \ell (\vec{\alpha}; x_i, y_i),
\end{equation}
which is often referred to as the \emph{training error}.
When we obtain a new input $x$, the \emph{prediction error} of a parameter setting $\vec{\alpha}$ is taken to be the expected loss
\begin{equation}
    R(\vec{\alpha})
    = \E_{(x, y)\sim P} \left[\ell (\vec{\alpha}; x, y)\right],
\end{equation}
where the expectation is with respect to the distribution $P$ from which the training examples are generated.


Achieving small prediction error $R(\vec{\alpha})$ is the ultimate goal of (quantum) machine learning. As $P$ is generally not known, the training error $\hat{R}_S(\vec{\alpha})$ is often taken as a proxy for $R(\vec{\alpha})$. This strategy can be justified via bounds on the \emph{generalization error} 
\begin{equation}
    \operatorname{gen}(\vec{\alpha})
    = R(\vec{\alpha}) - \hat{R}_S(\vec{\alpha}),
\end{equation}
which is the key quantity that we bound in our theorems.

\subsection{Analytical Results}\label{sec:theorems}

\begin{figure*}
    \centering
    \includegraphics[width=1.0\linewidth]{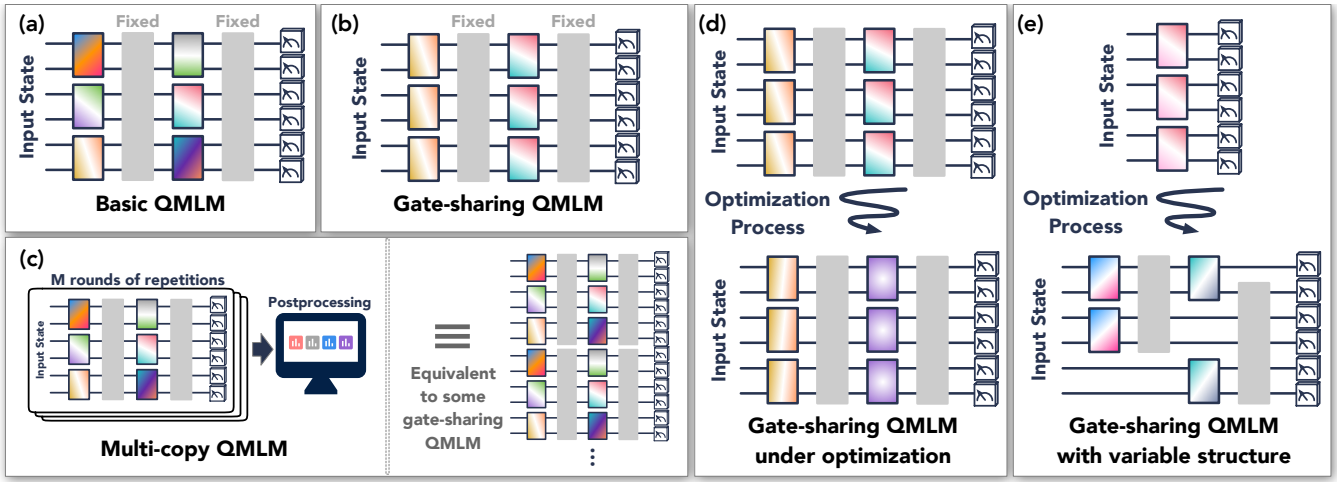}
    \caption{\textbf{Various types of Quantum Machine Learning Models (QMLMs).} 
    Panel (a) depicts a basic QMLM with $T=6$ independently parameterized gates. The gray boxes illustrate some global evolutions that are not trainable. 
    Panel (b) shows a gate-sharing QMLM with $T=2$ independently parameterized gates, each gate is repeatedly used for $M=3$ times. 
    In panel (c), we depict a multi-copy QMLM. We take measurement data from $M$ rounds of a basic QMLM with $T=6$ parameterized gates and post-process the measurement outcomes to produce an output. Running $M$ copies of a basic QMLM with $T$ gates is equivalent to running a gate-sharing QMLM with $T=6$ parameterized gates, in which each gate is repeated $M$ times. 
    Panel (d) describes a gate-sharing QMLM under optimization. The parameterized gate to the left undergoes a small change, while the one to the right undergoes a large change. If we sort the changes $\Delta_1, \Delta_2$ from large to small, then $\Delta_1 \gg \Delta_2 \approx 0$. 
    Finally, panel (e) illustrates gate-sharing QMLM with variable structure. The number $T$ of parameterized gates changes throughout the optimization. The figure begins with $T=1$ and ends with $T=2$.}
    \label{fig:QMLM-settings}
\end{figure*}

We prove probabilistic bounds on the generalization error of a QMLM. Our bounds guarantee that a good performance on a sufficiently large training data set implies, with high probability, a good performance on previously unseen data points. In particular, we provide a precise meaning of ``sufficiently large'' in terms of properties of the QMLM and the employed training procedure.

Fig.~\ref{fig:QMLM-settings} gives an overview of the different scenarios considered in this work. We begin with the basic form of our result. We consider a QMLM that has arbitrarily many non-trainable global quantum gates and $T$ trainable local quantum gates. Here, by local we mean $\kappa$-local for some $n$-independent locality parameter $\kappa$, and a local quantum gate can be a unitary or a quantum channel acting on $\kappa$ qubits. Then we have the following bound on the generalization error for the QMLM with final parameter setting $\vec{\alpha}^\ast$ after training:

\begin{theorem}[Basic QMLM] \label{thm:simpleQMLM}
For a QMLM with $T$ parameterized local quantum channels, with high probability over training data of size $N$, we have that
\begin{equation}
    \operatorname{gen}(\vec{\alpha}^\ast) \in \mathcal{O}\left( \sqrt{\frac{T \log T}{N}}\right).
\end{equation}
\end{theorem}

\begin{remark}\label{rmk:sample-complexity-bound}
    Theorem \ref{thm:simpleQMLM} directly implies sample complexity bounds: For any $\varepsilon>0$, we can, with high success probability, guarantee that $\operatorname{gen}(\vec{\alpha}^\ast) \leq \varepsilon$, already with training data of size $N\sim\nicefrac{T\log T}{\varepsilon^2}$, which scales effectively linearly with $T$, the number of parameterized gates. 
    
    For efficiently implementable QMLMs with $T\in\mathcal{O}(\poly n)$, a sample size of $N\in\mathcal{O}\left(\nicefrac{\poly n}{\varepsilon^2}\right)$ is already sufficient. More concretely, if $T\in\mathcal{O}\left(n^D\right)$ for some degree $D$, then the corresponding sufficient sample complexity obtained from Theorem \ref{thm:simpleQMLM} satisfies $N\in\tilde{\mathcal{O}}\left(\nicefrac{n^D}{\varepsilon^2}\right)$, where the $\tilde{\mathcal{O}}$ hides factors logarithmic in $n$. In the NISQ era~\cite{preskill2018quantum}, we expect the number $T$ of trainable maps to only grow mildly with the number of qubits, e.g., as in the architectures discussed in Refs.~\cite{cong2019quantum,romero2018strategies,abbas2020power}.
    In this case, Theorem~\ref{thm:simpleQMLM} gives an especially strong guarantee. 
\end{remark}

In various QMLMs, such as QCNNs, the same parameterized local gates are applied repeatedly. One could also consider running the same QMLM multiple times to gather measurement data and then post-processing that data.
In both cases, one should consider the QMLM as using the same parameterized local gates repeatedly. We assume each gate to be repeated at most $M$ times.
A direct application of Theorem \ref{thm:simpleQMLM} would suggest that we need a training data size $N$ of roughly $MT$, the total number of parameterized gates. However, the required number of training data actually is much smaller:

\begin{theorem}[Gate-sharing QMLM]\label{thm:gate-sharing-QMLM}
Consider a QMLM with $T$ independently parameterized local quantum channels, where each channel is reused at most $M$ times. With high probability over training data of size $N$, we have
\begin{equation}
    \operatorname{gen}(\vec{\alpha}^\ast) \in \mathcal{O}\left( \sqrt{\frac{T \log (MT)}{N}} \right).
\end{equation}
\end{theorem}

Thus, good generalization, as in Remark \ref{rmk:sample-complexity-bound}, can already be guaranteed, with high probability, when the data size effectively scales linearly in $T$ (the number of independently parameterized gates) and only logarithmically in $M$ (the number of uses).
In particular, applying multiple copies of the QMLM in parallel does not significantly worsen the generalization performance compared to a single copy.
Thus, as we discuss in Remark~\ref{rmk:multiple-copy-postprocessing}, Theorem~\ref{thm:gate-sharing-QMLM} ensures that we can increase the number of shots used to estimate expectation values at the QMLM output without substantially harming the generalization behavior.

The optimization process of the QMLM also plays an important role in the generalization performance. Suppose that during the optimization process, the $t^{\textrm{th}}$ local gate changed by a distance $\Delta_t$. We can bound the generalization error by a function of the changes $\{\Delta_t\}_t$.

\begin{theorem}[Gate-sharing QMLM under optimization]\label{thm:gate-sharing-QMLM-optimization}
Consider a QMLM with $T$ independently parameterized local quantum channels, where the $t^{\textrm{th}}$ channel is reused at most $M$ times and is changed by $\Delta_t$ during the optimization. Assume $\Delta_1 \geq \ldots \geq \Delta_T$. With high probability over training data of size $N$, we have
\small
\begin{equation}\label{eq:gate-sharing-QMLM-optimization-main}
   \operatorname{gen}(\vec{\alpha}^\ast) \in \mathcal{O}\left( \min_{K = 0, \ldots, T} \!\left\{ \sqrt{\frac{K \log (MT)}{N}} +\!\!\!\! \sum_{k=K+1}^T \!\!\!\!M \Delta_{k} \right\} \!\!\right).
\end{equation}
\end{theorem}
\normalsize

When only $K \ll T$ local quantum gates have undergone a significant change, then the generalization error will scale at worst linearly with $K$ and logarithmically in the total number of parameterized gates $MT$. Given that recent numerical results suggest that the parameters in a deep parameterized quantum circuit only change by a small amount during training~\cite{shirai2021quantumtangent, liu2021representation}, Theorem~\ref{thm:gate-sharing-QMLM-optimization} may find application in studying the generalization behavior of deep QMLMs.

Finally, we consider a more advanced type of variable ansatz optimization strategy that is also adopted in practice~\cite{grimsley2019adaptive,tang2019qubit,bilkis2021semi,zhu2020adaptive}. Instead of fixing the structure of the QMLM, such as the number of parameterized gates and how the parameterized gates are interleaved with the fixed gates, the optimization algorithm could vary the structure, e.g., by adding or deleting parameterized gates. 
We assume that for each number $T$ of parameterized gates, there are $G_T$ different QMLM architectures.

\begin{theorem}[Gate-sharing QMLM with variable structure]\label{thm:gate-sharing-QMLM-variable-structure}
Consider a QMLM with an arbitrary number of parameterized local quantum channels, where for each $T > 0$, we have $G_T$ different QMLM architectures with $T$ parameterized gates. Suppose that after optimizing on the data, the QMLM has $T$ independently parameterized local quantum channels, each repeated at most $M$ times. Then, with high probability over input training data of size $N$, 
\begin{equation}
    \operatorname{gen}(\vec{\alpha}^\ast) \in \mathcal{O}\left( \sqrt{\frac{T \log (MT)}{N}} + \sqrt{\frac{\log(G_T)}{N}} \right).
\end{equation}
\end{theorem}

Thus, even if the QMLM can in principle use exponentially many parameterized gates, we can control the generalization error in terms of the number of parameterized gates used in the QMLM after optimization, and the dependence on the number of different architectures is only logarithmic.  This logarithmic dependence is crucial as even in the cases when $G_T$ grows exponentially with $T$, we have $\nicefrac{\log(G_T)}{N}\in\OC(\nicefrac{T}{N})$.

\subsection{Numerical Results}\label{sec:numerics}

In this section we present generalization error results obtained by simulating the following two QML implementations: (1) using a QCNN to classify states belonging to different quantum phases, and (2) training a parameterized quantum circuit to compile a quantum Fourier transform matrix. 

\subsubsection{Phase classification}

\begin{figure*}
    \centering
    \includegraphics[width=1.0\linewidth]{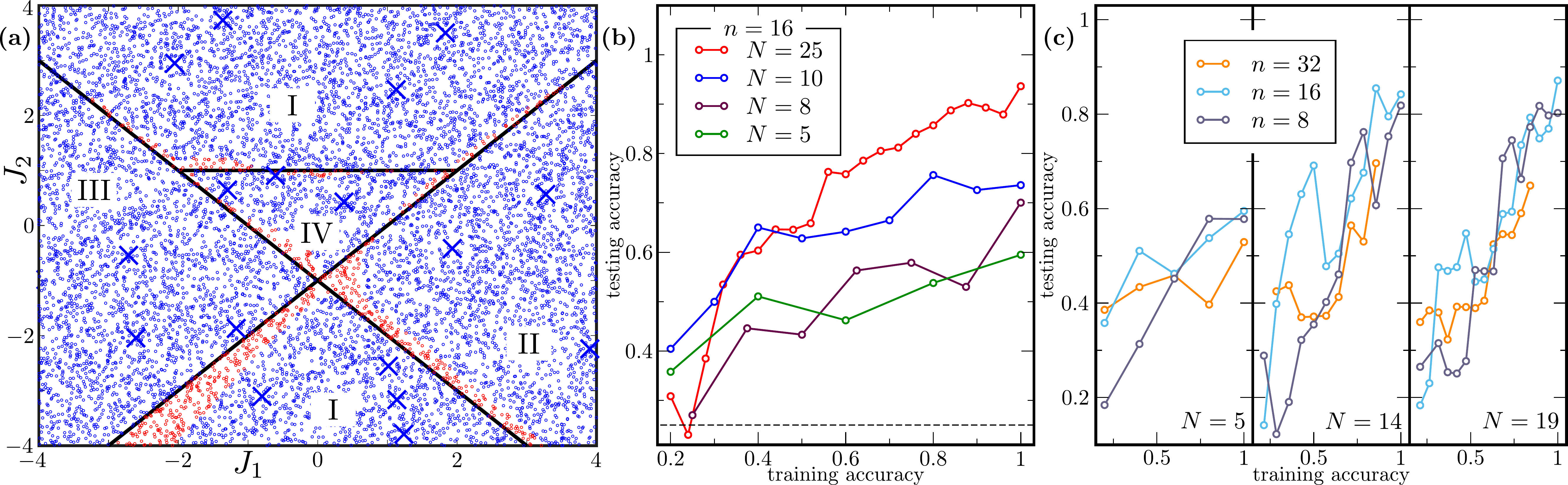}
    \caption{\textbf{Generalization performance of quantum phase recognition.} We employed the QCNN architecture for quantum phase recognition on ground states of the generalized cluster Hamiltonian $H$ of Eq.~\eqref{eq:Hamiltonian}. We evaluated the phase assigned by the QCNN to a point in the $J_1$-$J_2$-plane by sampling $8192$ computational basis measurement outcomes and taking the least frequent outcome as the predicted phase. 
    Panel (a) visualizes the performance of the QCNN for $16$-qubits, trained with $30$ data points, which were labelled according to the analytically determined phase diagram.
    Blue crosses denote training data points (not all $30$ are shown). 
    Blue (red) circles represent correctly (incorrectly) classified points.
    Panel (b) shows that, as the training data size increases, the training accuracy quickly becomes a good predictor for the testing accuracy on $10000$ randomly sampled points, i.e., the dependence of testing accuracy on training accuracy is approximately linear with slope increasing with $N$. The different points in the plot correspond to different parameter settings in the QCNN throughout the optimization. The dotted gray line shows the baseline accuracy of $25\%$ achieved by random guessing.
    Panel (c) shows that the improvement in the slope with growing training data size is similar for different numbers of qubits, reflecting the at-worst polylogarithmic dependence of $N$ on $n$ predicted by our bounds.
    }
    \label{fig:training-versus-testing-qpr}
\end{figure*}

The QCNN architecture introduced in~\cite{cong2019quantum} generalizes the model of (classical) convolutional neural networks with the goal of performing pattern recognition on quantum data. 
It is composed of so-called \emph{convolutional} and \emph{pooling} layers, which alternate. In a convolutional layer, a sequence of translationally invariant parameterized unitaries on neighbouring qubits is applied in parallel, which works as a filter between feature maps in different layers of the QCNN. Then, in the pooling layers, a subset of the qubits are measured to reduce the dimensionality of the state while preserving the relevant features of the data. Conditioned on the corresponding measurement outcomes, translationally invariant parameterized $1$-qubit unitaries are applied. The QCNN architecture has been employed for supervised QML tasks of classification of phases of matter and to devise quantum error correction schemes~\cite{cong2019quantum}. Moreover, QCNNs have been shown not to exhibit barren plateaus, making them a generically trainable QML architecture~\cite{pesah2020absence}.

The action of a QCNN can be considered as mapping an input state $\rho_{\text{in}}$ to an output state $\rho_{\text{out}}$ given as $\rho_{\text{out}}=\mathcal{E}^{\mathrm{QCNN}}_{\vec{\alpha}}(\rho_{\text{in}})$. Then, given $\rho_{\text {out}}$, one measures the expectation value of a task-specific Hermitian operator. 

In our implementation,  we employ a QCNN to classify states belonging to different symmetry protected topological phases. Specifically, we consider the generalized cluster Hamiltonian
\begin{equation}\label{eq:Hamiltonian}
    H=\sum_{j=1}^n \left(Z_j -J_1 X_jX_{j+1}-J_2X_{j-1}Z_jX_{j+1} \right)\,,
\end{equation}
where $Z_i$ ($X_i$) denote the Pauli $z$ ($x$) operator acting on qubit $i$, and where $J_1$ and $J_2$ are tunable coupling coefficients. As proved in~\cite{verresen2017one}, and as schematically shown in Fig.~\ref{fig:training-versus-testing-qpr}, the ground-state phase diagram of the Hamiltonian of Eq.~\eqref{eq:Hamiltonian} has four different phases: symmetry-protected topological (I), ferromagnetic (II), anti-ferromagnetic (III), and trivial (IV). In Section \ref{sec:Methods}, we provide additional details regarding the classical simulation of the ground states of $H$.

By sampling parameters in the $(J_1,J_2)$ plane, we create a training set $\{(\ket{\psi_i},y_i)\}_{i=1}^N$ composed of ground states $\ket{\psi_i}$ of $H$ and their associated labels $y_i$. Here, the labels are  in the form of length-two bit strings, i.e., $y_i\in\{0,1\}^{2}$, where each possible bit string corresponds to a phase that $\ket{\psi_i}$ can belong to. The QCNN maps the $n$-qubit input state $\ket{\psi_i}$ to a $2$-qubit output state. We think of the information about the phase as being encoded into the output state by which of the $4$ computational basis effect operators is assigned the smallest probability. 
Namely, we define the loss function as $\ell(\vec{\alpha};\ket{\psi_i}, y_i )\coloneqq \bra{y_i}\mathcal{E}^{\mathrm{QCNN}}_{\vec{\alpha}}(\ketbra{\psi_i}{\psi_i})\ket{y_i}$. This leads to an empirical risk given by 
\begin{equation} \label{eq:empirical_risk_QCNN}
    \hat{R}_S(\vec{\alpha})  
    = \frac{1}{N} \sum_{i=1}^N \bra{y_i}\mathcal{E}^{\mathrm{QCNN}}_{\vec{\alpha}}(\ketbra{\psi_i}{\psi_i})\ket{y_i} \, .
\end{equation}

In Fig.~\ref{fig:training-versus-testing-qpr}, we visualize the phase classification performance achieved by our QCNN, trained according to this loss function, while additionally taking the number of misclassified points into account. Moreover, we show how the true risk, or rather the test accuracy as proxy for it, correlates well with the achieved training accuracy, already for small training data sizes. 
This is in agreement with our theoretical predictions, discussed in more detail in Appendix \ref{Sct:application-quantum-phase-recognition}, which for QCNNs gives a generalization error bound polylogarithmic in the number of qubits. We note that Refs.~\cite{kottmann2021unsupervised, kottmann2021variational} observed similarly favorable training data requirements for a related task of learning phase diagrams.

\subsubsection{Unitary compiling}\label{SbSct:unitary-compiling}

Compiling is the task of transforming a high-level algorithm into a low-level code that can be implemented on a device. Unitary compiling is a paradigmatic task in the NISQ era where a target unitary is compiled into a gate sequence that complies with NISQ device limitations, e.g., hardware-imposed connectivity and shallow depth to mitigate errors. Unitary compiling is crucial to the quantum computing industry, as it is essentially always performed prior to running an algorithm on a NISQ device, and various companies have their own commercial compilers~\cite{cross2017open, smith2016apractical}. Hence, any ability to accelerate unitary compiling could have industrial impact.

\begin{figure*}[t]
    \centering
    \includegraphics[width=\linewidth]{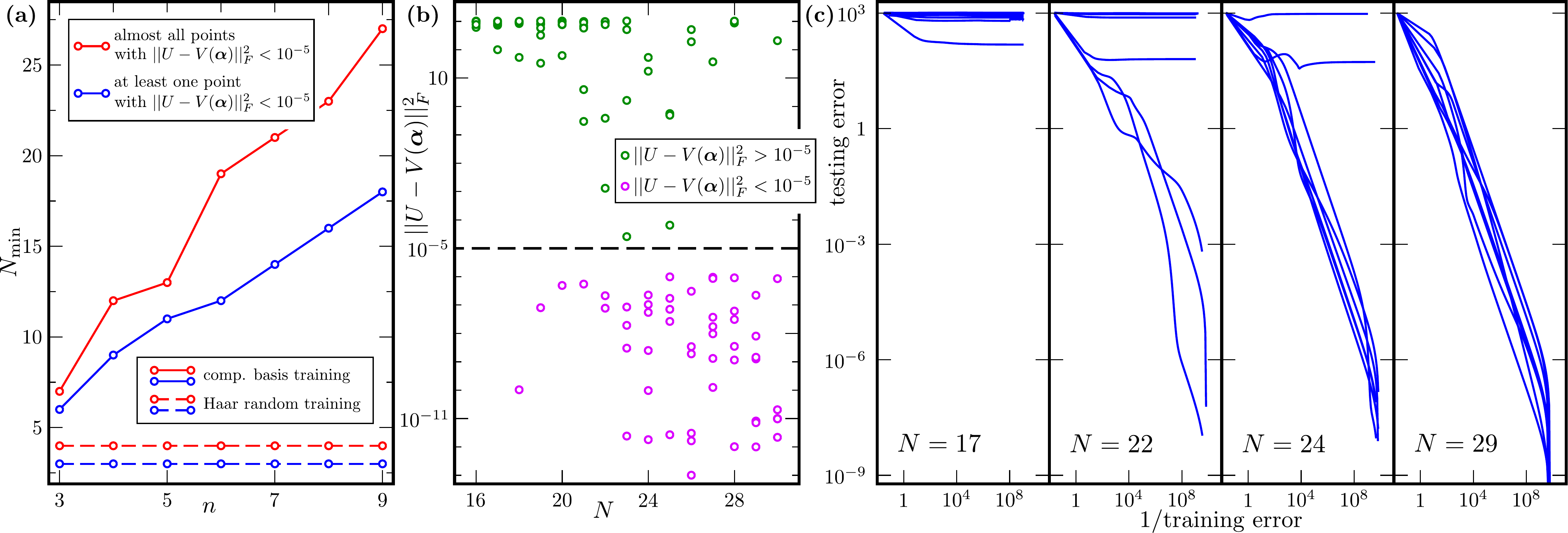}
    \caption{
    \textbf{Generalization performance of variational unitary compiling.} We employed a variable structure QMLM (as discussed near Theorem~\ref{thm:gate-sharing-QMLM-variable-structure}). 
    Panel (a) shows the dependence of $N_\mathrm{min}$, the minimum training data size for accurate compilation, on $n$, the number of qubits. Accurate compilation is defined as achieving $\lVert U-V(\vec{\alpha})\rVert^2_F <10^{-5}$ on $1$ out of $8$ (blue) or on $7$ out of $8$ (red) runs. For training data with random computational basis inputs (solid lines), $N_\mathrm{min}$ scales linearly in $n$. When training on examples with Haar random inputs (dashed lines), $N_\mathrm{min}$ is constant up to system size $n=9$. In Panel (b), for $n=9$ qubits, we plot the prediction error of successfully trained (training cost $<10^{-8}$) runs for $8$ training data sets with $N=16$ to $N=30$ random computational basis inputs. Panel (c) shows the dependence of the testing error on the reciprocal of the training error for different training data sizes, in the case of $9$ qubits. Here, the data consisted of random computational basis states and the corresponding outputs. As $N$ increases, small training error becomes a more reliable predictor for small testing error. 
    Each subplot shows $8$ different training runs, trained on different training data sets.
    }
    \label{fig:CompilingNumerics}
\end{figure*}

Here we consider the task of compiling the unitary $U$ of the $n$-qubit Quantum Fourier Transform (QFT)~\cite{nielsen2000quantum} into a short-depth parameterized quantum circuit $V(\vec{\alpha)}$. For $V(\vec{\alpha})$ we employ the VAns (Variable Ansatz) algorithm~\cite{bilkis2021semi, VAns}, which uses a machine learning protocol to iteratively grow a parameterized quantum circuit by placing and removing gates in a way that empirically leads to lower loss function values. Unlike traditional approaches that  train  just continuous parameters in a fixed structure circuit, VAns also trains discrete parameters, e.g., gate placement or type of gate, to explore the architecture hyperspace. 
In Appendix \ref{Sct:application-unitary-compiling}, we apply our theoretical results in this compiling scenario.

The training set for compilation is of the form $\{\ket{\psi_i},U\ket{\psi_i}\}_{i=1}^N$, consisting of input states $\ket{\psi_i}$ and output states obtained through the action of $U$. The $\ket{\psi_i}$ are drawn independently from an underlying data-generating distribution. In our numerics, we consider three such distributions: (1) random computational basis states, (2) random (non-orthogonal) low-entangled states, and (3) Haar random $n$-qubit states. Note that states in the first two distributions are easy to prepare on a quantum computer, whereas states from the last distribution become costly to prepare as $n$ grows. As the goal is to train $V(\vec{\alpha})$ to match the action of $U$ on the training set, we define the loss function as the squared trace distance between $U\ket{\psi_i}$ and $V(\vec{\alpha})\ket{\psi_i}$, i.e., $\ell(\alpha; \ket{\psi_i}, U\ket{\psi_i})\coloneqq\norm{U\ketbra{\psi_i}{\psi_i}U^\dagger - V(\vec{\alpha})\ketbra{\psi_i}{\psi_i}V(\vec{\alpha})^\dagger}_1^2$. This leads to the empirical risk
\small
\begin{equation} \label{eq:empirical_risk_unitary_comp}
    \hat{R}_S(\vec{\alpha})    
    =\frac{1}{N}\sum_{i=1}^N \norm{U\ketbra{\psi_i}{\psi_i}U^\dagger - V(\vec{\alpha})\ketbra{\psi_i}{\psi_i}V(\vec{\alpha})^\dagger}_1^2\,,
\end{equation}
\normalsize
where $\norm{\cdot}_1$ indicates the trace norm.

Fig.~\ref{fig:CompilingNumerics} shows our numerical results. As predicted by our analytical results, we can, with high success probability, accurately compile the QFT when training on a data set of size polynomial in the number of qubits. 
Our numerical investigation shows a linear scaling of the training requirements when training on random computational basis states. This better than the quadratic scaling implied by a direct application of our theory, which holds for any arbitrary data-generating distribution. Approximate implementations of QFT with a reduced number of gates~\cite{nam2020approximate}, combined with our results, could help to further study this apparent gap theoretically. When training on Haar random states, our numerics suggest that an even smaller number of training data points is sufficient for good generalization: Up to $n=9$ qubits, we generalize well from a constant number of training data points, independent of the system size.

Even more striking are our results when initializing close to the solution. In this case, as shown in Fig.~\ref{fig:CompilingNearSolution}, we find that two training data points suffice to obtain accurate generalization, which holds even up to a problem size of $40$ qubits. 
Our theoretical results in Theorem~\ref{thm:gate-sharing-QMLM-optimization} do predict reduced training requirements when initializing near the solution. Hence, the numerics are in agreement with the theory, although they paint an even more optimistic picture and suggest that further investigation is needed to understand why the training data requirements are so low. 
While the assumption of initialization near the solution is only viable assuming additional prior knowledge, it could be justified in certain scenarios. For example, if the unitaries to be compiled depend on a parameter, e.g., time, and if we have already compiled the unitary for one parameter setting, we might use this as initialization for unitaries with a similar parameter.

\section{Discussion}\label{sec:Discussion}

We conclude by discussing the impact of our work on specific applications, a comparison to prior work, the interpretation of our results from the perspective of quantum advantage, and some open questions.

\begin{figure}[t]
    \includegraphics[width=\columnwidth]{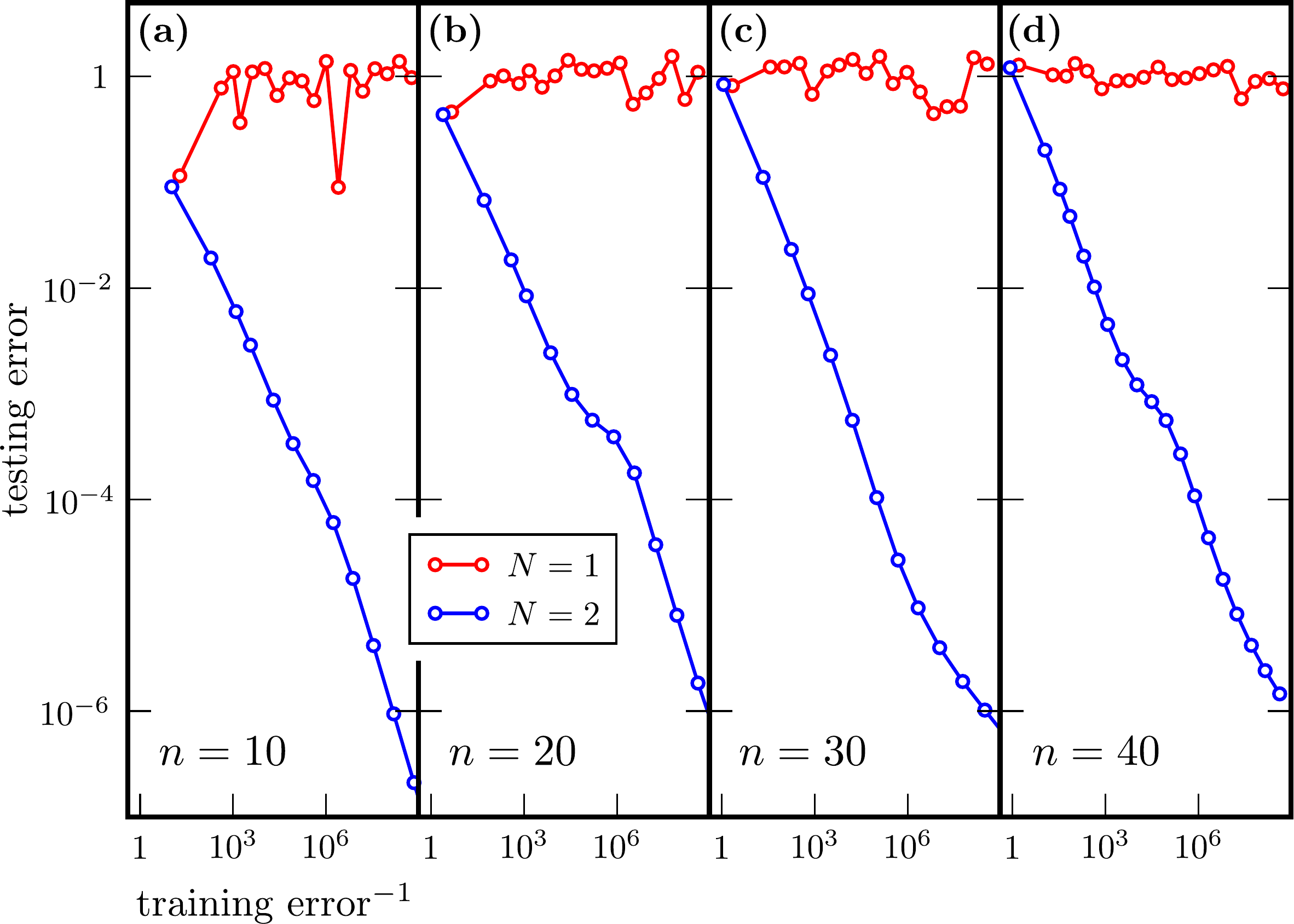}
    \caption{
    \textbf{Performance of variational unitary compiling when initializing near the solution.} Each panel shows the results of a single randomly initialized training run, where we used randomly drawn low-entangled states for training. The testing error on $20$ test states, which we allow to be more strongly entangled than the states used during training, is plotted versus the reciprocal of the training error for training data sizes $N=1,2$, for different system sizes $n$. 
    A training data set of size $N=1$ is not sufficient to guarantee good generalization: Even with decreasing training error, the testing error remains large. In contrast, assuming favorably initialized training, $N=2$ training data points suffice for good generalization, even for up to $n=40$ qubits. 
    }
    \label{fig:CompilingNearSolution}
\end{figure}

\subsection{Impact on specific applications}


Quantum phase classification is an exciting application of QML, to which Ref.~\cite{cong2019quantum} has successfully applied QCNNs. However, Ref.~\cite{cong2019quantum} only provided a heuristic explanation for the good generalization performance of QCNNs. Here, we have presented a rigorous theory that encompasses QCNNs and explains their performance, and we have confirmed it numerically for a fairly complicated phase diagram and a wide range of system sizes. In particular, our analysis allows us to go beyond the specific model of QCNNs and extract general principles for how to ensure good generalization. As generating training data for this problem asks an experimenter to prepare a variety of states from different phases of matter, which will require careful tuning of different parameters in the underlying Hamiltonian, good generalization guarantees for small training data sizes are crucial to allow for the implementation of phase classification through QML in actual physical experiments.


Several successful protocols for unitary compiling make use of training data \cite{qFactor,cincio2018learning,cincio2021machine}. However, prior work has relied on training data sets whose size scaled exponentially with the number of qubits. This scaling is problematic, both because it suggests a similarly bad scaling of the computational complexity of processing the data and because generating training data can be expensive in actual physical implementations. Our generalization bounds provide theoretical guarantees on the performance that unitary compiling with only polynomial-size training data can achieve, for the relevant case of efficiently implementable target unitaries. As we have numerically demonstrated in the case of the Quantum Fourier Transform, this significant reduction in training data size makes unitary compiling scalable beyond what previous approaches could achieve. Moreover, our results provide new insight into why the VAns algorithm~\cite{bilkis2021semi} is successful for unitary compiling. We believe that the QML perspective on unitary compiling advocated for in this work will lead to new and improved ansätze, which could scale to even larger systems.


Recent methods for variational dynamical simulation rely on quantum compiling to compile a Trotterized unitary into a structured ansatz with the form of a diagonalization~\cite{cirstoiu2020variational,commeau2020variational,gibbs2021long,geller2021experimental}. This technique allows for quantum simulations of times longer than an iterated Trotterization because parameters in the diagonalization may be changed by hand to provide longer-time simulations with a fixed depth circuit. We expect the quantum compiling results presented here to carry over to this application. This will allow these variational quantum simulation methods to use fewer training resources (either input-output pairs, or entangling auxiliary systems), yet still achieve good generalization and scalability.


Discovering quantum error correcting codes can be viewed as an optimization problem~\cite{fletcher2008channel,fletcher2008structured,kosut2008robust,kosut2009quantum,taghavi2010channel,johnson2017qvector,cong2019quantum}. Furthermore, it can be thought of as a machine learning problem, since computing the average fidelity of the code involves training data (e.g., chosen from a 2-design~\cite{johnson2017qvector}). Significant effort has been made to solve this problem on classical computers~\cite{fletcher2008channel,fletcher2008structured,kosut2008robust,kosut2009quantum,taghavi2010channel}. Such approaches can benefit from our generalization bounds, potentially leading to faster classical discovery of quantum codes. More recently, it was proposed to use near-term quantum computers to find such codes~\cite{johnson2017qvector,cong2019quantum}. Again our bounds imply good generalization performance with small training data for this application, especially for QCNNs~\cite{cong2019quantum}, due to their logarithmic number of parameters.


Finally, autoencoders and generative adversarial networks (GANs) have recently been generalized to the quantum setting~\cite{romero2017quantum,lloyd2018quantumgenerative,dallaire2018quantum,romero2021variational}. Both employ training data, and hence our generalization bounds provide quantitative guidance for how much training data to employ in these applications. Moreover, our results can provide guidance for ansatz design. While there is no standard ansatz yet for quantum autoencoders or quantum GANs, ansätze with a minimal number of parameters will likely lead to the best generalization performance.


\subsection{Related work on generalization}

Some prior works have studied the generalization capabilities of quantum models, among them the classical learning-theoretic approaches of \cite{caro2020pseudo, bu2021onthestatistical, bu2021effects, bu2021rademacher, gyurik2021structural, caro2021encodingdependent, chen2021expressibility, popescu2021learning, cai2022sample}; the more geometric perspective of \cite{abbas2020power, huang2021power}; and the information-theoretic technique of \cite{huang2021information, banchi2021generalization}.
Independently of this work, Ref.~\cite{du2021efficient} also investigated covering numbers in QMLMs. However our bounds are stronger, significantly more general, and broader in scope. We give a detailed comparison of our results to related work in Appendix~\ref{sec:related-work}.

\subsection{Quantum advantage and future outlook}

Our results do not prove a quantum advantage of quantum over classical machine learning. However, generalization bounds for QMLMs are necessary to understand their potential for quantum advantage. Namely, QMLMs can outperform classical methods, assuming both achieve small training error, only in scenarios in which QMLMs generalize well, but classical ML methods do not. We therefore consider our results a guide in the search for quantum advantage of QML: We need to identify a task in which QMLMs with few trainable gates achieve small training error, but classical models need substantially higher model complexity to achieve the same goal. Then, our bounds guarantee that the QMLM performs well also on unseen data, but we expect the classical model to generalize poorly due to the high model complexity.


We conclude with some open questions. For QMLMs with exponentially many independently trainable gates, our generalization error bounds scale exponentially with~$n$, and hence we do not make non-trivial claims about this regime. However, this does not yet imply that exponential-size QMLMs have bad generalization behavior. Whether and under which circumstances this is indeed the case is an interesting open question (e.g., see~\cite{huang2021power, banchi2021generalization} for some initial results). 
More generally, one can ask: Under what circumstances will a QMLM, even one of polynomial size, outperform our general bound. 
For example, if we have further prior knowledge about the loss, arising from specific target applications, it might be possible to use this information to tighten our generalization bounds. Moreover, as our generalization bounds are valid for arbitrary data-generating distributions, they may be overly pessimistic for favorable distributions. 
Concretely, in our numerical experiments for unitary compiling, highly entangled states were more favorable than especially efficiently preparable states from the perspective of generalization.
It may thus be interesting to investigate distribution-specific tightenings of our results. Finally, it may be fruitful to combine the generalization bounds for QMLMs studied in this work and the effect of data encodings in~\cite{caro2021encodingdependent} to yield a better picture on generalization in quantum machine learning.

\section{Methods}\label{sec:Methods}


This section gives an overview over our techniques. First, we outline the proof strategy that leads to the different generalization bounds stated above. Second, we present more details about our numerical investigations.

\subsection{Analytical methods}

An established approach to generalization bounds in classical statistical learning theory is to bound a complexity measure for the class under consideration. Metric entropies, i.e., logarithms of covering numbers, quantify complexity in exactly the way needed for generalization bounds, as one can show using the chaining technique from the theory of random processes~\cite{mohri2018foundations, vershynin2018highdimensional}. Therefore, a high level view of our proof strategy is: We establish novel metric entropy bounds for QMLMs and then combine these with known generalization results from classical learning theory. The strongest form of our generalization bounds is the following.

\begin{theorem}[Mother theorem]\label{thm:mother-main-text}
Consider a QMLM with an arbitrary number of parameterized local quantum channels, where for each $T > 0$, we have $G_T$ different QMLM architectures with $T$ trainable local gates. Suppose that after optimizing on the training data, the QMLM has $T$ independently parameterized local quantum channels, where the $t^{\textrm{th}}$ channel is reused at most $M$ times and is changed by $\Delta_t$ during the optimization. Without loss of generality, assume $\Delta_1 \geq \ldots \geq \Delta_T$. Then with high probability over input training data of size $N$, we have
\small
\begin{equation}
    \operatorname{gen}(\vec{\alpha}^\ast)
    \in \mathcal{O}\left( \min\limits_{K=0,\ldots,T}f(K) + \sqrt{\frac{\log(G_T)}{N}} \right),
\end{equation}
\normalsize
where $f(K)\coloneqq \sqrt{\frac{K \log (MT)}{N}} + \sum\limits_{k=K+1}^T M\Delta_k$ .
\end{theorem}

We give a detailed proof in Appendix \ref{Sct:analytical-results-details}. 
There, we also describe a variant in case the loss function cannot be evaluated exactly, but only estimated statistically. 
Here, we present only a sketch of how to prove Theorem \ref{thm:mother-main-text}. 

Before the proof sketch, however, we discuss how Theorem \ref{thm:mother-main-text} relates to the generalization bounds stated above. In particular, we demonstrate how to obtain Theorems \ref{thm:simpleQMLM}, \ref{thm:gate-sharing-QMLM}, \ref{thm:gate-sharing-QMLM-optimization}, and \ref{thm:gate-sharing-QMLM-variable-structure} as special cases of Theorem \ref{thm:mother-main-text}.

In the scenario of Theorem \ref{thm:simpleQMLM}, the QMLM architecture is fixed in advance, each trainable map is only used once, and we do not take properties of the optimization procedure into account. In the language of Theorem \ref{thm:mother-main-text}, this means: There exists a single $T>0$ with $G_T=1$ and we have $G_{\tilde{T}}=0$ for all $\tilde{T}\neq T$. Also, $M=1$. And instead of taking the minimum over $K=1,\ldots,T$, we consider the bound for $K=T$. Plugging this into the generalization bound of Theorem \ref{thm:mother-main-text}, we recover Theorem \ref{thm:simpleQMLM}.

Similarly, Theorem \ref{thm:mother-main-text} implies Theorems \ref{thm:gate-sharing-QMLM}, \ref{thm:gate-sharing-QMLM-optimization}, and \ref{thm:gate-sharing-QMLM-variable-structure}. Namely, if we take $G_T=1$ and $G_{\tilde{T}}=0$ for all $\tilde{T}\neq T$, and evaluate the bound for $K=T$, we recover Theorem~\ref{thm:gate-sharing-QMLM}. Choosing $G_T=1$ and $G_{\tilde{T}}=0$ for all $\tilde{T}\neq T$, the bound of Theorem~\ref{thm:mother-main-text} becomes that of Theorem~\ref{thm:gate-sharing-QMLM-optimization}. Finally, we can obtain Theorem~\ref{thm:gate-sharing-QMLM-variable-structure} by bounding the minimum in Theorem~\ref{thm:mother-main-text} in terms of the expression evaluated at $K=T$.

Now that we have established that Theorem \ref{thm:mother-main-text} indeed implies generalization bounds for all the different scenarios depicted in Fig.~\ref{fig:QMLM-settings}, we outline its proof. The first central ingredient to our reasoning are metric entropy bounds for the class of all $n$-qubit CPTP maps that a QMLM as described in Theorem \ref{thm:mother-main-text} can implement, where the distance between such maps is measured in terms of the diamond norm. Note: The trivial metric entropy bound obtained by considering this class of maps as compact subset of an Euclidean space of dimension exponential in $n$ is not sufficient for our purposes since it scales exponentially in $n$. Instead, we exploit the layer structure of QMLMs to obtain a better bound. More precisely, we show: If we fix a QMLM architecture with $T$ trainable $2$-qubit maps and a number of maps $0\leq K\leq T$, and we assume (data-dependent) optimization distances $\Delta_1\geq\ldots\geq \Delta_T$, then it suffices to take $(\nicefrac{\varepsilon}{KM})$-covering nets for each of the sets of admissible $2$-qubit CPTP maps for the first $K$ trainable maps to obtain a $(\varepsilon + \sum_{k=K+1}^T M\Delta_k)$-covering net for the whole QMLM. The cardinality of a covering net built in this way, crucially, is independent of $n$, but depends instead on $K$, $M$, and $T$. In detail, its logarithm can effectively be bounded as $\in \mathcal{O}\left(K\log\left(\nicefrac{MT}{\varepsilon}\right)\right)$. This argument directly extends from the $2$-local to the $\kappa$-local case, as we describe in Appendix~\ref{SbSct:CoveringNumberBounds}, Remark~\ref{Rmk:k-local}.

Now we employ the second core ingredient of our proof strategy. Namely, we combine a known upper bound on the generalization error in terms of the expected supremum of a certain random process with the so-called chaining technique. This leads to a generalization error bound in terms of a metric entropy integral. As we need a non-standard version of this bound, we provide a complete derivation for this strengthened form. This then tells us that, for each fixed $T$, $M$, $K$, and $\Delta_1\geq\ldots\geq \Delta_T$, using the covering net constructed above, we can bound the generalization error as $\operatorname{gen}(\vec{\alpha}^\ast)\in \mathcal{O}\left( \sqrt{\nicefrac{K\log(MT)}{N}} + \sum_{k=K+1}^T M\Delta_k \right)$, with high probability.

The last step of the proof consists of two applications of the union bound. The first instance is a union bound over the possible values of $K$. This leads to a generalization error bound in which we minimize over $K=0,\ldots,T$. So far, however, the bound still applies only to any QMLM with fixed architecture. We extend it to variable QMLM architectures by taking a second union bound over all admissible numbers of trainable gates $T$ and the corresponding $G_T$ architectures. As this is, in general, a union bound over countably many events, we have to ensure that the corresponding failure probabilities are summable. Thus, we invoke our fixed-architecture generalization error bound for a success probability that is proportional to $(G_T T^2)^{-1}$. In that way, the union bound over all possible architectures yields the logarithmic dependence on $G_T$ in the final bound and completes the proof of Theorem \ref{thm:mother-main-text}.

\subsection{Numerical methods}

This section discusses numerical methods used throughout the paper.
The subsections give details on computational techniques applied to phase classification of the cluster Hamiltonian in Eq.~\eqref{eq:Hamiltonian} and Quantum Fourier Transform compilation.

\subsubsection{Phase classification}

The training and testing sets consist of ground states $\ket{\psi_i}$ of the cluster Hamiltonian in Eq.~\eqref{eq:Hamiltonian}, computed for different coupling strengths $(J_1,J_2)$. The states $\ket{\psi_i}$ were obtained with the translation invariant Density Matrix Renormalization Group~\cite{white1992density}. The states in the training set (represented by blue crosses in Fig.~\ref{fig:training-versus-testing-qpr}(a)) are chosen to be away from phase transition lines, so accurate description of the ground states is already achieved at small bond dimension $\chi$. That value determines the cost of further computation involving the states $\ket{\psi_i}$ and we keep it small for efficient simulation.

We use Matrix Product State techniques~\cite{orus2014practical} to compute and optimize the empirical risk in Eq.~\eqref{eq:empirical_risk_QCNN}. The main part of that calculation is the simulation of the action of the QCNN $\mathcal{E}^{\mathrm{QCNN}}_{\vec{\alpha}}$ on a given ground state $\ket{\psi_i}$. The map $\mathcal{E}^{\mathrm{QCNN}}_{\vec{\alpha}}$ consists of alternating convolutional and pooling layers. In our implementation the layers are translationally invariant and are represented by parameterized two-qubit gates. The action of a convolutional layer on an MPS amounts to updating two nearest neighbor MPS tensors in a way similar to the time-evolving block decimation algorithm~\cite{vidal2007classical}. The pooling layer is simulated in two steps. First, we simulate the action of all two-qubit gates on an MPS. This is analogous to the action of a convolutional layer, but performed on a different pair of nearest neighbor MPS tensors. This step is followed by a measurement of half of the qubits. We use the fact that the MPS can be written as a unitary tensor network and hence allows for perfect sampling techniques~\cite{ferris2012perfect}. The measurement step results in a reduction of the system size by a factor of two.

We repeat the application of convolutional and pooling layers using the MPS as described above until the system size becomes small enough to allow for an exact description. A~few final layers are simulated in a standard way and the empirical risk is given by a two-qubit measurement according to the label $y_i$, as in Eq.~\eqref{eq:empirical_risk_QCNN}. The empirical risk is optimized with the Simultaneous Perturbation Stochastic Approximation algorithm~\cite{spall1998overview}. We grow the number of shots used in pooling layer measurements as the empirical risk is minimized. This results in a shot-frugal optimization~\cite{kubler2020adaptive}, as one can control the accuracy of the gradient based on the current optimization landscape. 

\subsubsection{Unitary compiling} 

In Section~\ref{sec:numerics}, we show that the task of unitary compilation can be translated into minimization of the empirical risk $\hat{R}_S(\vec{\alpha})$ defined in Eq.~\eqref{eq:empirical_risk_unitary_comp}. Here, $\vec{\alpha} = (\vec{\theta},\vec{k})$ denotes a set of parameters that specifies a trainable unitary $V(\vec{\alpha})$. The optimization is performed in the space of all shallow circuits. It has discrete and continuous components. The discrete parameters $\vec{k}$ control the circuit layout, that is, the placement of all gates used in the circuit. 
Those gates are described by the continuous parameters $\vec{\theta}$. The optimization $\min_{\vec{\alpha}} \hat{R}_S(\vec{\alpha})$ is performed with the recently introduced VAns algorithm~\cite{bilkis2021semi, VAns}. The unitary $V(\vec{\alpha})$ is initialized with a circuit that consists of a few randomly placed gates. In subsequent iterations, VAns modifies the structure parameter $\vec{k}$ according to certain rules that involve randomly placing a resolution of the identity and removing gates that do not significantly contribute to the minimization of the empirical risk $\hat{R}_S(\vec{\alpha})$. 
A modified~\cite{qfactornote} qFactor algorithm~\cite{qFactor} is used to optimize over continuous parameters $\vec{\theta}$ for fixed $\vec{k}$. This optimization is performed after each update to the structure parameter $\vec{k}$. In subsequent iterations, VAns makes a probabilistic decision whether the new set of parameters $\vec{\alpha'}$ is kept or rejected. This decision is based on the change in empirical risk $\hat{R}_S(\vec{\alpha'}) - \hat{R}_S(\vec{\alpha})$, an artificial temperature $T$, and a factor $\Lambda$ that sets the penalty for growing the circuit too quickly. To that end, we employ a simulated annealing technique, gradually decreasing $T$ and $\Lambda$, and repeat the iterations described above until $\hat{R}_S(\vec{\alpha})$ reaches a sufficiently small value. 

Let us now discuss the methods used to optimize the empirical risk when $V(\vec{\alpha})$ is initialized close to the solution. Here, we start with a textbook circuit for performing the QFT and modify it in the following way. First, the circuit is rewritten such that it consists of two-qubit gates only. Next, each two-qubit gate $u$ is replaced with $u' = u e^{i \delta h}$, where $h$ is a random Hermitian matrix and $\delta$ is chosen such that $|| u - u' || = \epsilon$ for an initially specified $\epsilon$. The results presented in Section~\ref{sec:numerics} are obtained with $\epsilon = 0.1$. The perturbation considered here does not affect the circuit layout and hence the optimization over continuous parameters $\vec{\theta}$ is sufficient to minimize the empirical risk $\hat{R}_S(\vec{\alpha})$. We use qFactor to perform that optimization.

The input states $\ket{\psi_i}$ in the training set $\{\ket{\psi_i},U_\mathrm{QFT} \ket{\psi_i}\}_{i=1}^N$ are random MPSs of bond dimension $\chi=2$. The QFT is efficiently simulable~\cite{browne2007efficient} for such input states, which means that $U_\mathrm{QFT} \ket{\psi_i}$ admits an efficient MPS description. Indeed, we find that a bond dimension $\chi < 20$ is sufficient to accurately describe $U_\mathrm{QFT} \ket{\psi_i}$. In summary, the use of MPS techniques allows us to construct the training set efficiently. Note that the states $V(\vec{\alpha})\ket{\psi_i}$ are in general more entangled than $U_\mathrm{QFT} \ket{\psi_i}$, especially at the beginning of the optimization. Because of that, we truncate the evolved MPS during the optimization. We find that a maximal allowed bond dimension of $100$ is large enough to perform stable, successful minimization of the empirical risk with qFactor. The testing is performed  with 20 randomly chosen initial states. We test with bond dimension $\chi = 10$ MPSs, so the testing is done with more strongly entangled states than the training. Additionally, for system sizes up to $16$ qubits, we verify that the trained unitary $V$ is close (in the trace norm) to $U_\mathrm{QFT}$, when training is performed with at least two states. 







\bibliography{quantum.bib}

\begin{thebibliography}{123}
\bibitem [{qfa()}]{qfactornote}%
  \BibitemOpen
  \href@noop {} {}\bibinfo {note} {The original qFactor
  algorithm~\cite{qFactor} can be modified to work with a set of pairs of
  states as opposed to a target unitary.}\BibitemShut {Stop}%
\end{thebibliography}

\section*{Acknowledgements}

MCC was supported by the TopMath Graduate Center of the TUM Graduate School at the Technical University of Munich, Germany, the TopMath Program at the Elite Network of Bavaria, and by a doctoral scholarship of the German Academic Scholarship Foundation (Studienstiftung des deutschen Volkes). HH is supported by the J. Yang \& Family Foundation. MC, KS, and LC were initially supported by the LANL LDRD program under project number 20190065DR. MC was also supported by the Center for Nonlinear Studies at LANL. KS acknowledges support the Department of Defense. ATS and PJC acknowledge initial support from the LANL ASC Beyond Moore's Law project. PJC and LC were also supported by the U.S. Department of Energy (DOE), Office of Science, Office of Advanced Scientific Computing Research, under the Accelerated Research in Quantum Computing (ARQC) program. LC also acknowledges support from LANL LDRD program under project number 20200022DR. ATS was additionally supported by the U.S. Department of Energy, Office of Science, National Quantum Information Science Research Centers, Quantum Science Center. This research used resources provided by the Los Alamos National Laboratory Institutional Computing Program, which is supported by the U.S. Department of Energy National Nuclear Security Administration under Contract No. 89233218CNA000001.

\vspace{0.15in}

\section*{Author contributions}

The project was conceived by M.C.C., H.-Y.H., and P.J.C.
The manuscript was written by M.C.C., H.-Y.H., M.C., K.S., A.S., L.C., and P.J.C.
Theoretical results were proved by M.C.C.~and H.-Y.H.
The applications were conceived by M.C.C., H.-Y.H., M.C., K.S., A.S., L.C., and P.J.C.
Numerical implementations were performed by M.C. and L.C.


\onecolumngrid

\appendix
\appendixtitleon
\renewcommand{\thesubsection}{\arabic{subsection}.}
\renewcommand{\thesubsubsection}{\arabic{subsubsection}.}

\setcounter{theorem}{0}
\setcounter{corollary}{0}
\setcounter{lemma}{0}
\setcounter{proposition}{0}
\setcounter{definition}{0}
\setcounter{remark}{0}
\setcounter{example}{0}

\begin{appendices}

\newpage

\setcounter{page}{1}
\renewcommand\thefigure{\thesection\arabic{figure}}
\setcounter{figure}{0} 

\begin{center}
\large{ Supplementary Information for \\ ``Generalization in quantum machine learning from few training data''
}
\end{center}

\section{Related Work}\label{sec:related-work}

\subsection{Related Work on Generalization Bounds for Quantum Machine Learning}

In statistical learning theory, a variety of techniques for obtaining generalization bounds are known. The classical approach is based on complexity measures for the class of functions describing the machine learning model (MLM) under consideration. Among these complexity measures, the VC-dimension \cite{vapnik1971uniform}, the pseudo-dimension \cite{pollard1984convergence}, the Rademacher complexity \cite{gine1984some, bartlett2002rademacher}, and covering numbers (and the related metric entropies)~\cite{dudley1999uniform} are particularly well known. More recently, different approaches that take properties of the learning algorithm into account have been investigated, such as stability (introduced by \cite{bousquet2002stability}), differential privacy (going back to \cite{dwork2006calibrating}), sample compression (due to \cite{littlestone1986relating}), and the PAC-Bayesian framework (described in \cite{mcallester1999some}).
The theory of generalization for quantum machine learning (QML) is less developed. Nevertheless, there has been some prior work, of which we now give an overview.

Ref.~\cite{caro2020pseudo} proved bounds on the pseudo-dimension of quantum circuits in which the local unitaries can be varied. In particular, these pseudo-dimension bounds imply generalization bounds for learning polynomial-depth unitary quantum circuits from data. While the data encoding considered in Ref.~\cite{caro2020pseudo} was a simple product encoding, this can be understood as an early investigation of the generalization behavior of variational quantum circuits. In particular, the techniques of Ref.~\cite{caro2020pseudo} can also be applied for more general quantum data encodings. Ref.~\cite{popescu2021learning} has recently extended the generalization guarantees of Ref.~\cite{caro2020pseudo} from the realizable to the agnostic setting, using covering number arguments. We note that all our generalization bounds apply to the agnostic setting, but for more general QMLMs than considered in~\cite{popescu2021learning}.

Ref.~\cite{abbas2020power} suggested the so-called effective dimension, derived from the (empirical) Fisher information matrix, as a complexity measure for the parameter space of a QMLM. In particular, Ref.~\cite{abbas2020power} showed how to derive generalization bounds from bounds on the effective dimension and investigated this complexity measure numerically for different QMLMs. Contrary to the conclusions drawn in Ref.~\cite{abbas2020power}, the recommendations for QMLMs which we deduce from our generalization bounds are not unequivocably in favour of higher expressivity. Instead, we emphasize the ability to fit training data and the ability to generalize have to be balanced carefully. See Section \ref{sec:Discussion} for a discussion of implications of our results for a potential quantum advantage in QML.

Refs.~\cite{bu2021onthestatistical, bu2021effects, bu2021rademacher} studied the Rademacher complexity of parameterized quantum circuits and thus QMLMs. They proved bounds on this complexity measure that depend on the size and depth of the circuit as well as on a measure of magic in the circuit. Ref.~\cite{bu2021effects} provides a resource-theoretic perspective on the Rademacher complexity of a quantum circuit and Ref.~\cite{bu2021rademacher} investigated the effects of noise in the circuit.

We mention one more related work that approaches generalization in QML via complexity measures. Ref.~\cite{du2021efficient} provides bounds on covering numbers of QMLMs and, using these, deduces generalization bounds.
We have developed our approach independently from Ref.~\cite{du2021efficient} and have obtained both stronger and more general results. In particular, Theorem \ref{thm:prediction-error-bound-qnn} shows that the generalization error bound scales as $\sqrt{T / N}$, where $T$ is the number of trainable gates and $N$ is the number of training data, compared to $T / \sqrt{N}$ in Theorem $2$ in the first version of Ref.~\cite{du2021efficient}. In addition, and in contrast to Ref.~\cite{du2021efficient}, we also consider the practically relevant scenarios of CPTP (not just unitary) QMLMs, of multiple uses/copies of trainable maps, and of variable QMLM structure. Moreover, our optimization-dependent generalization bounds for QMLMs are the first bounds of this kind for QML and showcase a new way of using covering numbers.

Ref.~\cite{banchi2021generalization} has proposed an information-theoretic strategy towards studying the approximation and generalization capabilities of QMLMs. In particular, Ref.~\cite{banchi2021generalization} demonstrates how the approximation and generalization errors of a QMLM can be bounded in terms of (Rényi) mutual informations between the quantum embedding achieved by the QMLM (before the final measurement) and the label or instance marginals of the data, respectively. 

Ref.~\cite{huang2021power} considered a class of QMLM (quantum kernels) that is equivalent to training arbitrarily deep quantum circuits. The work also established generalization error bounds to study when quantum machine learning models would predict more accurately than classical machine learning models. Ref.~\cite{huang2021power} showed that even if we are training an arbitrarily deep quantum circuits, the generalization performance can still be good if a certain geometric criterion is met.
Ref.~\cite{wang2021towards} provided generalization error bounds for quantum kernels in noisy quantum circuits, and Ref.~\cite{kubler2021inductive} studied the generalization performance of quantum kernels for some embeddings.
Our work considers finite size quantum circuits and the resulting generalization error bounds are very different.

Even more recently, Ref.~\cite{gyurik2021structural} has proved bounds on the VC-dimension and the fat-shattering dimension of a QMLM, by viewing the QMLM in terms of a parameterized measurement performed on the quantum data encoding. These complexity bounds lead to generalization bounds for QMLMs that depend on spectral properties (more precisely, rank or Frobenius norm) of the parameterized measurement.

Shortly thereafter, Ref.~\cite{caro2021encodingdependent} studied the generalization capabilities of QMLMs with a focus on the strategy used to encode classical data into the quantum circuit. In particular, they considered data encodings via Hamiltonian evolutions, where data re-uploading is allowed. For corresponding QMLMs, Ref.~\cite{caro2021encodingdependent} established generalization bounds that depend explicitly on properties of the Hamiltonians used for data-encoding. These results are complementary to our work: The generalization guarantees of Ref.~\cite{caro2021encodingdependent} depend only on the encoding strategy used in the QMLM, whereas our results are in formulated in terms of properties of the trainable part of the QMLM only.

Ref.~\cite{chen2021expressibility} investigated the expressibility and the generalization behavior of specific QMLMs. By combining light cone arguments with insights into how a specific data-encoding leads to effective dimensionality limitations (see also \cite{caro2021encodingdependent}), Ref.~\cite{chen2021expressibility} obtained VC-dimension bounds for the hardware efficient ansatz. These bounds depend on the number of qubits and on the number of trainable layers. Ref.~\cite{chen2021expressibility} interpreted the overall limitation on the VC-dimension imposed by the data-encoding as an automatic regularization, which is helpful in avoiding overfitting.

Lastly, Ref.~\cite{cai2022sample} investigated a problem of learning parametrized unitary quantum circuits from training data consisting of pairs of input and corresponding output states. They established generalization bounds, and thus sample complexity bounds, by first identifying a universal family of variational quantum circuit architecture, then considering a finite discretization of this family, and finally applying a standard generalization bound for finite hypothesis classes. We note that the generalization guarantee obtained from Theorem~\ref{thm:prediction-error-bound-qnn} is tighter than that obtained in Ref.~\cite{cai2022sample}: For a variational $n$-qubit QMLM with at most $n^c$ gates,~\cite[Theorem 2]{cai2022sample} implies that a sample complexity of $\tilde{\mathcal{O}}(n^{c+1})$ suffices for good generalization, whereas Theorem~\ref{thm:prediction-error-bound-qnn} tells us that already $\tilde{\mathcal{O}}(n^{c})$ samples suffice. Additionally, our generalization guarantees apply for more general architectures than those considered in~\cite{cai2022sample}.

\subsection{Related Work on Quantum Phase Recognition}

Recognizing quantum phases of matter is an important question in physics. Recently, many works have considered training machine learning models to classify quantum phases.
The works include the use of quantum neural networks \cite{cong2019quantum} and classical machine learning models \cite{neupeurt2017, leiwang2016, evert2017nature, carrasquilla2017nature}.
Most of the existing works do not come with rigorous guarantees. Thus, it is not clear whether the respective machine learning models will predict well after training.
Our work shows that when a quantum neural network, such as a QCNN \cite{cong2019quantum}, can perform well on a training set with a moderate amount of examples, the quantum neural network will also predict well on new data. This is particularly prominent in QCNNs, for which the required training data size scales at most polylogarithmically in the system size.
However, in order for quantum neural networks to achieve a small training error, one still needs to address various challenges, such as barren plateau in the training landscape \cite{mcclean2018barren, cerezo2020cost}.

Recently, \cite{huang2021provably} has proposed provably efficient classical machine learning models that can classify a wide range of quantum phases of matter, including symmetry-broken phases, topological phases, and symmetry-protected topological phases. These classical machine learning models are efficient in both computational time and the required training data \cite{huang2021provably}.
Furthermore, the numerical experiments of \cite{huang2021provably} have shown that no labels of the different phases are needed to train the classical machine learning models. The classical algorithm can automatically uncover the quantum phases of matter in an unsupervised learning procedure.

It remains to be seen if QMLMs, such as QCNNs \cite{cong2019quantum}, can improve upon classical machine learning models in classifying quantum phases of matter.
For example, \cite{huang2021provably} shows that the prediction performance of classical machine learning models sometimes degrades when the correlation length in the ground state wave function is high.
It would be interesting to understand whether QMLMs can still work well when classical machine learning models fail.

\subsection{Related Work on Quantum Compiling}


Compiling of quantum circuits is a broad field with many distinct approaches. For example, temporal planning~\cite{venturelli2018compiling,booth2018comparing}, reinforcement learning~\cite{mckiernan2019automated}, and supervised learning~\cite{cincio2018learning,cincio2021machine} are three alternative approaches that have been applied to quantum compiling. Moreover, while classical methods for quantum compiling are the most common, it has also been proposed to do quantum-assisted quantum compiling where a quantum computer is involved in the compiling process~\cite{khatri2019quantum,sharma2019noise,heya2018variational,jones2022robustquantum}.


While not all methods employ training data, it is worth noting that some state-of-the-art methods are in fact based on training data~\cite{qFactor,cincio2018learning,cincio2021machine}. It is also worth remarking that noise-aware quantum compiling methods can involve training data~\cite{cincio2021machine}. For these methods, it has largely been assumed that one would need an amount of training data that grows exponentially with the number of qubits. Naturally, this exponential scaling places a cutoff on the size of unitaries that one can compile. However, with our results in mind (allowing for only polynomial-sized training sets), this cutoff can be significantly extended to larger unitary sizes.


For quantum compiling, the benefit of our work is two-fold, in that both classical methods and quantum methods can potentially be sped-up. Classical methods for quantum compiling are currently being used in the quantum computing industry to enhance the performance of cloud-based quantum computing (e.g., by companies such as Rigetti and IBM). Therefore, speeding up classical methods for quantum compiling can potentially have a direct impact on cloud-based quantum computing. Both standard compiling and noise-aware compiling are important for industrial near-term quantum computing, and our work impacts both of these approaches.


In addition, quantum-assisted methods for quantum compiling can also reduce their resource costs based on our results. Variational quantum algorithms for quantum compiling have been introduced~\cite{khatri2019quantum, sharma2019noise,heya2018variational,jones2022robustquantum}. Specifically, Refs.~\cite{khatri2019quantum, sharma2019noise,sharma2020reformulation} discussed methods that employ an entangled state on $2n$ qubits to compile an $n$-qubit unitary. Due to our work, this entangled state can apparently be reduced in size, namely only needing a Schmidt rank that is polynomial in $n$ (instead of a Schmidt rank that is exponential in $n$). Ref.~\cite{heya2018variational} proposed a slightly different approach that did not involve an auxiliary system, but simply used multiple training data points. Our work shows that the amount of training data here does not need to grow exponentially in $n$, making the approach in Ref.~\cite{heya2018variational} potentially scalable.


Finally, we note that variable ansatz methods (e.g., Ref.~\cite{bilkis2021semi, VAns}) for quantum compiling is a state-of-the-art approach that is employed, e.g., in Refs.~\cite{cincio2018learning,cincio2021machine}. As noted in the main text, our results are general enough to cover the variable ansatz case (where the structure of the circuit changes during the optimization). Hence we provide guidance for how much training data is needed for the variable ansatz case as well. 

\section{Auxiliary Lemmata}

Before presenting our results, we use this section to recall some well known auxiliary results that enter our proofs.

\subsection{Auxiliary Lemmata from Statistical Learning Theory}

We use two standard concentration inequalities. The first is due to Wassily Hoeffding.

\begin{lemma}[Hoeffding's Concentration Inequality \cite{hoeffding1963probability}]\label{lem:Hoeffding}
Let $X_1,\ldots,X_N$ be independent $\mathbb{R}$-valued random variables. Assume that, for every $1\leq i\leq N$, $X_i\in[a_i,b_i]$ almost surely, where $a_i,b_i\in\mathbb{R}$, $a_i\leq b_i$. Then, for every $\varepsilon>0$, 
\begin{align}
    \mathbb{P}\left[\sum\limits_{i=1}^N (X_i - \mathbb{E}[X_i])\geq\varepsilon\right]
    &\leq \exp\left( -\nicefrac{2\varepsilon^2}{\sum\limits_{i=1}^N (b_i-a_i)^2}\right),\\
    \mathbb{P}\left[\left\lvert\sum\limits_{i=1}^N (X_i - \mathbb{E}[X_i])\right\rvert\geq\varepsilon\right]
    &\leq 2\exp\left( -\nicefrac{2\varepsilon^2}{\sum\limits_{i=1}^N (b_i-a_i)^2}\right).
\end{align}
\end{lemma}

The second is the bounded differences inequality, originally due to Colin McDiarmid.

\begin{lemma}[McDiarmid's Concentration Inequality \cite{mcdiarmid1989onthemethod}]\label{lem:McDiarmid}
Let $X_1,\ldots,X_N$ be independent random variables, each with values in $\mathcal{Z}$. Let $\varphi:\mathcal{Z}^N\to\mathbb{R}$ be a measurable function s.t., whenever $z\in\mathcal{Z}^n$ and $z'\in\mathcal{Z}^n$ differ only in the $i^\mathrm{th}$ entry, then $\lvert\varphi(z)-\varphi(z')\rvert\leq b_i$. Then, for every $\varepsilon>0$, we have
\begin{equation}
    \mathbb{P}\left[\varphi(Z)-\mathbb{E}[\varphi(Z)]\geq\varepsilon\right]
    \leq \exp\left(-\nicefrac{2\varepsilon^2}{\sum\limits_{i=1}^Nb_i^2} \right).
\end{equation}
\end{lemma}

The third well known ingredient that we will employ in our reasoning without giving a proof is the following.

\begin{lemma}[Massart's Lemma \cite{massart2000some}]\label{lem:Massart}
Let $N\in\mathbb{N}$. Let $A\subset\mathbb{R}^N$ be a finite set contained in a Euclidean ball of radius $r>0$. Then
\begin{equation}
    \mathbb{E}_\sigma \left[\sup\limits_{a\in A}\frac{1}{N}\sum\limits_{i=1}^N \sigma_i a_i\right]
    \leq \frac{r\sqrt{2\log\lvert A\rvert}}{N},
\end{equation}
where the expectation is w.r.t.~i.i.d. Rademacher random variables $\sigma_1,\ldots,\sigma_N$.
\end{lemma}

\subsection{Auxiliary Lemmata from Quantum Information Theory}

From quantum information theory, we crucially make use of the following lemma.

\begin{lemma}[Subadditivity of diamond distance; see \cite{watrous2018thetheory}, Proposition $3.48$] \label{lem:subadditivity-diamond} For any completely positive and trace-preserving maps $\mathcal{A}, \mathcal{B}, \mathcal{C}, \mathcal{D}$, where $\mathcal{B}$ and $\mathcal{D}$ map from $n$-qubit to $m$-qubit systems and $\mathcal{A}$ and $\mathcal{C}$ map from $m$-qubit to $k$-qubit systems, we have the following inequality
\begin{equation}
\norm{\mathcal{A}\mathcal{B} - \mathcal{C}\mathcal{D}}_{\diamond} \leq \norm{\mathcal{A} - \mathcal{C}}_\diamond + \norm{\mathcal{B} - \mathcal{D}}_\diamond.
\end{equation}
\end{lemma}

Also, to translate between the spectral norm of unitaries and the diamond norm of the corresponding channels, we employ the following result.

\begin{lemma}[Spectral norm and diamond norm of unitary channels] \label{lem:diamond-to-operator}
Let $\mathcal{U}(\rho) = U \rho U^\dagger$ and $\mathcal{V}(\rho) =V \rho V^\dagger$ be unitary channels. Then,
$\tfrac{1}{2}\| \mathcal{U}(|\psi \rangle \! \langle \psi|) - \mathcal{V} (|\psi \rangle \! \langle \psi|) \|_1 \leq \| (U-V) | \psi \rangle \|_{\ell_2}$ for any pure state $|\psi \rangle$.
Therefore,
\begin{equation}
    \tfrac{1}{2} \| \mathcal{U} - \mathcal{V} \|_\diamond \leq \|U - V \|.
\end{equation}
\end{lemma}
\begin{proof}
The proof is adapted from \cite{chen2021concentration}.
Fix an input $|\psi \rangle$ and denote the output state vectors by $|u \rangle=U| \psi \rangle$ and $|v \rangle = V | \psi \rangle$, respectively. 
Normalization ensures that
these state vectors obey $|\langle u,v \rangle | \leq 1$, as well as $\| |u \rangle - |v \rangle \|_{\ell_2}=\sqrt{2(1-\mathrm{Re}(\langle u,v \rangle))}$.
Apply the Fuchs--van de Graaf relations~\cite{fuchs1999cryptographic} to convert the output trace distance into a (pure) output fidelity:
\begin{align}
\tfrac{1}{2} \||u \rangle \! \langle u| - |v \rangle \! \langle v| \|_1 
&= \sqrt{1-| \langle u,v \rangle|^2} \\
&= \sqrt{(1+| \langle u,v \rangle |) (1-| \langle u,v\rangle|)}\\
&\leq \sqrt{2(1-\mathrm{Re}(\langle u,v \rangle))}\\
&=\| |u \rangle - |v \rangle \|_{\ell_2}.
\end{align}
The diamond distance bound then is a direct consequence of this relation. Using the fact that stabilization is not necessary for computing the diamond distance of two unitary channels \cite{watrous2018thetheory}, we get
\begin{align}
\tfrac{1}{2} \| \mathcal{U} - \mathcal{V} \|_\diamond 
&= \max_{|\psi \rangle \! \langle \psi|} \tfrac{1}{2} \| \mathcal{U}(|\psi \rangle \! \langle \psi|) - \mathcal{V}(|\psi \rangle \! \langle \psi| ) \|_1\\
&\leq \max_{| \psi \rangle}\|(U-V)| \psi \rangle \|_{\ell_2}=\|U-V \| .
\end{align}
Here, we have also used the definition of the operator norm.
\end{proof}

\section{Analytical Results: Details and Proofs}\label{Sct:analytical-results-details}

We first introduce some standard notation.
Let $\mathcal{D}(\mathcal{H})$ denote the set of density operators (positive semi-definite with unit trace) acting on the Hilbert space $\mathcal{H}$. Let $\mathcal{L}(\mathcal{H})$ denote the space of square linear operators acting on $\mathcal{H}$. Let $\mathcal{L}(\mathcal{H}, \mathcal{H}^{'})$ denote the set of linear operators taking $\mathcal{H}$ to a Hilbert space $\mathcal{H}^{'}$. The trace norm of a linear operator $A\in \mathcal{L}(\mathcal{H}, \mathcal{H}^{'})$ is defined as $\Vert A \Vert_1 :=\Tr[|A|]$, where $|A|:=\sqrt{A^{\dagger}A}$. The trace distance between any two operators $A, B \in \mathcal{L}(\mathcal{H}, \mathcal{H}^{'})$ is $\Vert A - B \Vert_1$, and for two quantum  states $\rho, \sigma  \in \mathcal{D}(\mathcal{H})$ it is linearly related to the maximum success probability of distinguishing $\rho$ and $\sigma$ in a quantum hypothesis testing experiment. A linear map $\mathcal{N}_{A \to B}: \mathcal{L}(\mathcal{H}_A) \to \mathcal{L}(\mathcal{H}_B)$ is called a completely positive (CP) map if $(\mathcal{I}_R \otimes \mathcal{N}_{A\to B})(X_{RA})$ is positive semi-definite for all positive semi-definite $X_{RA} \in \mathcal{L}(\mathcal{H}_{RA})$, where $\mathcal{H}_{RA} = \mathcal{H}_R \otimes \mathcal{H}_A$ and the reference system $R$ can be of arbitrary size. Moreover, a linear map $\mathcal{N}_{A \to B}: \mathcal{L}(\mathcal{H}_A) \to \mathcal{L}(\mathcal{H}_B)$ is trace preserving (TP) if $\Tr(\mathcal{N}_{A\to B}(X_A)) = \Tr(X_A)$ for all $X_A \in \mathcal{L}(\mathcal{H}_A)$. A linear map $\mathcal{N}_{A \to B}$ is called a quantum channel if it is completely positive and trace preserving (CPTP). Let $\mathcal{N}_{A \to B}$ and $\mathcal{M}_{A \to B}$ denote quantum channels. Then the diamond distance between $\mathcal{N}_{A \to B}$ and $\mathcal{M}_{A \to B}$ is defined as
\begin{equation}\label{eq:diamond-distance}
\Vert \mathcal{N}_{A \to B} - \mathcal{M}_{A \to B} \Vert_{\diamond} := \sup_{\rho_{RA} \in \mathcal{D}(\mathcal{H}_{RA})} \Vert (\mathcal{I}_R \otimes \mathcal{N}_{A \to B})(\rho_{RA}) - (\mathcal{I}_R \otimes \mathcal{M}_{A \to B})(\rho_{RA})\Vert_1,
\end{equation}
where $\mathcal{I}_R$ is the identity map acting on $\mathcal{H}_R$. 

As a consequence of the convexity of the trace norm and the Schmidt decomposition theorem, it suffices to optimize Eq.~\eqref{eq:diamond-distance} over pure states in $\mathcal{H}_{RA}$ with $\operatorname{dim}(\mathcal{H}_R)=\operatorname{dim}(\mathcal{H}_A)$.

With the notation in place, we now present our analytical results. Generalization performance depends crucially on the metric entropy, which characterizes both classical and quantum machine learning models \cite{huang2021information}. Metric entropy is a measure of complexity or expressiveness for a set of objects endowed with a distance metric.

In Section \ref{SbSct:CoveringNumberBounds}, we take the diamond norm as the distance metric and prove metric entropy bounds for two sets of interest. First, we examine the set $\mathcal{U}_\mathcal{A}$ of all unitaries that can be represented using a (fixed) variational quantum circuit $\mathcal{A}$ with $T$ parameterized $2$-qubit unitary gates. More precisely, we consider the corresponding set of unitary channels. Second, we study the set $\mathcal{CPTP}_\mathcal{A}$ of all CPTP maps that can be represented using a (fixed) variational quantum circuit $\mathcal{A}$ with $T$ parameterized 2-qubit CPTP maps. The latter scenario generalizes the former and corresponds to the difference between perfect and noisy implementations. Note that, in both cases, the variational quantum circuit itself could contain more than $T$ gates. However, these additional gates would have to be fixed and not trainable.

Using these metric entropy bounds and variants thereof, we establish prediction error bounds for variational quantum machine learning models (QMLMs) in terms of the number of trainable elements in Section \ref{SbSctPredictionErrorBounds}. We consider different scenarios of interest, among them that of using multiple copies of a quantum neural network (such that parameters are reused over different copies), as well as both fixed and variable circuit architectures.

\subsection{Covering Number Bounds for Variational Quantum Circuits}\label{SbSct:CoveringNumberBounds}

In this section, we provide bounds on the expressivity of the class of CPTP maps (or unitaries) that a quantum machine learning model (QMLM) can implement in terms of the number of trainable elements used in the architecture. As a measure of expressivity, we choose covering numbers and metric entropies w.r.t.~(the metric induced by) the diamond norm. We first recall the general definition of covering numbers and metric entropies.

\begin{definition}[Covering nets, covering numbers, and metric entropies]\label{def:covering-numbers}
Let $(X,d)$ be a metric space. Let $K\subset X$ be a subset and let $\varepsilon >0$.
\begin{itemize}
    \item $N\subseteq K$ is an $\varepsilon$-covering net of $K$ if $\forall x\in K~\exists y\in N$ such that $d(x,y)\leq \varepsilon$. That is, $N\subseteq K$ is an $\varepsilon$-covering net of $K$ if and only if $K$ can be covered by $\varepsilon$-balls around the points in $N$.
    \item The covering number $\mathcal{N}(K,d,\varepsilon)$ is the smallest possible cardinality of an $\varepsilon$-covering net of $K$.
    \item The metric entropies $\log_2\mathcal{N}(K,d,\varepsilon)$ are  the logarithm of the covering numbers.
\end{itemize}
\end{definition}

In finite-dimensional real spaces, the covering numbers of norm balls, and thereby of norm-bounded sets, can be bounded easily. We make use of this observation to provide basic covering number bounds for the classes of $2$-qubit unitaries and $2$-qubit CPTP maps. We first state the bound for the unitary case.

\begin{lemma}[Covering number bounds for $2$-qubit unitaries] \label{lem:Covering2QuditUnitaries}
Let $\norm{\cdot}$ be a unitarily invariant norm on complex $4\times 4$-matrices. The covering number of the set of $2$-qubit unitaries $\mathcal{U}\left( \mathbb{C}^2 \otimes \mathbb{C}^2 \right)$ w.r.t.~the norm $\norm{\cdot}$ can be bounded as
\begin{equation}
    \mathcal{N}\left(\mathcal{U}\left( \mathbb{C}^2 \otimes \mathbb{C}^2\right), \norm{\cdot},\varepsilon\right)\leq \left(\frac{6\norm{\mathds{1}_{\mathbb{C}^{4}}}}{\varepsilon} \right)^{32},\quad \textrm{for }0<\varepsilon\leq \norm{\mathds{1}_{\mathbb{C}^4}}. 
\end{equation}
\end{lemma}
\begin{proof}
It is well known (see, e.g., Section $4.2$ in \cite{vershynin2018highdimensional}) that the covering numbers of a norm-ball of radius $R>0$ around some point $x\in\mathbb{R}^K$, for $0<\varepsilon\leq R$, can be bounded as 
\begin{equation}
    \mathcal{N}(B_R(x), \norm{\cdot}, \varepsilon)\leq \left( 1 + \frac{2R}{\varepsilon}\right)^K \leq \left(\frac{3R}{\varepsilon}\right)^K,
\end{equation}
where the ball and the coverings are taken w.r.t.~the same norm.

In our scenario, we can apply this as follows: As $\norm{\cdot}$ is assumed to be unitarily invariant, we have $\norm{U} = \norm{\mathds{1}_{\mathbb{C}^{4}}}$ for every unitary $U\in \mathcal{U}\left( \mathbb{C}^2 \otimes \mathbb{C}^2\right)$. In particular, we have, for $R:=\norm{\mathds{1}_{\mathbb{C}^{4}}}$ that $\mathcal{U}\left( \mathbb{C}^2 \otimes \mathbb{C}^2\right)\subset B_R(0)$, where $B_R(0)$ is the ball of matrices with $4\times 4=16$ complex entries around the $0$-matrix is taken w.r.t.~$\norm{\cdot}$. Therefore, we have
\begin{equation}
    \mathcal{N}\left(\mathcal{U}\left( \mathbb{C}^2 \otimes \mathbb{C}^2\right), \norm{\cdot},\varepsilon\right)
    \leq \mathcal{N}\left(B_R(0), \norm{\cdot},\frac{\varepsilon}{2}\right)
    \leq \left(\frac{6\norm{\mathds{1}_{\mathbb{C}^{4}}}}{\varepsilon} \right)^{2\cdot 16},\quad \textrm{for }0<\varepsilon\leq \norm{\mathds{1}_{\mathbb{C}^4}},
\end{equation}
where the first step uses the approximate monotonicity of (interior) covering numbers (see, e.g., Section $4.2$ in \cite{vershynin2018highdimensional}).
\end{proof}

This covering number bound becomes particularly useful for the spectral norm, for which $\norm{\mathds{1}_{\mathbb{C}^4}}=1$.

With an analogous reasoning, we can prove a covering number bound for $2$-qubit CPTP maps.

\begin{lemma}[Covering number bounds for $2$-qubit CPTP maps]\label{lem:Covering2QuditCPTP}
The covering number of the set of $2$-qubit CPTP maps $\mathcal{CPTP}\left( \mathbb{C}^2 \otimes \mathbb{C}^2 \right)$ w.r.t.~the diamond distance can be bounded as
\begin{equation}
    \mathcal{N}\left(\mathcal{CPTP}\left( \mathbb{C}^2 \otimes \mathbb{C}^2\right), \norm{\cdot}_\diamond,\varepsilon\right)\leq \left(\frac{6}{\varepsilon} \right)^{512},\quad \textrm{for }0<\varepsilon\leq 1. 
\end{equation}
\end{lemma}
\begin{proof}
As CPTP maps have diamond norm equal to $1$, this follows (analogously to the previous Lemma) by upper-bounding the covering number of the diamond-norm unit ball, which lives in a $(2^4\times 2^4)$-dimensional space over the complex numbers. The latter can be achieved as in the previous Lemma.
\end{proof}

We combine these basic upper bounds for single trainable elements with sub-additivity of the diamond norm (Lemma \ref{lem:subadditivity-diamond}) to obtain a covering number bound for the class of maps that can be implemented by a variational QMLM, understood as a parametrized CPTP map as described in the main text. Again, we first state the bound for the unitary case.

\begin{theorem}[Metric entropy bounds for unitary QMLMs]\label{thm-metric-entropy-vqc-unitary}
Let $\mathcal{E}^{\mathrm{QMLM}}_{\vec{\theta}}(\cdot)$ be an $n$-qubit QMLM with $T$ parameterized 2-qubit unitary gates and an arbitrary number of non-trainable, global unitary gates. Let $\mathcal{U}^\mathrm{QMLM}\subset\mathcal{U}(\mathbb{C}^{2^n})$ denote the set of $n$-qubit unitaries that can be implemented by the QMLM $\mathcal{E}^{\mathrm{QMLM}}_{\vec{\theta}}(\cdot)$.

Then, for every $\varepsilon\in (0,1]$, there exists an $\varepsilon$-covering net $\mathcal{N}_\varepsilon$ of (the set of unitary channels corresponding to) $\mathcal{U}^\mathrm{QMLM}$ w.r.t.~the diamond distance such that the logarithm of its size can be upper bounded as
\begin{equation}
    \log(\lvert\mathcal{N}_\epsilon\rvert) \leq 32 T \log\left(\frac{12 T}{\varepsilon} \right).\label{eq:unitary-circuit-covering-number-bound}
\end{equation}
\end{theorem}
\begin{proof}
Let $\varepsilon\in (0,1]$, write $\tilde{\varepsilon}:=\frac{\varepsilon}{2T}$. By Lemma \ref{lem:Covering2QuditUnitaries}, there exists an $\tilde{\varepsilon}$-net $\tilde{\mathcal{N}}_{\tilde{\varepsilon}}$ of $\mathcal{U}\left( \mathbb{C}^2\otimes \mathbb{C}^2\right)$ w.r.t.~the spectral norm of size $\lvert \tilde{\mathcal{N}}_{\tilde{\varepsilon}}\rvert\leq \left(\nicefrac{6}{\tilde{\varepsilon}} \right)^{32} = \left(\nicefrac{12T}{\varepsilon} \right)^{32}$.

Note that any $U\in\mathcal{U}^\mathrm{QMLM}$ is of the form $U= V_T U_T V_{T-1} U_{T-1} V_{T-2}\ldots V_1 U_1 V_0$, where $U_t$, $1\leq t\leq T$, are a particular choice of the trainable $2$-qubit unitaries and $V_s$, $0\leq s\leq T+1$, are the non-trainable $n$-qubit unitaries occurring in the QMLM. (For ease of readability, we have not written out the tensor factors of identities accompanying the $U_t$.) We now consider the set of unitaries obtained by plugging the elements of $\tilde{\mathcal{N}}_{\tilde{\varepsilon}}$ as trainable $2$-qubit unitaries into the QMLM. That is, we take
\begin{equation}
    \mathcal{N}_\varepsilon 
    := \left\{V_T U_T V_{T-1} U_{T-1} V_{T-2}\ldots V_1 U_1 V_0~|~U_t\in \tilde{\mathcal{N}}_{\tilde{\varepsilon}}, 1\leq t\leq T\right\}.
\end{equation}
Let $U\in\mathcal{U}^\mathrm{QMLM}$ be an arbitrary $n$-qubit unitary that can be implemented by the QMLM, i.e., $U= V_T U_T V_{T-1} U_{T-1} V_{T-2}\ldots V_1 U_1 V_0$ for some $U_t\in\mathcal{U}(\mathbb{C}^2\otimes\mathbb{C}^2)$, $1\leq t\leq T$. Let $\mathcal{U}$ denote the corresponding unitary channel. As $\tilde{\mathcal{N}}_{\tilde{\varepsilon}}$ is an $\tilde{\varepsilon}$-net for the set of $2$-qubit unitaries, we can find $\tilde{U}_t\in\tilde{\mathcal{N}}_{\tilde{\varepsilon}}$, $1\leq t\leq T$, such that $\norm{U_t-\tilde{U}_t}\leq\tilde{\varepsilon}$ for all $1\leq t\leq T$. Then, the unitary channel $\tilde{\mathcal{U}}$ described by $\tilde{U}:=V_T \tilde{U}_T V_{T-1} \tilde{U}_{T-1} V_{T-2}\ldots V_1 \tilde{U}_1 V_0\in\mathcal{N}_\varepsilon$ satisfies 
\begin{equation}
    \norm{\mathcal{U}-\tilde{\mathcal{U}}}_\diamond
    \leq \sum\limits_{s=0}^{T+1} \norm{\mathcal{V}_{s} - \tilde{\mathcal{V}}_{s}}_\diamond + \sum\limits_{t=1}^T \norm{\mathcal{U}_t - \tilde{\mathcal{U}}_t}_\diamond
    \leq 2\sum\limits_{t=1}^T \norm{U_t - \tilde{U}_t}
    \leq \varepsilon,
\end{equation}
where we iteratively applied sub-additivity of the diamond distance (Lemma \ref{lem:subadditivity-diamond}) in the first step, then used the relation between the diamond distance of unitary channels to the spectral norm distance of the corresponding unitaries (Lemma \ref{lem:diamond-to-operator}), and in the final step plugged in the definition of $\tilde{\varepsilon}$.

Thus, we have shown that the set of unitary channels with unitaries in $\mathcal{N}_\varepsilon$ is an $\varepsilon$-covering net of the set of unitary channels with unitaries in $\mathcal{U}^\mathrm{QMLM}$ w.r.t.~the diamond distance. As $\lvert \mathcal{N}_\varepsilon\rvert = \lvert \tilde{\mathcal{N}}_{\tilde{\varepsilon}}\rvert^T$ (by definition of $\mathcal{N}_\varepsilon$), plugging in the bound on the size of $\tilde{\mathcal{N}}_{\tilde{\varepsilon}}$ then gives the desired bound on the cardinality of $\mathcal{N}_\varepsilon$ and thereby of our covering net.
\end{proof}

For variational quantum circuits consisting of CPTP maps, we obtain an analogous result upon replacing Lemma~\ref{lem:Covering2QuditUnitaries} by Lemma~\ref{lem:Covering2QuditCPTP} in the previous proof:

\begin{theorem}[Metric entropy bounds for QMLMs of CPTP maps]\label{thm-metric-entropy-vqc-cptp} 
Let $\mathcal{E}^{\mathrm{QMLM}}_{\vec{\theta}}(\cdot)$ be an $n$-qubit QMLM with $T$ parameterized $2$-qubit CPTP maps and an arbitrary number of non-trainable, global CPTP maps. Let $\mathcal{CPTP}^{\mathrm{QMLM}}\subset\mathcal{CPTP}\left((\mathbb{C}^{2})^{\otimes n}\right)$ denote the set of $n$-qubit CPTP maps that can be implemented by the circuit QMLM $\mathcal{E}^{\mathrm{QMLM}}_{\vec{\theta}}(\cdot)$.

For any $\varepsilon\in (0,1]$, there exists an $\varepsilon$-covering net $\mathcal{N}_\varepsilon$ of $\mathcal{CPTP}^{\mathrm{QMLM}}$ w.r.t.~the diamond distance such that the logarithm of its size can be upper bounded as
\begin{equation}
    \log(|\mathcal{N}_\epsilon|) \leq 512 T \log\left(\frac{6 T}{\varepsilon} \right).\label{eq:cptp-circuit-covering-number-bound}
\end{equation}
\end{theorem}

In both scenarios, the metric entropy grows at worst slightly super-linearly with the number of parameterized (and thus trainable) operations.

We also provide a generalization of these metric entropy bounds that is natural for the scenario in which trainable gates are reused in the quantum machine learning model:

\begin{theorem}[Metric entropy bounds for QMLMs of reused CPTP maps]\label{thm-metric-entropy-cptp-multiple-copies}
Let $\mathcal{E}^{\mathrm{QMLM}}_{\vec{\theta}}(\cdot)$ be an $n$-qubit QMLM with $T$ parameterized $2$-qubit CPTP maps, in which the $t^{\mathrm{th}}$ of these maps is used $M_t$ times, and an arbitrary number of non-trainable, global CPTP maps. Let $\mathcal{CPTP}^{\mathrm{QMLM}}\subset\mathcal{CPTP}\left((\mathbb{C}^2)^{\otimes n}\right)$ denote the set of $n$-qubit CPTP maps that can be implemented by the QMLM $\mathcal{E}^{\mathrm{QMLM}}_{\vec{\theta}}(\cdot)$.

For any $\varepsilon\in (0,1]$, there exists an $\varepsilon$-covering net $\mathcal{N}_\varepsilon$ of $\mathcal{CPTP}^{\mathrm{QMLM}}$ w.r.t.~the diamond distance such that the logarithm of its size can be upper bounded as
\begin{equation}
    \log(|\mathcal{N}_\epsilon|) \leq 512 \left(T \log\left(\frac{6 T}{\varepsilon} \right) + \sum\limits_{t=1}^{T} \log(M_t)\right).\label{eq:cptp-circuit-multiple-copies-covering-number-bound}
\end{equation}
\end{theorem}
\begin{proof}
We can use the same reasoning as in the proof of Theorems \ref{thm-metric-entropy-vqc-unitary} and \ref{thm-metric-entropy-vqc-cptp} to show that we can define an $\varepsilon$-covering net $\mathcal{N}_\varepsilon$ for $\mathcal{CPTP}^{\mathrm{QMLM}}$ (w.r.t.~$\norm{\cdot}_\diamond$) by plugging the elements of an $\tilde{\varepsilon}_t$-net for $\mathcal{CPTP}(\mathbb{C}^2\otimes\mathbb{C}^2)$ into the positions of the QMLM corresponding to the $t^{\mathrm{th}}$ independently trainable $2$-qubit map, where $\tilde{\varepsilon}_t := \frac{\varepsilon}{T\cdot M_t}$. When picking the $\tilde{\varepsilon}_t$-nets with cardinality bounded as in Lemma \ref{lem:Covering2QuditCPTP}, the cardinality of $\mathcal{N}_\varepsilon$ can be bounded as 
\begin{equation}
    \lvert \mathcal{N}_\varepsilon\rvert 
    \leq\prod\limits_{t=1}^T \left( \frac{6TM_t}{\varepsilon}\right)^{512}
    = \left( \left( \frac{6T}{\varepsilon}\right)^{T}\cdot \left(\prod\limits_{t=1}^T M_t\right)\right)^{512}.
\end{equation}
Taking a logarithm gives the claimed metric entropy bound.
\end{proof}

The growth of the metric entropies in terms of $T$, the number of independently trainable maps, is still at most slightly super-linear. But the growth in terms of the numbers of times that the trainable maps are reused is only logarithmic.

Note that we have formulated the metric entropy bounds for the qubit case only, but they can naturally be extended to the qudit case. Then the upper bound will depend polynomially on the dimension $d$.

We provide one more metric entropy bound for QMLMs, which also takes the training procedure into account, in Theorem \ref{thm:metric-entropy-optimization-dependent}. Formulating this bound, however, requires us to fix some (notational) assumptions on the optimization procedure used for training. Therefore, we postpone this final metric entropy bound to Subsection \ref{SbSbSct:optimization-dependent}.

\begin{remark}\label{Rmk:k-local}
    Both in this section and in the following ones, we formulate our results for QMLMs whose parametrized gates act on (at most) $2$ qubits. Our proofs and results straightforwardly extend to the case in which the parametrized gates act on (at most) $\kappa$ qubits. In particular, when going from $2$- to $\kappa$-local, the $T$-dependence remains the same. Only the constant prefactors in the metric entropy bounds (and thus the generalization bounds) change, namely from $2\cdot 2^4$ to $2\cdot 2^{2\kappa}$ in the unitary case, and from $2\cdot 2^8$ to $2\cdot 2^{4\kappa}$ in the CPTP case. Since $\kappa$ is constant, then the latter is just prefactor that does not change the scaling of our theorems. 
\end{remark}

\subsection{Prediction error bounds for quantum machine learning models}\label{SbSctPredictionErrorBounds}

Using well-established tools from statistical learning theory, we can derive prediction error bounds for QMLMs from the covering number bounds established in Section \ref{SbSct:CoveringNumberBounds}. Before doing so, we describe our setting in detail.

During the training process, we optimize the parameters $\vec{\alpha}$ in the (CPTP map implemented by the) quantum machine learning model $\mathcal{E}^{\mathrm{QMLM}}_{\vec{\alpha}}(\cdot)$ according to some criteria and depending on the training data. Here, we write $\vec{\alpha}=(\vec{\theta}, \vec{k})$ if we consider both discrete, structural parameters $\vec{k}$ and continuous parameters $\vec{\theta}$. If the QMLM structure is fixed and only the continuous parameters are optimized, we write only $\vec{\theta}$ (instead of $\vec{\alpha}$). Note that we do not make any further assumptions on how the QMLM $\mathcal{E}^{\mathrm{QMLM}}_{\vec{\alpha}}(\cdot)$ depends on the parameters $\vec{\alpha}=(\vec{\theta}, \vec{k})$ other than that the discrete parameters only encode different choices of quantum circuit architectures. In particular, the dependence of the trainable gates on the continuous parameters $\vec{\theta}$ can be arbitrary.

We use an observable to quantify how good/bad the output state is, this will serve as our loss function.
More concretely, for an input $x_i$ and (classical or quantum) target output $y_i$, we define the loss function of the parameter setting $\vec{\alpha}$ to be
\begin{equation}
    \ell (\vec{\alpha};x_i,y_i)
    =\Tr\left[O^{\mathrm{loss}}_{x_i,y_i}(\mathcal{E}^{\mathrm{QMLM}}_{\vec{\alpha}}\otimes \operatorname{id})(\rho(x_i))\right],\label{eq:loss-function-appendix}
\end{equation}
for some Hermitian observables $O^{\mathrm{loss}}_{x_i, y_i}$. Here, $x\mapsto\rho(x)$ is some encoding of the classical data into quantum states that is fixed in advance.

As is common in classical learning theory, the prediction error bounds will depend on the largest (absolute) value that the loss function can attain. In our case, we therefore assume $C_\mathrm{loss}:=\sup_{x,y}\norm{O^{\mathrm{loss}}_{x, y}}<\infty$. That is, we assume that the spectral norm can be bounded uniformly over all possible loss observables.

For a training dataset $S=\{(x_i,y_i)\}_{i=1}^N$ of size $N\in\mathbb{N}$, the average loss on the training data is given by
\begin{equation}
    \hat{R}_S(\vec{\alpha})
    :=\frac{1}{N} \sum_{i=1}^N \ell (\vec{\alpha}; x_i, y_i)
    = \frac{1}{N} \sum_{i=1}^N \Tr\left[O^{\mathrm{loss}}_{x_i,y_i}(\mathcal{E}^{\mathrm{QMLM}}_{\vec{\alpha}}\otimes \operatorname{id})(\rho(x_i))\right],\label{eq:training-error}
\end{equation}
which is often referred to as the \emph{training error} or \emph{empirical risk}. This quantity can (in principle) be evaluated, given the parameter setting $\vec{\alpha}$ and the training data.

When we obtain a new input $x$, the prediction error of a parameter setting $\vec{\alpha}$ is taken to be the expected loss
\begin{equation}
    R(\vec{\alpha})
    :=\E_{(x, y)\sim P} \left[\ell (\vec{\alpha}; x, y)\right]
    = \E_{(x, y)\sim P} \Tr\left[O^{\mathrm{loss}}_{x,y}(\mathcal{E}^{\mathrm{QMLM}}_{\vec{\alpha}}\otimes \operatorname{id})(\rho(x))\right],\label{eq:prediction-error}
\end{equation}
where the expectation is w.r.t.~the distribution $P$ from which the training examples are generated. This quantity is called the \emph{prediction error} or \emph{expected risk}. The goal of any (classical or quantum) machine learning procedure is to achieve a small prediction error with high success probability.

As the underlying distribution $P$ is usually unknown, we cannot directly evaluate the prediction error, even if we know the parameters $\vec{\alpha}$. In practice, one therefore often takes the training error as a proxy for the prediction error. However, this procedure can only succeed if the difference between the prediction and the training error, the so-called \emph{generalization error}, is small. Our covering number bounds allow us to prove rigorous bounds on the generalization error incurred by a variational quantum machine learning method in the so-called \emph{``Probably Approximately Correct'' (PAC)} sense. That is, we provide bounds on the generalization error in terms of the desired success probability and the training data size. Thereby, our results provide guarantees on the prediction performance of a quantum machine learning model on unseen data, if that model performs well on the training data.

Our main result is the following:
\begin{theorem}[Mother Theorem]\label{thm-mother}
Let $\mathcal{E}^{\mathrm{QMLM}}_{\vec{\alpha}}(\cdot)$ be a QMLM with a variable structure. Suppose that, for every $k\in\mathbb{N}$, there are at most $G_\tau\in\mathbb{N}$ allowed structures with exactly $\tau$ parameterized $2$-qubit CPTP maps, in which the $t^{\mathrm{th}}$ of these maps is taken from a set $\mathcal{M}_t$ and used $M_t$ times, and an arbitrary number of non-trainable, global CPTP maps. Also, for each $t\in\mathbb{N}$, let $\mathcal{E}^0_t\in\mathcal{CPTP}\left((\mathbb{C})^{\otimes 2}\right)$ be a fixed reference CPTP map.
Let $P$ be a probability distribution over input-output pairs. Suppose that, given training data $S=\{(x_i,y_i)\}_{i=1}^N$ of size $N$, our optimization of the QMLM over structures and parameters w.r.t.~the loss function $\ell(\vec{\alpha};x_i,y_i)=\Tr\left[O^{\mathrm{loss}}_{x_i,y_i}(\mathcal{E}^{\mathrm{QMLM}}_{\vec{\alpha}}\otimes \operatorname{id})(\rho(x_i))\right]$ yields a (data-dependent) structure with $T=T(N)$ independently parameterized $2$-qubit CPTP maps, in which the $t^{\mathrm{th}}$ of these maps is taken from $\mathcal{M}_t$ and used $M_t$ times, as well as the parameter setting $\vec{\alpha}^\ast = \vec{\alpha}^\ast(S)$.

Then, with probability at least $1-\delta$ over the choice of i.i.d.~training data $S$ of size $N$ according to $P$,
\small
\begin{align}
    &R(\vec{\alpha}^\ast) - \hat{R}_S(\vec{\alpha}^\ast)\nonumber\\
    &\in \mathcal{O}\left(C_\mathrm{loss}\min\left\{\sqrt{\frac{K\max\limits_{1\leq t\leq T}c_t\log(K)}{N}} + \sqrt{\frac{K\log(T)}{N}} + \sqrt{\frac{K\max\limits_{1\leq t\leq T}c_t\log(M_t)}{N}} + \sum_{\substack{t=1 \\t\neq t_1,\ldots,t_K}}^T M_t \Delta_{t} + \sqrt{\frac{\log(G_T)}{N}} +\sqrt{\frac{\log(\nicefrac{1}{\delta})}{N}}\right\}\right),
\end{align}
\normalsize
where $\Delta^T_1,\ldots,\Delta^T_T$ denote the (data-dependent) distance between the trainable maps in the output QMLM to the respective reference maps $\mathcal{E}^0_1,\ldots,\mathcal{E}^0_T$, $C_\mathrm{loss}=\sup\limits_{x,y}\norm{O^{\mathrm{loss}}_{x, y}}$ is the maximum (absolute) value attainable by the loss function, and the minimum is over all $K\in\{0,\ldots,T\}$ and choices of pairwise distinct $t_1,\ldots,t_K\in\{1,\ldots,T\}$.

Moreover, if the loss is not evaluated exactly, but an unbiased estimator is built from $\sigma_{\mathrm{est}}$ subsampled training data points (as in Section \ref{SbSbSct:estimated-training-error}), we only incur an additional error of $\mathcal{O}\left(\sqrt{\nicefrac{\log(\nicefrac{1}{\delta})}{\sigma_{\mathrm{est}}}}\right)$.
\end{theorem}

Some of the important aspects of the upper bound on the generalization error of a QMLM provided by Theorem \ref{thm-mother} are: a dependence on the square root of the inverse of the training data size ($N$); an at worst slightly superlinear dependence on the number of trainable maps ($T$), which can improve if only a smaller number ($K$) of gates experience non-negligible changes during the optimization; a logarithmic dependence on the number of uses ($M_t$); a logarithmic dependence on the number of different architectures ($G_T$); and a logarithmic dependence on the reciprocal of the desired confidence level ($\delta$).

We build up to the proof of Theorem \ref{thm-mother} by first establishing our basic QML generalization error bound and then extending it in different directions. More precisely, we structure our presentation as follows: We start with the pedagogical Subsection \ref{SbSbSct:prelude}, in which we show a simple proof of how metric entropy bounds lead to generalization bounds, albeit not yet to their strongest form. In Subsection \ref{SbSbSct:prediction-error-fixed-architecture-single-copy}, we demonstrate how to improve upon the simple proof strategy using a more involved technique. Then, we extend the basic generalization error bounds in multiple directions, namely to multiple copies and reused trainable maps (Subsection \ref{SbSbSct:prediction-error-multiple-copies}), variable architecture (Subsection \ref{SbSbSct:variable-architecture}), optimization-dependent guarantees (Subsection \ref{SbSbSct:optimization-dependent}), and to a scenario in which we can not evaluate the loss function exactly, but only indirectly through an unbiased estimator (Subsection \ref{SbSbSct:estimated-training-error}). Finally, we bring together all these extensions into the most general form of our result (Subsection \ref{SbSbSct:mother-theorem}). Our line of reasoning is summarized in Fig.~\ref{fig:flowchart}.

\begin{figure}[ht]
    \centering
    \includegraphics[width=\linewidth]{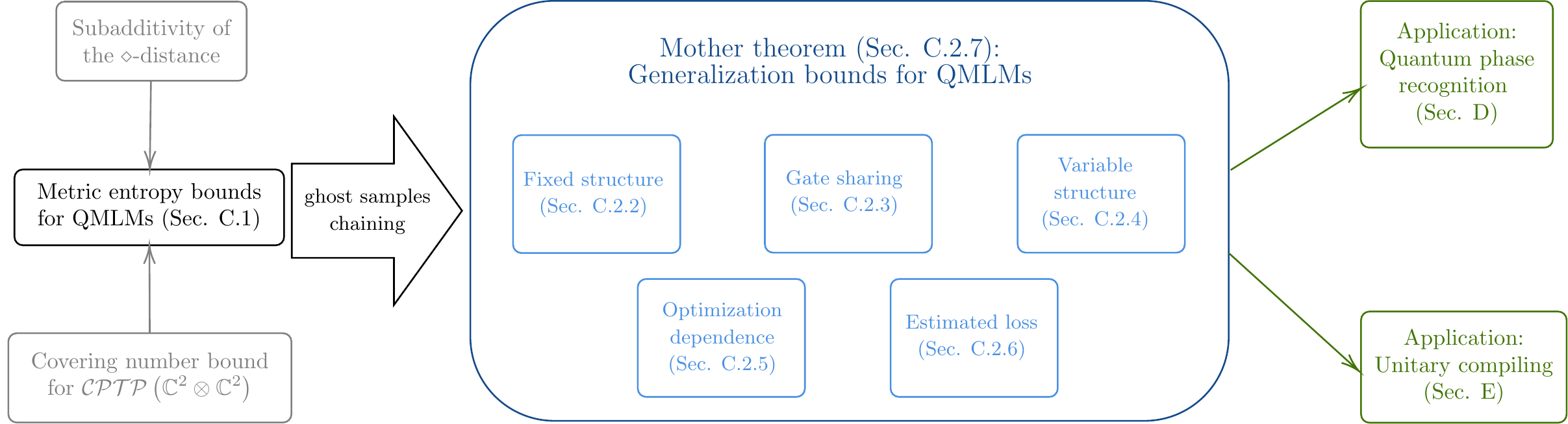}
    \caption{\textbf{Visualization of the proof structure.} We prove metric entropy bounds and use them to derive generalization bounds for different QMLM settings. We then apply our theory to quantum phase recognition and unitary compiling.}
    \label{fig:flowchart}
\end{figure}

\begin{remark}\label{rmk:nonlinear-loss}
    A loss function of the form of Eq.~\eqref{eq:loss-function-appendix} automatically has a certain linear structure, namely, it depends linearly on the output state $(\mathcal{E}^{\mathrm{QMLM}}_{\vec{\alpha}}\otimes \operatorname{id})(\rho(x_i))$.
    Notice, however, that we can introduce also a certain type of nonlinearity through the spectral decomposition of the loss observables $O^{\mathrm{loss}}_{x_i,y_i}$.
    Namely, suppose that we obtain a classical output from the QMLM by measuring an observable $O_{\textrm{out}}$ with spectral decomposition $O_{\textrm{out}} = \sum_j \lambda_j \ketbra{j}{j}$. That is, given an input $x_i$, we output $\lambda_j$ with probability $p_j(\vec{\alpha},x_i) = \Tr[\ketbra{j}{j}(\mathcal{E}^{\mathrm{QMLM}}_{\vec{\alpha}}\otimes \operatorname{id})(\rho(x_i))] = \bra{j}(\mathcal{E}^{\mathrm{QMLM}}_{\vec{\alpha}}\otimes \operatorname{id})(\rho(x_i))\ket{j}$.
    Now, we can, for example, define the loss observables as $O^{\mathrm{loss}}_{x_i,y_i}\coloneqq \sum_j (y_i - \lambda_j)^2\ketbra{j}{j}$, so that $\ell (\vec{\alpha};x_i,y_i) = \mathbb{E}_{}[(y_i - \lambda_j)^2]$ becomes the expected square loss between the true label $y_i$ and our output $\lambda_i$. Here, the expectation is w.r.t.~$(p_j(\vec{\alpha},x_i))_j$.
    Clearly, here we can replace $(y_i - \lambda_j)^2$ by any nonlinear loss function $\tilde{\ell}(y_i,\lambda_j)$ of interest.
\end{remark}

\subsubsection{Prelude: Metric entropy bounds imply generalization error bounds}\label{SbSbSct:prelude}

This section is intended to help readers not yet familiar with the theory of classical machine learning gain an intuition for how we derive our analytical results. We present a technically simple proof of a generalization bound for a fixed-architecture QMLM, which, however, is worse than that of Theorem \ref{thm:prediction-error-bound-qnn} by a factor logarithmic in the training data size. Therefore, readers already well versed in statistical learning theory, or readers who want to focus on the results and not the proofs, can safely skip this pedagogical section.

We demonstrate how the metric entropy bound from Theorem \ref{thm-metric-entropy-vqc-cptp} gives rise to a generalization bound for QMLMs with a fixed architecture, in which each trainable $2$-qubit CPTP map is used only once. The simplified proof given in this section consists in combining Hoeffding's concentration inequality (Lemma \ref{lem:Hoeffding}) with a union bound over a suitable covering net. Informally, we show that it suffices to prove good generalization simultaneously for all elements in a covering net, which we can obtain from a union bound over standard concentration guarantees for each single element of the covering net.

\begin{theorem}[Prediction error bound for quantum machine learning - Fixed structure (Preliminary version)]\label{thm:prediction-error-bound-qnn-simplified} Let $\mathcal{E}^{\mathrm{QMLM}}_{\vec{\theta}}(\cdot)$ be a QMLM with a fixed architecture consisting of $T$ parameterized 2-qubit CPTP maps and an arbitrary number of non-trainable, global CPTP maps. Let $P$ be a probability distribution over input-output pairs. Suppose that, given training data $S=\{(x_i,y_i)\}_{i=1}^N$ of size $N$, our optimization yields the parameter setting $\vec{\theta}^\ast = \vec{\theta}^\ast(S)$.

Then, with probability at least $1-\delta$ over the choice of i.i.d.~training data $S$ of size $N$ according to $P$, 
\begin{equation}
    R(\vec{\theta}^\ast) - \hat{R}_S(\vec{\theta}^\ast)
    \in \mathcal{O}\left(C_\mathrm{loss} \left(\sqrt{\frac{T\log\left(TN\right)}{N}} + \sqrt{\frac{\log(\nicefrac{1}{\delta})}{N}} \right)\right).
\end{equation}
\end{theorem}
\begin{proof}
For any parameter setting $\vec{\theta}$, fixed independently of the choice of training data, we see that $\ell (\vec{\theta}; x_1, y_1),\ldots,\ell(\vec{\theta}; x_N, y_N)$ are independent random variables taking values in $[-C_{\mathrm{loss}} , C_\mathrm{loss}]$. So Hoeffding's Lemma (Lemma \ref{lem:Hoeffding}) tells us that, $\forall\eta>0$, have
\begin{equation}
    \mathbb{P}_S\left[R(\vec{\theta}) - \hat{R}_S(\vec{\theta}) > \eta \right] \leq \exp\left( - \frac{N \eta^2}{2C_{\mathrm{loss}}^2} \right).\label{eq:hoeffding-fixed-parameter}
\end{equation}
Here, $\mathbb{P}_S [\cdot] = \mathbb{P}_{S\sim P^N} [\cdot]$ denotes the probability over training data sets $S=\{(x_i,y_i)\}_{i=1}^N$ of size $N$, with the $(x_i,y_i)$ drawn i.i.d.~from the probability measure $P$.
Next, we let $\varepsilon=\sqrt{\nicefrac{T}{N}}>0$, take $\mathcal{N}_\varepsilon$ to be an $\varepsilon$-covering net of the set of CPTP maps that can be implemented by the QMLM, and we take a union bound over $\mathcal{N}_\varepsilon$, with which we obtain
\begin{equation}
    \mathbb{P}_S\left[\exists \mathcal{E}^{\mathrm{QMLM}}_{\vec{\theta}}\in\mathcal{N}_\varepsilon:~R(\vec{\theta}) - \hat{R}_S(\vec{\theta}) > \eta\right] \leq \lvert\mathcal{N}_\varepsilon\rvert\cdot \exp\left( - \frac{N \eta^2}{2C_{\mathrm{loss}}^2} \right).
\end{equation}
As we took $\mathcal{N}_\varepsilon$ to be an $\varepsilon$-covering net (w.r.t.~the diamond norm) of the class of CPTP maps that the QMLM can implement, and since $\norm{\mathcal{E}-\tilde{\mathcal{E}}}_\diamond\leq\varepsilon$ directly implies, for all $x\in\mathcal{X}$, $y\in\mathcal{Y}$,
\begin{equation}
    \left\lvert\Tr\left[O^{\mathrm{loss}}_{x,y}(\mathcal{E}\otimes\operatorname{id})(\rho(x))\right]-\Tr\left[O^{\mathrm{loss}}_{x,y}(\tilde{\mathcal{E}}\otimes\operatorname{id})(\rho(x))\right]\right\rvert
    \leq \norm{O^{\mathrm{loss}}_{x,y}}\cdot \norm{\mathcal{E}-\tilde{\mathcal{E}}}_\diamond
    \leq C_\mathrm{loss}\varepsilon,
\end{equation}
we conclude, because of the form of the loss function $\ell$, that
\begin{align}
    \mathbb{P}_S\left[ R(\vec{\theta}^\ast) - \hat{R}_S(\vec{\theta}^\ast) > \eta + 2C_\mathrm{loss}\varepsilon\right] 
    &\leq \mathbb{P}_S\left[\exists \mathcal{E}^{\mathrm{QMLM}}_{\vec{\theta}}\in\mathcal{N}_\varepsilon:~\hat{R}_S(\vec{\theta}) > R(\vec{\theta}) + \eta\right] \\
    &\leq \lvert\mathcal{N}_\varepsilon\rvert\cdot \exp\left( - \frac{N \eta^2}{2C_{\mathrm{loss}}^2} \right)
\end{align}
Thus, for any $\delta\in (0,1)$, by choosing $\eta = C_{\mathrm{loss}} \sqrt{\frac{2\log(\nicefrac{\lvert\mathcal{N}_\varepsilon\rvert}{\delta})}{N}}$, we can guarantee that, with probability at least $1 - \delta$ over the choice of training data $S$ of size $N$, we have
\begin{equation}
    R(\vec{\theta}^\ast) - \hat{R}_S(\vec{\theta}^\ast)
    \leq C_{\mathrm{loss}} \sqrt{\frac{2\log(\nicefrac{\lvert\mathcal{N}_\varepsilon\rvert}{\delta})}{N}} + 2C_{\mathrm{loss}}\sqrt{\frac{T}{N}}.
\end{equation}
Now, we recall that, by Theorem \ref{thm-metric-entropy-vqc-cptp}, we can take $\mathcal{N}_\varepsilon$ to satisfy $\log(\lvert\mathcal{N}_\varepsilon\rvert) \leq 512T\log(\nicefrac{6T}{\varepsilon})$. Plugging this into the previous bound, we see that, with probability at least $1 - \delta$ over the choice of training data of size $N$, we have
\begin{align}
    R(\vec{\theta}^\ast) - \hat{R}_S(\vec{\theta}^\ast)
    &\leq C_{\mathrm{loss}} \sqrt{\frac{2\cdot (512T\log(\nicefrac{6T}{\varepsilon}) + \log(\nicefrac{1}{\delta}))}{N}} + 2C_{\mathrm{loss}}\sqrt{\frac{T}{N}}\\
    &\leq C_{\mathrm{loss}} \sqrt{\frac{2\cdot (512T\log(6\sqrt{TN}) + \log(\nicefrac{1}{\delta}))}{N}} + 2C_{\mathrm{loss}}\sqrt{\frac{T}{N}}\\
    &\in \mathcal{O}\left(C_{\mathrm{loss}} \left(\sqrt{\frac{T\log(TN)}{N}} + \sqrt{\frac{\log(\nicefrac{1}{\delta})}{N}}\right) \right),
\end{align}
which is the claimed generalization error bound.
\end{proof}

\begin{remark}
At first glance, it might seem that simply plugging the parameter setting $\vec{\theta}^\ast$ into Eq.~\eqref{eq:hoeffding-fixed-parameter} would already give us a good concentration bound for the parameter setting $\vec{\theta}^\ast$ obtained through training and that the union bound over the covering net is not actually necessary in the above proof. However, as the parameter setting $\vec{\theta}^\ast=\vec{\theta}^\ast(S)$ depends on the whole training data set $S$, the random variables $\ell (\vec{\theta}^\ast; x_i, y_i)$, $i=1,\ldots,N$, are not statistically independent. Thus, Hoeffding's inequality alone cannot be used to obtain a version of Eq.~\eqref{eq:hoeffding-fixed-parameter} with $\vec{\theta}$ replaced by the data-dependent $\vec{\theta}^\ast$. 
\end{remark}

The generalization bound established in Theorem \ref{thm:prediction-error-bound-qnn-simplified} already shows the right behavior in terms of the dependence on $T$, the number of trainable maps. However, the dependence on $N$, the sample size, still contains an undesirable logarithmic term. In classical statistical learning theory, it is well known that a proof strategy as above, based on combining Hoeffding's concentration inequality with a union bound over a covering net, incurs such a $\log(N)$-term. Fortunately, a technique for removing this term is also known and we will use it to tighten the prediction error bound in the next subsection.

\subsubsection{Basic prediction error bound for fixed architecture}\label{SbSbSct:prediction-error-fixed-architecture-single-copy}

Our first prediction error bound is for the case of a variational QMLM with a fixed architecture. In particular, while the parameters in the trainable $2$-qubit CPTP maps can be optimized over, the structure of the QMLM, i.e., the arrangement of the different elements, and in particular the overall depth and size, remain fixed. (We provide a generalization to variable circuit architectures in Section \ref{SbSbSct:variable-architecture}.) In this scenario, we have the following generalization error bound:

\begin{theorem}[Prediction error bound for quantum machine learning - Fixed structure]\label{thm:prediction-error-bound-qnn}
Let $\mathcal{E}^{\mathrm{QMLM}}_{\vec{\theta}}(\cdot)$ be a QMLM with a fixed architecture consisting of $T$ parameterized 2-qubit CPTP maps and an arbitrary number of non-trainable, global CPTP maps. Let $P$ be a probability distribution over input-output pairs. Suppose that, given training data $S=\{(x_i,y_i)\}_{i=1}^N$ of size $N$, our optimization yields the parameter setting $\vec{\theta}^\ast = \vec{\theta}^\ast(S)$.

Then, with probability at least $1-\delta$ over the choice of i.i.d.~training data $S$ of size $N$ according to $P$, 
\begin{equation}
    R(\vec{\theta}^\ast) - \hat{R}_S(\vec{\theta}^\ast)
    \in \mathcal{O}\left(C_\mathrm{loss} \left(\sqrt{\frac{T\log\left(T\right)}{N}} + \sqrt{\frac{\log(\nicefrac{1}{\delta})}{N}} \right)\right).\label{eq:gen-bound-fixed-architecture}
\end{equation}
\end{theorem}

In case the training data contains quantum labels, we assume the training data states to be reproducible so that we can use the data both for the optimization procedure and for evaluating the training error.

\begin{proof}
The proof proceeds in two steps: The first step is to upper-bound the generalization error in terms of the expected supremum of a random process. (This well known technique is described, e.g., in Theorem $3.3$ in \cite{mohri2018foundations}.) In the second step, we invoke the chaining technique to further upper-bound this expected supremum in terms of covering numbers. (This method goes back to \cite{dudley1999uniform}. See, e.g., Section $8$ of \cite{vershynin2018highdimensional} for a pedagogical presentation.) At this point, we apply our covering numbers bounds to finish the proof.

For ease of notation in the first step, we define $\varphi:(\mathcal{X}\times\mathcal{Y})^N\to\mathbb{R}$ as $\varphi(S):= \sup_{\vec{\theta}} \left\{R(\vec{\theta}) - \hat{R}_S(\vec{\theta})\right\}$, where the supremum goes over all possible parameter settings in the QMLM.
We first observe that, if $S=\{(x_i,y_i)\}_{i=1}^N$ and $\tilde{S}=\{(\tilde{x}_i,\tilde{y}_i)\}_{i=1}^N$ differ only in a single labelled example, then $\lvert\varphi(S)-\varphi(\tilde{S})\rvert\leq \nicefrac{2C_\mathrm{loss}}{N}$ (because the loss function has values in $[-C_\mathrm{loss}, C_\mathrm{loss}]$). Therefore, we can apply McDiarmid's inequality (Lemma \ref{lem:McDiarmid}) and obtain that, for every $\varepsilon>0$, $\mathbb{P}_{S}[\varphi(S)-\mathbb{E}_{\tilde{S}}[\varphi(\tilde{S})]\geq\varepsilon]\leq \exp\left(-\nicefrac{N\varepsilon^2}{2C_\mathrm{loss}^2}\right)$. Hence, for every $\delta\in (0,1)$, with probability $\geq 1-\nicefrac{\delta}{2}$ over the choice of training data, we have , with $\vec{\theta}^\ast = \vec{\theta}^\ast(S)$ as in the statement of the Theorem,
\begin{equation}
    R(\vec{\theta}^\ast) - \hat{R}_S(\vec{\theta}^\ast)
    \leq \varphi(S)
    \leq \mathbb{E}_{\tilde{S}}[\varphi(\tilde{S})] +C_\mathrm{loss}\sqrt{\frac{2\log(\nicefrac{2}{\delta})}{N}}. \label{eq:risk-concentration}
\end{equation}
We now upper-bound $\mathbb{E}_{\tilde{S}}[\varphi(\tilde{S})]$. To this end, we introduce a so-called ghost sample. Namely, we take $S'=\{(x'_i,y'_i)\}_{i=1}^N$ to be an i.i.d.~copy of $\tilde{S}$. Then, we can bound
\begin{align}
    \mathbb{E}_{\tilde{S}}[\varphi(\tilde{S})]
    &= \mathbb{E}_{\tilde{S}}\left[\sup_{\vec{\theta}} \left\{\frac{1}{N}\sum\limits_{i=1}^N \left(\mathbb{E}_{(x'_i, y'_i)\sim P}[\ell (\vec{\theta}; x'_i, y'_i)] - \ell (\vec{\theta}; \tilde{x}_i, \tilde{y}_i) \right)\right\}\right]\\
    &\leq \mathbb{E}_{\tilde{S}, S'}\left[\sup\limits_{\vec{\theta}}\left\{\frac{1}{N}\sum\limits_{i=1}^N \left(\ell (\vec{\theta}; x'_i, y'_i) - \ell (\vec{\theta}; \tilde{x}_i, \tilde{y}_i)  \right)\right\}\right].
\end{align}
Now, we use a standard symmetrization argument with i.i.d.~Rademacher random variables to further upper-bound the right hand side. That is, we let $\sigma_1,\ldots,\sigma_N$ be i.i.d.~Rademacher random variables, each distributed uniformly on $\{-1,1\}$. As multiplying $\left( \ell (\vec{\theta}; x'_i, y'_i) - \ell (\vec{\theta}; \tilde{x}_i, \tilde{y}_i)\right)$ by $- 1$ is equivalent to interchanging the i.i.d.~copies $(\tilde{x}_i, \tilde{y}_i)$ and $(x'_i,y'_i)$, which leaves the expectation invariant, we can introduce an additional expectation value over Rademacher variables as follows:
\begin{align}
    \mathbb{E}_{\tilde{S}, S'}\left[\sup\limits_{\vec{\theta}}\left\{\frac{1}{N}\sum\limits_{i=1}^N \left(\ell (\vec{\theta}; x'_i, y'_i) - \ell (\vec{\theta}; \tilde{x}_i, \tilde{y}_i)  \right)\right\}\right]
    &= \mathbb{E}_{\tilde{S}, S'}\mathbb{E}_\sigma\left[\sup\limits_{\vec{\theta}}\left\{\frac{1}{N}\sum\limits_{i=1}^N \sigma_i\left(\ell (\vec{\theta}; x'_i, y'_i) - \ell (\vec{\theta}; \tilde{x}_i, \tilde{y}_i)  \right)\right\}\right]\\
    &\leq 2\mathbb{E}_{\tilde{S}}\mathbb{E}_\sigma\left[\sup\limits_{\vec{\theta}} \frac{1}{N}\sum\limits_{i=1}^N \sigma_i \ell (\vec{\theta}; \tilde{x}_i, \tilde{y}_i)\right].
\end{align}
The quantity on the right hand side is not an empirical quantity, i.e., it cannot be directly evaluated only from the training data without knowledge of the underlying distribution $P$. However, another application of McDiarmid's inequality shows that, for every $\varepsilon>0$, 
\begin{equation}
    \mathbb{P}_S\left[\mathbb{E}_{\tilde{S}}\mathbb{E}_\sigma\left[\sup\limits_{\vec{\theta}} \frac{1}{N}\sum\limits_{i=1}^N \sigma_i \ell (\vec{\theta}; \tilde{x}_i, \tilde{y}_i)\right] - \mathbb{E}_\sigma\left[\sup\limits_{\vec{\theta}} \frac{1}{N}\sum\limits_{i=1}^N \sigma_i \ell (\vec{\theta}; x_i, y_i)\right] \geq \varepsilon\right]
    \leq \exp\left(- \frac{N\varepsilon^2}{2C_\mathrm{loss}^2}\right),
\end{equation}
where we again used that the loss function has values in $[-C_\mathrm{loss}, C_\mathrm{loss}]$. In other words, for every $\delta\in (0,1)$, with probability $\geq 1-\nicefrac{\delta}{2}$ over the choice of training data, we have
\begin{equation}
    \mathbb{E}_{\tilde{S}}\mathbb{E}_\sigma\left[\sup\limits_{\vec{\theta}} \frac{1}{N}\sum\limits_{i=1}^N \sigma_i \ell (\vec{\theta}; \tilde{x}_i, \tilde{y}_i)\right]
    \leq 
    \mathbb{E}_\sigma\left[\sup\limits_{\vec{\theta}} \frac{1}{N}\sum\limits_{i=1}^N \sigma_i \ell (\vec{\theta}; x_i, y_i)\right]
    +C_\mathrm{loss}\sqrt{\frac{2\log(\nicefrac{2}{\delta})}{N}}.\label{eq:expectation-Rademacher-bound}
\end{equation}
When applying a union bound, we can combine Eq.~\eqref{eq:risk-concentration} and \eqref{eq:expectation-Rademacher-bound} to conclude: For every $\delta\in (0,1)$, with probability $\geq 1-\delta$ over the choice of training data of size $N$, we have
\begin{equation}
    R(\vec{\theta}^\ast) - \hat{R}(\vec{\theta}^\ast)
    \leq 2\mathbb{E}_\sigma\left[\sup\limits_{\vec{\theta}} \frac{1}{N}\sum\limits_{i=1}^N \sigma_i \ell (\vec{\theta}; x_i, y_i)\right] + 3C_\mathrm{loss}\sqrt{\frac{2\log(\nicefrac{2}{\delta})}{N}}.\label{eq:Rademacher-generalization-bound}
\end{equation}
This concludes the first step of the proof.

As a second step, we use chaining to upper-bound $\mathbb{E}_\sigma\left[\sup_{\vec{\theta}} \frac{1}{N}\sum_{i=1}^N \sigma_i \ell (\vec{\theta}; x_i, y_i)\right]$ in terms of covering numbers. For $j\in\mathbb{N}_0$, define $\alpha_j:=2^{-j}C_\mathrm{loss}$. By Theorem \ref{thm-metric-entropy-vqc-cptp}, for every $j\in\mathbb{N}_0$, there exists an $2^{-j}$-covering net $\mathcal{N}_j$ (w.r.t.~the diamond norm) of the set of CPTP maps that can be implemented by the QMLM, satisfying $\left\lvert\mathcal{N}_j\right\rvert= \left(\nicefrac{6T}{2^{-j}}\right)^{512 T} = \left(6T\cdot 2^j\right)^{512 T}$. In particular, for every $j\in\mathbb{N}$ and for every parameter setting $\vec{\theta}$, there exists a CPTP map $\mathcal{E}_{\vec{\theta},j}\in \mathcal{N}_j$ and $\norm{\mathcal{E}^{\mathrm{QMLM}}_{\vec{\theta}}-\mathcal{E}_{\vec{\theta},j}}\diamond\leq 2^{-j}$. For $j=0$, we can take the $1$-covering net $\mathcal{N}_0=\{0\}$.

With this observation at hand, we can bound, for any $m\in\mathbb{N}$,
\begin{align}
    &\hphantom{=}\mathbb{E}_\sigma\left[\sup\limits_{\vec{\theta}} \frac{1}{N}\sum\limits_{i=1}^N \sigma_i \ell (\vec{\theta}; x_i, y_i)\right]\\
    &= \mathbb{E}_\sigma\left[\sup\limits_{\vec{\theta}} \frac{1}{N}\sum\limits_{i=1}^N \sigma_i \Tr\left[O^\mathrm{loss}_{x_i,y_i}(\mathcal{E}^{\mathrm{QMLM}}_{\vec{\theta}}\otimes \operatorname{id})(\rho(x_i))\right]\right]\\
    &= \frac{1}{N}\mathbb{E}_\sigma\left[\sup\limits_{\vec{\theta}}\left\{ \sum\limits_{i=1}^N \sigma_i \left(\Tr\left[O^\mathrm{loss}_{x_i,y_i}((\mathcal{E}^{\mathrm{QMLM}}_{\vec{\theta}} -\mathcal{E}_{\vec{\theta},m})\otimes \operatorname{id})(\rho(x_i))\right] + \sum_{j=1}^m \Tr\left[O^\mathrm{loss}_{x_i,y_i}((\mathcal{E}_{\vec{\theta},j} - \mathcal{E}_{\vec{\theta},j-1})\otimes \operatorname{id})(\rho(x_i))\right]\right)\right\}\right]\\
    &\leq \frac{1}{N}\mathbb{E}_\sigma\left[\sup\limits_{\vec{\theta}} \sum\limits_{i=1}^N \sigma_i \Tr\left[O^\mathrm{loss}_{x_i,y_i}((\mathcal{E}^{\mathrm{QMLM}}_{\vec{\theta}} -\mathcal{E}_{\vec{\theta},m})\otimes \operatorname{id})(\rho(x_i))\right]\right]\\
    &\hphantom{\leq~}+ \frac{1}{N}\sum_{j=1}^m\mathbb{E}_\sigma\left[\sup\limits_{\vec{\theta}} \sum\limits_{i=1}^N \sigma_i\Tr\left[O^\mathrm{loss}_{x_i,y_i}((\mathcal{E}_{\vec{\theta},j} - \mathcal{E}_{\vec{\theta},j-1})\otimes \operatorname{id})(\rho(x_i))\right] \right]\label{eq:chaining-setup}
\end{align}
where we used the telescopic sum representation $\mathcal{E}^{\mathrm{QMLM}}_{\vec{\theta}} = \mathcal{E}^{\mathrm{QMLM}}_{\vec{\theta}} - \mathcal{E}_{\vec{\theta},m} + \sum_{j=1}^m (\mathcal{E}_{\vec{\theta},j} - \mathcal{E}_{\vec{\theta},j-1})$.
We bound the two summands appearing in Eq.~\eqref{eq:chaining-setup} separately. For the first term, we can apply Hölder's inequality to obtain
\begin{align}
    &\frac{1}{N}\mathbb{E}_\sigma\left[\sup\limits_{\vec{\theta}} \sum\limits_{i=1}^N \sigma_i \Tr\left[O^\mathrm{loss}_{x_i,y_i}((\mathcal{E}^{\mathrm{QMLM}}_{\vec{\theta}} -\mathcal{E}_{\vec{\theta},m})\otimes \operatorname{id})(\rho(x_i))\right]\right]\\
    &\leq \frac{1}{N}\mathbb{E}_\sigma\left[\sup\limits_{\vec{\theta}} \sum\limits_{i=1}^N \left\lvert\Tr\left[O^\mathrm{loss}_{x_i,y_i}((\mathcal{E}^{\mathrm{QMLM}}_{\vec{\theta}} -\mathcal{E}_{\vec{\theta},m})\otimes \operatorname{id})(\rho(x_i))\right]\right\rvert\right]\\
    &\leq \frac{1}{N}\mathbb{E}_\sigma\left[\sup\limits_{\vec{\theta}} \sum\limits_{i=1}^N C_\mathrm{loss}\norm{((\mathcal{E}^{\mathrm{QMLM}}_{\vec{\theta}} -\mathcal{E}_{\vec{\theta},m})\otimes \operatorname{id})(\rho(x_i))}_1 \right]\\
    &\leq \frac{C_\mathrm{loss}}{N}\mathbb{E}_\sigma\left[\sup\limits_{\vec{\theta}} \sum\limits_{i=1}^N \norm{\mathcal{E}^{\mathrm{QMLM}}_{\vec{\theta}} -\mathcal{E}_{\vec{\theta},m}}_\diamond \right]\\
    &\leq C_\mathrm{loss}\cdot 2^{-m}\\
    &= \alpha_m.\label{eq:chaining-first-term}
\end{align}
For the second term, we observe that, thanks to Minkowski's inequality, for every parameter setting $\vec{\theta}$, 
\begin{align}
    &\hphantom{=}\sqrt{\sum\limits_{i=1}^N \left\lvert \Tr\left[O^\mathrm{loss}_{x_i,y_i}((\mathcal{E}_{\vec{\theta},j} - \mathcal{E}_{\vec{\theta},j-1})\otimes \operatorname{id})(\rho(x_i))\right]\right\rvert^2}\\
    &\leq \sqrt{\sum\limits_{i=1}^N \left\lvert \Tr\left[O^\mathrm{loss}_{x_i,y_i}((\mathcal{E}_{\vec{\theta},j} - \mathcal{E}^\mathrm{QMLM}_{\vec{\theta}})\otimes \operatorname{id})(\rho(x_i))\right]\right\rvert^2}  + \sqrt{\sum\limits_{i=1}^N \left\lvert \Tr\left[O^\mathrm{loss}_{x_i,y_i}((\mathcal{E}^\mathrm{QMLM}_{\vec{\theta}} - \mathcal{E}_{\vec{\theta},j-1})\otimes \operatorname{id})(\rho(x_i))\right]\right\rvert^2}\\
    &\leq \sqrt{N}C_\mathrm{loss}\left(\norm{\mathcal{E}_{\vec{\theta},j}-\mathcal{E}^{\mathrm{QMLM}}_{\vec{\theta}}}_\diamond + \norm{\mathcal{E}^{\mathrm{QMLM}}_{\vec{\theta}} -\mathcal{E}_{\vec{\theta},j-1}}_\diamond \right)\\
    &\leq \sqrt{N}(\alpha_j + \alpha_{j-1})\\
    &\leq 3\alpha_j\sqrt{N}.
\end{align}
Therefore, for each $1\leq j\leq m$, we can apply Massart's Lemma (Lemma \ref{lem:Massart}) to the set
\begin{equation}
    A
    \coloneqq \left\{\left((\Tr\left[O^\mathrm{loss}_{x_i,y_i}(\mathcal{E}_{\vec{\theta},j} - \mathcal{E}_{\vec{\theta},j-1})(\rho(x_i))\right]\right)_{i=1}^N\right\}_{\mathcal{E}_{\vec{\theta},j}\in\mathcal{N}_j,\mathcal{E}_{\vec{\theta},j-1}\in\mathcal{N}_{j-1}}\subset \mathbb{R}^N
\end{equation}
with radius $3\alpha_j\sqrt{N}$ and cardinality $\leq \lvert\mathcal{N}_j\rvert\cdot \lvert\mathcal{N}_{j-1}\rvert\leq \lvert\mathcal{N}_j\rvert^2$ to obtain
\begin{align}
    \frac{1}{N}\sum_{j=1}^m\mathbb{E}_\sigma\left[\sup\limits_{\vec{\theta}} \sum\limits_{i=1}^N \sigma_i\Tr\left[O^\mathrm{loss}_{x_i,y_i}(\mathcal{E}_{\vec{\theta},j} - \mathcal{E}_{\vec{\theta},j-1})(\rho(x_i))\right] \right]
    &\leq \frac{3}{\sqrt{N}}\sum_{j=1}^m\alpha_j\sqrt{2\log\left( \lvert\mathcal{N}_j\rvert^2\right)}\\
    &\leq \frac{6}{\sqrt{N}}\sum_{j=1}^m\alpha_j\sqrt{512T\log\left(6T\cdot 2^{j}\right)},
\end{align}
where we used the bound on the sizes of the covering nets in the last step.

If we now use $2^{-j}=2\int_{2^{-j-1}}^{2^{-j}}\mathrm{d}\alpha$, we can rewrite the upper bound as
\begin{align}
    \frac{6}{\sqrt{N}}\sum_{j=1}^m\alpha_j\sqrt{512T\log\left(6T\cdot 2^{j}\right)}
    &= \frac{12}{\sqrt{N}}\sum_{j=1}^m\int\limits_{2^{-j-1}}^{2^{-j}}C_\mathrm{loss}\sqrt{512T\log\left(6T\cdot 2^{j}\right)}~\mathrm{d}\alpha\\
    &\leq \frac{12}{\sqrt{N}}\sum_{j=1}^m\int\limits_{2^{-j-1}}^{2^{-j}}C_\mathrm{loss}\sqrt{512T\log\left(\frac{6T}{\alpha}\right)}~\mathrm{d}\alpha\\
    &=  \frac{12C_\mathrm{loss}}{\sqrt{N}}\int\limits_{2^{-(m+1)}}^{2^{-1}}\sqrt{512T\log\left(\frac{6T}{\alpha}\right)}~\mathrm{d}\alpha,\label{eq:chaining-second-term}
\end{align}
where, in the first inequality, we used that $2^j\leq\nicefrac{1}{\alpha}$ holds inside the limits of the integral.

Combining Eq.~\eqref{eq:chaining-first-term} and \eqref{eq:chaining-second-term}, we have proved that, for every $m\in\mathbb{N}$,
\begin{equation}
    \mathbb{E}_\sigma\left[\sup\limits_{\vec{\theta}} \frac{1}{N}\sum\limits_{i=1}^N \sigma_i \ell (\vec{\theta}; x_i, y_i)\right]
    \leq\alpha_m + \frac{12 C_\mathrm{loss}}{\sqrt{N}}\int\limits_{2^{-(m+1)}}^{2^{-1}}\sqrt{512T\log\left(\frac{6T}{\alpha}\right)}~\mathrm{d}\alpha.\label{eq:chaining-before-limit}
\end{equation}
If we take the limit $m\to\infty$, this becomes
\begin{align}
    \mathbb{E}_\sigma\left[\sup\limits_{\vec{\theta}} \frac{1}{N}\sum\limits_{i=1}^N \sigma_i \ell (\vec{\theta}; x_i, y_i)\right]
    &\leq
    \frac{12 C_\mathrm{loss}}{\sqrt{N}} \sqrt{512T} \int\limits_{0}^{\nicefrac{1}{2}}\sqrt{\log\left(\frac{6T}{\alpha}\right)}~\mathrm{d}\alpha\\
    &\leq \frac{12 C_\mathrm{loss}}{\sqrt{N}} \sqrt{512T} \cdot\int\limits_{0}^{\nicefrac{1}{2}}\sqrt{\log\left( 6T\right) + \log\left(\frac{1}{\alpha}\right)}~\mathrm{d}\alpha\\
    &\leq \frac{12 C_\mathrm{loss}}{\sqrt{N}} \sqrt{512T}\cdot\int\limits_{0}^{\nicefrac{1}{2}}\sqrt{\log\left( 6T\right)} + \sqrt{\log\left(\frac{1}{\alpha}\right)}~\mathrm{d}\alpha\\
    &= \frac{12 C_\mathrm{loss}}{\sqrt{N}} \sqrt{512T}\cdot \left(\frac{1}{2}\sqrt{\log\left( 6T\right)} + \int\limits_{0}^{\nicefrac{1}{2}}\sqrt{\log\left(\frac{1}{\alpha}\right)}~\mathrm{d}\alpha\right)\\
    &= \frac{12 C_\mathrm{loss}}{\sqrt{N}} \sqrt{512T}\cdot \left(\frac{1}{2}\sqrt{\log\left( 6T\right)} + \frac{1}{2}\sqrt{\log 2} - \frac{\sqrt{\pi}}{2}\operatorname{erf}(\sqrt{\log 2}) - \frac{\sqrt{\pi}}{2}\right),\label{eq:dudley-bound}
\end{align}
where we used the integral $\int\sqrt{\log\nicefrac{1}{x}}~\mathrm{d}x = x\sqrt{\log\nicefrac{1}{x}} - (\nicefrac{\sqrt{\pi}}{2}) \cdot \operatorname{erf}(\sqrt{\log\nicefrac{1}{x}})$, with the error function defined as $\operatorname{erf}(x)=\frac{2}{\sqrt{\pi}}\int_{0}^{x}\exp(-t^2)~\mathrm{d}t$.

We can now combine Eq.~\eqref{eq:Rademacher-generalization-bound} with \eqref{eq:dudley-bound} and obtain: With probability $\geq 1-\delta$ over the choice of training data of size $N$, we have
\begin{align}
    R(\vec{\theta}^\ast) - \hat{R}(\vec{\theta}^\ast)
    &\leq \frac{24 C_\mathrm{loss}}{\sqrt{N}} \sqrt{512T}\cdot \left(\frac{1}{2}\sqrt{\log\left( 6T\right)} + \frac{1}{2}\sqrt{\log 2} - \frac{\sqrt{\pi}}{2}\operatorname{erf}(\sqrt{\log 2}) - \frac{\sqrt{\pi}}{2}\right) + 3C_\mathrm{loss}\sqrt{\frac{2\log(\nicefrac{2}{\delta})}{N}} \\
    &\in \mathcal{O}\left(C_\mathrm{loss}\left(\sqrt{\frac{T\log \left(T\right)}{N}} + \sqrt{\frac{\log(\nicefrac{1}{\delta})}{N}} \right) \right),
\end{align}
which is the claimed prediction error bound.
\end{proof}

\begin{remark}
For simplicity, throughout the proof of Theorem \ref{thm:prediction-error-bound-qnn} we have treated $\vec{\theta}^\ast(S)$ as a deterministic function of $S$. However, the proof extends to the case in which the parameter setting $\vec{\theta}^\ast(S)$ is a random variable depending on $S$. Then, the generalization error bound would hold with high probability over the choice of the training data and over the internal randomness of the optimization procedure. This is the case for all our prediction error bounds and is important because quantum subroutines in QML procedures make them inherently probabilistic.
\end{remark}

\begin{remark}\label{rmk:improvement-chaining}
    A conceptual difference between the proof of Theorem~\ref{thm:prediction-error-bound-qnn-simplified} and that of Theorem~\ref{thm:prediction-error-bound-qnn}, which can also be seen as an underlying reason for why the latter leads to a tighter bound than the former, is the following: 
    To prove Theorem~\ref{thm:prediction-error-bound-qnn-simplified}, we used a $\sqrt{\nicefrac{T}{N}}$-covering net for the set of CPTP maps that the QMLM can implement. This can be seen as measuring the complexity of the QMLM at a single specific resolution, namely the resolution $\varepsilon=\sqrt{\nicefrac{T}{N}}$. 
    In contrast, the proof of Theorem~\ref{thm:prediction-error-bound-qnn} considers a complexity measure for the QMLM obtained by averaging over complexities (here, covering numbers) at multiple different resolutions. 
    Thus, from a high-level view, the chaining-based proof strategy for Theorem~\ref{thm:prediction-error-bound-qnn} improves upon the reasoning behind Theorem~\ref{thm:prediction-error-bound-qnn-simplified} by taking multiple resolutions into account.
\end{remark}

Theorem \ref{thm:prediction-error-bound-qnn} can be interpreted as follows: By taking the training data size $N$ to effectively scale linearly in the number of trainable elements $T$, we can ensure that a small training error also implies a small prediction error (with high probability).

In the following, we describe extensions of Theorem \ref{thm:prediction-error-bound-qnn} to different scenarios of interest, and then summarize these in a general ``mother theorem.''

\subsubsection{Extension to multiple copies and gate-sharing}\label{SbSbSct:prediction-error-multiple-copies}

In practice, one often employs quantum machine learning models that reuse the same parameterized gates multiple times, such as quantum convolutional neural networks (QCNNs) \cite{cong2019quantum}. In such a scenario, we speak of ``gate-sharing''. While the number of trainable elements in such models can still be large, only few of them can be trained \emph{independently}. As a first extension of Theorem \ref{thm:prediction-error-bound-qnn}, we show that the generalization performance of such a models is determined by the effective number of independently trainable elements.

\begin{corollary}\label{cor:weightieQNN}
Let $\mathcal{E}^{\mathrm{QMLM}}_{\vec{\theta}}(\cdot)$ be a QMLM with a fixed architecture consisting of $T$ independently parameterized 2-qubit CPTP maps, in which the $t^{\mathrm{th}}$ of these maps is used $M_t$ times, and an arbitrary number of non-trainable, global CPTP maps. Let $P$ be a probability distribution over input-output pairs. Suppose that, given training data $S=\{(x_i,y_i)\}_{i=1}^N$ of size $N$, our optimization yields the parameter setting $\vec{\theta}^\ast = \vec{\theta}^\ast(S)$.

Then, with probability at least $1-\delta$ over the choice of i.i.d.~training data $S$ of size $N$ according to $P$, 
\begin{equation}
    R(\vec{\theta}^\ast) - \hat{R}_S(\vec{\theta}^\ast)
    \in \mathcal{O}\left( C_\mathrm{loss}\left(\sqrt{\frac{T\log(T)}{N}} + \sqrt{\frac{\sum_{t=1}^{T} \log(M_t)}{N}} + \sqrt{\frac{\log(\nicefrac{1}{\delta})}{N}}\right)\right).
\end{equation}
\end{corollary}
\begin{proof}
The proof strategy is the same as for Theorem \ref{thm:prediction-error-bound-qnn}, we only change the covering number bound to be applied. Namely, instead of Theorem \ref{thm-metric-entropy-vqc-cptp}, we use Theorem \ref{thm-metric-entropy-cptp-multiple-copies}.

More precisely, we recall Eq.~\eqref{eq:Rademacher-generalization-bound}, which tells us: For every $\delta\in (0,1)$, with probability $\geq 1-\delta$ over the choice of training data of size $N$, we have
\begin{equation}
    R(\vec{\theta}^\ast) - \hat{R}(\vec{\theta}^\ast)
    \leq 2\mathbb{E}_\sigma\left[\sup\limits_{\vec{\theta}} \frac{1}{N}\sum\limits_{i=1}^N \sigma_i \ell (\vec{\theta}; x_i, y_i)\right] + 3C_\mathrm{loss}\sqrt{\frac{2\log(\nicefrac{2}{\delta})}{N}}.\label{eq:rademacher-generalization-bound-repetition}
\end{equation}
And, with the same chaining technique as detailed in the proof of Theorem \ref{thm:prediction-error-bound-qnn}, we can bound the above expectation over Rademacher random variables as
\begin{equation}
    \mathbb{E}_\sigma\left[\sup\limits_{\vec{\theta}} \frac{1}{N}\sum\limits_{i=1}^N \sigma_i \ell (\vec{\theta}; x_i, y_i)\right]
    \leq
    \frac{24 C_\mathrm{loss}}{\sqrt{N}} \int\limits_{0}^{\nicefrac{1}{2}}\sqrt{\log\left(\mathcal{N(\mathcal{CPTP}^{\mathrm{QMLM}},\norm{\cdot}_\diamond,\alpha)}\right)}~\mathrm{d}\alpha,\label{eq:dudley-bound-multiple-copies}
\end{equation}
where we used the notation from Theorem \ref{thm-metric-entropy-cptp-multiple-copies} for the set $\mathcal{CPTP}^{\mathrm{QMLM}}$ of $n$-qubit CPTP maps that can be implemented by the QMLM. Now, we use the metric entropy bound proved in Theorem \ref{thm-metric-entropy-cptp-multiple-copies} to further upper bound the integral as
\begin{align}
    \int\limits_{0}^{\nicefrac{1}{2}}\sqrt{\log\left(\mathcal{N(\mathcal{CPTP}_\mathcal{A},\norm{\cdot}_\diamond,\alpha)}\right)}~\mathrm{d}\alpha
    &\leq \int\limits_{0}^{\nicefrac{1}{2}}\sqrt{512 \left(T \log\left(\frac{6 T}{\alpha} \right) + \sum\limits_{t=1}^{T} \log(M_t)\right)}~\mathrm{d}\alpha\\
    &\leq \sqrt{512T\log(6T)}\int\limits_{0}^{\nicefrac{1}{2}}\sqrt{\log\left(\frac{1}{\alpha}\right)}~\mathrm{d}\alpha + \frac{\sqrt{512}}{2}\sqrt{\sum\limits_{t=1}^{T} \log(M_t)}.\label{eq:metric-entropy-integral-multiple-copies}
\end{align}
As $x\mapsto\log(\nicefrac{1}{x})$ has an integrable singularity at $x=0$, the integral in this expression is simply a multiplicative constant. Therefore, after plugging in the bound of Eq.~\eqref{eq:metric-entropy-integral-multiple-copies} into Eq.~\eqref{eq:dudley-bound-multiple-copies}, and then plugging the resulting bound on the Rademacher expectation into Eq.~\eqref{eq:rademacher-generalization-bound-repetition}, we obtain: For every $\delta\in (0,1)$, with probability $\geq 1-\delta$ over the choice of training data of size $N$, we have
\begin{equation}
    R(\vec{\theta}^\ast) - \hat{R}(\vec{\theta}^\ast)
    \in \mathcal{O}\left( C_\mathrm{loss}\left(\sqrt{\frac{T\log(T)}{N}} + \sqrt{\frac{\sum_{t=1}^{T} \log(M_t)}{N}} + \sqrt{\frac{\log(\nicefrac{1}{\delta})}{N}}\right)\right),
\end{equation}
the claimed generalization bound.
\end{proof}

A naive approach to the scenario of Corollary \ref{cor:weightieQNN} would be to upper-bound the metric entropy, and thus the prediction error, in terms of the total number of trainable elements in the QMLM. That, however, would lead to a significantly worse dependence of the prediction error bound on $M_t$, the numbers of uses, namely, of the form 
\begin{equation}
    C_\mathrm{loss}\sqrt{\frac{T\left(\sum_{t=1}^T M_t\right) \log\left(T\sum_{t=1}^T M_t \right)}{N}}.
\end{equation}
Our more careful analysis shows the tighter bound in which the numbers of uses $M_t$ only appear logarithmically, which is crucial for our application of the bound to quantum phase recognition with QCNNs (see Section \ref{Sct:application-quantum-phase-recognition}). This is possible because, even though there are in principle $T\sum_{t=1}^T M_t$ trainable elements in the quantum neural network, they are not trained independently. Rather, the parameter setting for the $t^{\textrm{th}}$ parameterized $2$-qubit CPTP map is reused $M_t$ times. This clearly shows that reusing parameters is, from a generalization perspective, preferable to having more independent parameters.

As a special case of Corollary \ref{cor:weightieQNN}, we obtain a prediction error bound for the scenario in which multiple copies of a QMLM (with the same parameter settings) are run in parallel:

\begin{corollary}\label{cor:MultipleCopiesQNN}
Let $\mathcal{E}^{\mathrm{QMLM}}_{\vec{\theta}}(\cdot)$ be a QMLM with a fixed architecture consisting of $T$ independently parameterized 2-qubit CPTP maps and an arbitrary number of non-trainable, global CPTP maps. By using $M$ copies of this model in parallel, we can consider loss functions of the form
\begin{equation}
    \ell (\vec{\theta}; x, y)
    = \Tr\left[O^{\mathrm{loss}}_{x,y} \left( (\mathcal{E}^{\mathrm{QMLM}}_{\vec{\theta}}\otimes\operatorname{id}) (\rho(x)) \right)^{\otimes M}\right],
\end{equation}
where $O^{\mathrm{loss}}_{x,y}$ are observables acting on the $M$-fold tensor product of an $n$-qubit system.
Let $P$ be a probability distribution over input-output pairs. Suppose that, given training data $S=\{(x_i,y_i)\}_{i=1}^N$ of size $N$, our optimization yields the parameter setting $\vec{\theta}^\ast = \vec{\theta}^\ast(S)$.

Then, with probability at least $1-\delta$ over the choice of i.i.d.~training data $S$ of size $N$ according to $P$, 
\begin{equation}
    R(\vec{\theta}^\ast) - \hat{R}_S(\vec{\theta}^\ast)
    \in \mathcal{O}\left(C_\mathrm{loss}\left( \sqrt{\frac{T\log(T)}{N}} + \sqrt{\frac{T\log(M) }{N}}+ \sqrt{\frac{\log(\nicefrac{1}{\delta})}{N}}\right)\right).
\end{equation}
\end{corollary}

Once we observe that $\sqrt{a+b}\leq \sqrt{a}+\sqrt{b}\leq \sqrt{2(a+b)}$ holds for all $a,b\geq 0$, we see that the upper bound in Corollary \ref{cor:MultipleCopiesQNN} can be rewritten as
\begin{equation}
    R(\vec{\theta}^\ast) - \hat{R}_S(\vec{\theta}^\ast)
    \in \mathcal{O}\left(C_\mathrm{loss}\left( \sqrt{\frac{T\log(TM)}{N}} + \sqrt{\frac{\log(\nicefrac{1}{\delta})}{N}}\right)\right).
\end{equation}
If $\delta$ is taken to be a fixed desired accuracy level and $C_\mathrm{loss}$ is also considered to be a fixed constant, this becomes the bound stated in Theorem \ref{thm:gate-sharing-QMLM} in the main text.

Corollary~\ref{cor:MultipleCopiesQNN} tells us that, even when using many copies of the QMLM, as the expressiveness of the corresponding function class grows at most logarithmically with the number of copies, we can still obtain a good prediction error. Note that, as in Corollary~\ref{cor:weightieQNN}, it is crucial that the same parameter setting is used for each copy.

We can also phrase the result as follows: We can upper-bound the prediction error incurred when using multiple copies of a quantum machine learning model for evaluating the loss by an expression that depends crucially on the number of trainable elements per copy, but only mildly on the number of copies.

\begin{remark}\label{rmk:multiple-copy-postprocessing}
    Cases of interest that Corollary~\ref{cor:MultipleCopiesQNN} describes are, e.g., the loss functions obtained by first performing (independent) product measurements on each of the $M$ copies, then taking an average (for a continuous target space) or a majority vote (for a discrete target space) of the obtained measurement outcomes, and finally post-processing this value by a classical loss function (such as the squared error loss).
    Such procedures arise naturally when taking into account that multiple shots are needed to accurately estimate the expectation value of an observable measured on the QMLM output.
    Note, however, that we cannot apply arbitrary procedures for post-processing single-copy measurement outcomes and still hope for a good prediction error. If $C_\mathrm{loss}$, which here is the supremum over the spectral norms of the $M$-copy observables $O^{\mathrm{loss}}_{x,y}$, scales badly (e.g., linearly) with $M$, the prediction error bound does so as well.
\end{remark}

\subsubsection{Extension to variable circuit architecture}\label{SbSbSct:variable-architecture}

For practical purposes, it might not be advantageous to fix the number of trainable elements in the QMLM, or even its structure more generally, in advance. Rather, one might also want to optimize over a discrete set of possible architectures, e.g., by growing or truncating the QMLM during the training phase. Therefore, in our second extension of Theorem \ref{thm:prediction-error-bound-qnn}, we provide a prediction error bound for such a variable structure scenario.

\begin{corollary}\label{cor:VariableGateNumberImproved}
Let $\mathcal{E}^{\mathrm{QMLM}}_{\vec{\alpha}}(\cdot)$ be a QMLM with a variable structure. Suppose that, for every $\tau\in\mathbb{N}$, there are at most $G_\tau\in\mathbb{N}$ allowed structures with exactly $\tau$ parameterized $2$-qubit CPTP maps and an arbitrary number of non-trainable, global CPTP maps.
Let $P$ be a probability distribution over input-output pairs. Suppose that, given training data $S=\{(x_i,y_i)\}_{i=1}^N$ of size $N$, our optimization yields a (data-dependent) structure with $T=T(S)$ parameterized $2$-qubit CPTP maps and the parameter setting $\vec{\alpha}^\ast = \vec{\alpha}^\ast(S)$.

Then, with probability at least $1-\delta$ over the choice of i.i.d.~training data $S$ of size $N$ according to $P$,
\begin{equation}
    R(\vec{\alpha}^\ast) - \hat{R}_S(\vec{\alpha}^\ast)
    \in \mathcal{O}\left(C_\mathrm{loss} \left(\sqrt{\frac{T\log\left(T\right)}{N}} + \sqrt{\frac{\log G_T}{N}} + \sqrt{\frac{\log(\nicefrac{1}{\delta})}{N}} \right)\right).
\end{equation}
\end{corollary}
\begin{proof}
By Theorem \ref{thm:prediction-error-bound-qnn}, for every $\tau\in \mathbb{N}$, for every one of the $G_\tau$ allowed structures with exactly $\tau$ parameterized $2$-qubit CPTP maps, with probability $\geq 1-\nicefrac{\delta}{2 G_\tau \tau^2}$ over the choice of i.i.d.~training data $S$ of size $N$ according to $P$, if $\vec{\theta}^\ast_\tau$ is a (continuous) parameter setting (for the $\tau$ parameterized maps) obtained through optimization upon input of training data $S$, we have the generalization error bound
\begin{equation}
    \E_{(x, y)\sim P} \left[\ell (\vec{\theta}^\ast_\tau; x, y)\right] -\frac{1}{N} \sum_{i=1}^{N} \ell (\vec{\theta}^\ast_\tau; x_i, y_i) 
    \in \mathcal{O}\left(C_\mathrm{loss} \left(\sqrt{\frac{\tau\log\left(\tau\right)}{N}}+ \sqrt{\frac{\log(\nicefrac{2G_\tau \tau^2}{\delta})}{N}} \right)\right).
\end{equation}
Thus, first taking a union bound over the $G_\tau$ structures with exactly $\tau$ parameterized $2$-qubit CPTP maps, and then a union bound over $\tau\in\mathbb{N}$, we see: With probability $\geq 1 - \sum_\tau \nicefrac{\delta}{2\tau^2}\geq 1-\delta$ over the choice of i.i.d.~training data $S$ of size $N$ according to $P$, if the optimization upon input of data $S$ outputs a QMLM architecture with $T=T(N)$ parameterized $2$-qubit CPTP maps and the (continuous and discrete) parameter setting $\vec{\alpha}^\ast = \vec{\alpha}^\ast(S)$
\begin{align}
    R(\vec{\alpha}^\ast) - \hat{R}_S(\vec{\alpha}^\ast)
    &\in \mathcal{O}\left(C_\mathrm{loss} \left(\sqrt{\frac{T\log\left(T\right) }{N}}+ \sqrt{\frac{\log(\nicefrac{2 G_T T^2}{\delta})}{N}} \right)\right)\\
    &\in \mathcal{O}\left(C_\mathrm{loss} \left(\sqrt{\frac{T\log\left(T\right) }{N}} + \sqrt{\frac{\log G_T}{N}} +  \sqrt{\frac{\log(\nicefrac{1}{\delta})}{N}} \right)\right),
\end{align}
as claimed.
\end{proof}

We can understand Corollary \ref{cor:VariableGateNumberImproved} as saying that the prediction error of a QMLM with a variable structure depends strongly (namely linearly) on $T$, the number of trainable elements that is used in the output structure of the optimization procedure, but only mildly (namely logarithmically) on $G_T$, the number of different possible structures with the same number of gates as the output structure. Note that the bound does not depend on all structures potentially considered during the optimization, but only on a subset of those.
In particular, if the number $T$ of trainable $2$-qubit maps is fixed in advance, optimizing not only over the parameter settings of the model but also over exponentially-in-$T$ many structures with $T$ trainable elements does not worsen the asymptotic behavior of the generalization error.

\subsubsection{Extension taking the optimization into account}\label{SbSbSct:optimization-dependent}

In our previous results, we have provided bounds on the generalization error that depended on the QMLM, e.g., via the number of trainable elements or the number of copies, or even on how many different architectures are admissible. So far, however, the bounds are agnostic w.r.t.~the procedure used to train the QMLM. In this section, we refine our approach to prove optimization-dependent generalization bounds, that explicitly take properties of the training process into account.

Take $\mathcal{M}_1, \ldots, \mathcal{M}_T$ to be sets of 2-qubit CPTP maps. Each set $\mathcal{M}_t$ denotes the space of 2-qubit CPTP maps that one permits for the $t^{\textrm{th}}$ trainable map during the training of the QMLM $\mathcal{E}^{\mathrm{QMLM}}_{\vec{\theta}}$. Hence, each $\mathcal{M}_t$ should be seen as the trainable space for a particular gate in the QMLM.
For example, $\mathcal{M}_t$ could be the space of all 2-qubit unitary channels, or the space of all tensor products of single-qubit CPTP maps.

As discussed in the proofs of Lemmas \ref{lem:Covering2QuditUnitaries} and \ref{lem:Covering2QuditCPTP}, as $\mathcal{CPTP}(\mathbb{C}^2\otimes\mathbb{C}^2)$ is compact, for every $1\leq t\leq T$, there exists a constant $c_t \geq 1$, depending, e.g., on the diameter and on the effective ambient dimension of $\mathcal{M}_t$, such that
\begin{equation}
    \log(\mathcal{N}(\mathcal{M}_t, \norm{\cdot}_\diamond, \epsilon)) \leq c_t \log\left(1 + \frac{1}{\epsilon}\right). \label{eq:covnet}
\end{equation}
Note that, as a worst-case estimate, we have $c_t\leq 1024$. This can be seen by arguing as in the proofs of Lemmas \ref{lem:Covering2QuditUnitaries} and \ref{lem:Covering2QuditCPTP}, with ambient dimension $512$ and diameter $2$, and then applying Bernoulli's inequality.

Given a fixed choice of parameters $\vec{\theta}$, and thereby a fixed choice $\mathcal{E}_1\in\mathcal{M}_1,\ldots,\mathcal{E}_T\in\mathcal{M}_T$ of the trainable $2$-qubit CPTP maps, the (fixed-architecture) QMLM implements the $n$-qubit CPTP map
\begin{equation}
    \mathcal{E}^\mathrm{QMLM}_{\vec{\theta}} = \mathcal{E}^{\mathrm{QMLM}} (\mathcal{E}_1, \ldots, \mathcal{E}_T) := \left(\prod_{t=1}^T \mathcal{F}_t \mathcal{E}_t\right) \mathcal{F}_0,
\end{equation}
where the $\mathcal{F}_t$, for $0\leq t\leq T$, are fixed (potentially global) CPTP maps.

Suppose that we begin the optimization of the QMLM from an initial point independent of the training data $S=\{(x_i, y_i)\}_{i=1}^N$, described by a parameter setting $\vec{\theta}_0$. We denote the choices for the $T$ trainable maps appearing in the initial CPTP map by
\begin{equation}
    \mathcal{E}^0_1 \in \mathcal{M}_1, \ldots, \mathcal{E}^0_T \in \mathcal{M}_T.
\end{equation}
After utilizing the training data for multiple rounds of optimization, the training of the QMLM finishes at a (data-dependent) point in $\mathcal{CPTP}\left((\mathbb{C})^{\otimes n} \right)$, described by a parameter vector $\vec{\theta}^\ast$, which we denote by the choice
\begin{equation}
    \mathcal{E}^\ast_1 \in \mathcal{M}_1, \ldots, \mathcal{E}^\ast_T \in \mathcal{M}_T,
\end{equation}
of trainable maps. Note that $\mathcal{E}^\ast_1, \ldots, \mathcal{E}^\ast_T$ depend on the training data $S$.
For each of the $T$ trainable local CPTP maps $M_t$, we denote the distance (measured w.r.t.~$\norm{\cdot}_\diamond$ between the initial and the final point of the training procedure by
\begin{equation}
    \Delta_t = \norm{\mathcal{E}^\ast_t - \mathcal{E}^0_t}_{\diamond}\leq 2, \textrm{ for } t = 1, \ldots, T.
\end{equation}
In the following Theorem, we provide a generalization guarantee for the resulting QMLM defined by the choice of trainable local maps $\mathcal{E}^\ast_1, \ldots, \mathcal{E}^\ast_T$ in terms of the optimization distances $\Delta_t$, the number $T$ of trainable maps, and the number $N$ of training data points. 

\begin{theorem}[Optimization-dependent prediction error bound for quantum machine learning]\label{thm:optimization-dependent-gen-bound}
Let $\mathcal{E}^{\mathrm{QMLM}}_{\vec{\theta}}(\cdot)$ be a QMLM with a fixed architecture consisting of $T$ parameterized $2$-qubit CPTP maps, in which the $t^{\mathrm{th}}$ of these maps is taken from $\mathcal{M}_t$, and an arbitrary number of non-trainable, global CPTP maps. Let $P$ be a probability distribution over input-output pairs. Suppose that, given training data $S=\{(x_i,y_i)\}_{i=1}^N$ of size $N$, the optimization procedure yields the parameter setting $\vec{\theta}^\ast=\vec{\theta}^\ast(S)$. As described above, denote by $\Delta_t=\Delta_t(S)$ the optimization distance (measured in diamond norm) of the $t^{\mathrm{th}}$ trainable map.

Then, with probability $\geq 1-\delta$ over the choice of i.i.d.~training data $S$ of size $N\geq 4$ according to $P$,
\begin{equation}
    R(\vec{\theta}^\ast)-\hat{R}_S(\vec{\theta}^\ast)
    \in \mathcal{O}\left(C_\mathrm{loss}\min\left\{\sqrt{\frac{K\max\limits_{1\leq t\leq T}c_t\log(K)}{N}} + \sqrt{\frac{K\log(T)}{N}} + \sum_{\substack{t=1 \\t\neq t_1,\ldots,t_K}}^T \Delta_{t} + \sqrt{\frac{\log(\nicefrac{1}{\delta})}{N}}\right\}\right),
\end{equation}
where the minimum is over all $K\in\{0,\ldots,T\}$ and choices of pairwise distinct $t_1,\ldots,t_K\in\{1,\ldots,T\}$.
\end{theorem}

The proof of Theorem \ref{thm:optimization-dependent-gen-bound} again hinges on a metric entropy bound, this time for the class of CPTP maps that can be reached by the QMLM under the optimization procedure. Hence, let us first prove the following theorem. 

\begin{theorem}\label{thm:metric-entropy-optimization-dependent}
Let $\mathcal{E}^{\mathrm{QMLM}}_{\vec{\theta}}(\cdot)$ be a QMLM with a fixed architecture consisting of $T$ parameterized $2$-qubit CPTP maps, in which the $t^{\mathrm{th}}$ of these maps is taken from $\mathcal{M}_t$, and an arbitrary number of non-trainable, global CPTP maps. Let $0\leq \Delta_1,\ldots,\Delta_T\leq 2$ be a sequence of distances for the trainable $2$-qubit CPTP maps. Let $\mathcal{CPTP}^{\mathrm{QMLM}}_{(\Delta_t)_t}\subset\mathcal{CPTP}\left((\mathbb{C}^2)^{\otimes n}\right)$ denote the set of $n$-qubit CPTP maps that can be implemented by the QMLM, under the additional restriction that the $t^{\mathrm{th}}$ trainable gate is at most diamond-distance $\Delta_t$ away from the fixed initial point $\mathcal{E}_t^{0}$.

Let $K\in\{0,\ldots,T\}$. Let $t_1,\ldots,t_K\in\{1,\ldots,T\}$ be pairwise distinct.
Then, for any $\varepsilon\in (0,1]$, if we write 
\begin{equation}
    \varepsilon_K:=\varepsilon + \sum_{\substack{t=1 \\t\neq t_1,\ldots,t_K}}^T \Delta_{t}
\end{equation}
there exists an $\varepsilon_K$-covering net $\mathcal{N}_{\varepsilon_K}$ of $\mathcal{CPTP}^{\mathrm{QMLM}}_{(\Delta_t)_t}$ w.r.t.~the diamond distance such that the logarithm of its size can be upper bounded as
\begin{equation}
    \log(\lvert\mathcal{N}_{\varepsilon_K}\rvert)
    \leq K\log(T) + K\max\limits_{1\leq t\leq T}c_t\log\left(1 + \frac{K}{\varepsilon} \right). \label{eq:sizegamma}
\end{equation}
\end{theorem}
\begin{proof}
By the assumptions on the structure of $\mathcal{CPTP}^{\mathrm{QMLM}}_{(\Delta_t)_t}$, there exist fixed, global CPTP maps $\mathcal{F}_0,\ldots,\mathcal{F}_T\in\mathcal{CPTP}\left((\mathbb{C}^2)^{\otimes n}\right)$ such that any $\mathcal{E}\in\mathcal{CPTP}^{\mathrm{QMLM}}_{(\Delta_t)_t}$ can be written as $\mathcal{E} = \mathcal{F}_T\mathcal{E}_T\mathcal{F}_{T-1}\ldots\mathcal{F}_1\mathcal{E}_1\mathcal{F}_0$ for some $2$-qubit CPTP maps $\mathcal{E}_t\in\mathcal{M}_t$, $1\leq t\leq T$ such that $\norm{\mathcal{E}_t - \mathcal{E}^0_t}_\diamond\leq \Delta_t$.

As discussed above, for each $1\leq t\leq T$, we can take $\mathcal{N}_t$ to be an $(\nicefrac{\varepsilon}{K})$-covering net for $\mathcal{M}_t$ w.r.t.~$\norm{\cdot}_\diamond$ whose cardinality satisfies $\log(\lvert\mathcal{N}_t\rvert)\leq c_t \log\left(1 + \nicefrac{K}{\epsilon}\right)$. We define $\mathcal{N}_{\varepsilon_K}$ to be the set of CPTP maps that can be implemented by $\mathcal{A}$ if exactly $K$ of the trainable $2$-qubit CPTP maps are taken from $\mathcal{N}_{t_1},\ldots,\mathcal{N}_{t_K}$, respectively, and the last $T-K$ trainable maps are left at the initial point of the optimization. That is, we define
\begin{equation}
    \mathcal{N}_{\varepsilon_K} 
    := \left\{\mathcal{E}(\tilde{\mathcal{E}}_1, \ldots, \tilde{\mathcal{E}}_T) = \left(\prod_{t=1}^T \mathcal{F}_{t} \tilde{\mathcal{E}}_{t}\right) \mathcal{F}_0 ~\Bigg\vert~ \lvert\{1\leq t\leq T~|~\tilde{\mathcal{E}}_t \neq M^0_t\}\vert = K\textrm{ and  }\tilde{\mathcal{E}}_{t}\in\mathcal{N}_t\textrm{ whenever } \tilde{\mathcal{E}}_t \neq \mathcal{E}^0_t\right\}.
\end{equation}
Using the subadditivity of the distance induced by the diamond norm (Lemma \ref{lem:subadditivity-diamond}), it is easy to see that, for every $\mathcal{E}=\mathcal{E}(\mathcal{E}_1,\ldots,\mathcal{E}_T)\in\mathcal{CPTP}^{\mathrm{QMLM}}_{(\Delta_t)_t}$, there exists an $\tilde{\mathcal{E}}=\mathcal{E}(\tilde{\mathcal{E}}_1,\ldots,\tilde{\mathcal{E}}_T)\in\mathcal{N}_{\varepsilon_K}$ s.t.
\begin{equation}
    \norm{\mathcal{E} - \tilde{\mathcal{E}}}_\diamond
    \leq K\cdot\frac{\varepsilon}{K} + \sum_{\substack{t=1 \\t\neq t_1,\ldots,t_K}}^T \Delta_{t}
    =\varepsilon_K.
\end{equation}
Thus, $\mathcal{N}_{\varepsilon_K}$ is indeed an $\varepsilon_K$-covering net for $\mathcal{CPTP}^{\mathrm{QMLM}}_{(\Delta_t)_t}$, as claimed.

It remains to observe that, by definition of $\mathcal{N}_{\varepsilon_K}$, we have
\begin{align}
    \log(\lvert\mathcal{N}_{\varepsilon_K}\rvert) 
    &= \log\left(\binom{T}{K}\cdot \prod_{k=1}^K |\mathcal{N}_t|\right)\\
    &\leq K\log(T) + \left(\sum_{k=1}^K c_t\right)\log\left(1 + \frac{K}{\epsilon}\right)\\
    &\leq K\log(T) + K\max\limits_{1\leq t\leq T}c_t\log\left(1 + \frac{K}{\varepsilon} \right),
\end{align}
as claimed.
\end{proof}

Armed with this metric entropy bound, we can now prove Theorem \ref{thm:optimization-dependent-gen-bound}.

\begin{proof}[Proof of Theorem \ref{thm:optimization-dependent-gen-bound}]
Starting from the metric entropy bound of Theorem \ref{thm:metric-entropy-optimization-dependent}, we again argue as in the proof of Theorem \ref{thm:prediction-error-bound-qnn}. Recall that the first step of said proof was to establish Eq.~\eqref{eq:Rademacher-generalization-bound}. This was then followed in a second step by upper-bounding the obtained expression using a covering number integral.
The first step, leading to Eq.~\eqref{eq:Rademacher-generalization-bound}, is also valid in the scenario of this Theorem. That is, we again have that, with probability $\geq 1-\delta$ over the choice of training data of size $N$,
\begin{equation}
    R(\vec{\theta}^\ast) - \hat{R}(\vec{\theta}^\ast)
    \leq 2\mathbb{E}_\sigma\left[\sup\limits_{\vec{\theta}} \frac{1}{N}\sum\limits_{i=1}^N \sigma_i \ell (\vec{\theta}; x_i, y_i)\right] + 3C_\mathrm{loss}\sqrt{\frac{2\log(\nicefrac{2}{\delta})}{N}}.
\end{equation}
However, we have to change the second step. To this end, we first observe that, by a reasoning completely analogous the one leading up to Eq.~\eqref{eq:chaining-before-limit}, for every $m\in\mathbb{N}_0$,
\begin{equation}\label{eq:chaining-intermediate-optimization}
    \mathbb{E}_\sigma\left[\sup\limits_{\vec{\theta}} \frac{1}{N}\sum\limits_{i=1}^N \sigma_i \ell (\vec{\theta}; x_i, y_i)\right]
    \leq C_\mathrm{loss}\cdot 2^{-m} + \frac{12 C_\mathrm{loss}}{\sqrt{N}}\int\limits_{2^{-(m+1)}}^{2^{-1}}\sqrt{\log\left(\mathcal{N}(\mathcal{CPTP}^{\mathrm{QMLM}}_{(\Delta_t)_t},\norm{\cdot}_\diamond,\alpha\right)}~\mathrm{d}\alpha,
\end{equation}
where we used the notation from Theorem \ref{thm:metric-entropy-optimization-dependent}.
Fix a $K\in\{0,\ldots,T\}$ and pairwise distinct $t_1,\ldots,t_K\in\{1,\ldots,T\}$ such that
\begin{equation}
    \tilde{\Delta}
    \coloneqq\sum_{\substack{t=1 \\t\neq t_1,\ldots,t_K}}^T \Delta_{t}
    <\frac{1}{2}
\end{equation}
and take $m\in\mathbb{N}_0$ such that $\tilde{\Delta}< 2^{-(m+1)} < 2\tilde{\Delta}$. Then in particular $2^{-m}\leq 4\tilde{\Delta}$ and we can further upper bound the expression in Eq.~\eqref{eq:chaining-intermediate-optimization} as 
\begin{align}
    &\mathbb{E}_\sigma\left[\sup\limits_{\vec{\theta}} \frac{1}{N}\sum\limits_{i=1}^N \sigma_i \ell (\vec{\theta}; x_i, y_i)\right]\\
    &\leq  4C_\mathrm{loss}\tilde{\Delta} + \frac{12 C_\mathrm{loss}}{\sqrt{N}}\int\limits_{\tilde{\Delta}}^{2^{-1}}\sqrt{\log\left(\mathcal{N}(\mathcal{CPTP}^{\mathrm{QMLM}}_{(\Delta_t)_t},\norm{\cdot}_\diamond,\alpha\right)}~\mathrm{d}\alpha\\
    &=  4C_\mathrm{loss}\tilde{\Delta}+ \frac{12 C_\mathrm{loss}}{\sqrt{N}}\int\limits_{0}^{2^{-1} - \tilde{\Delta}}\sqrt{\log\left(\mathcal{N}(\mathcal{CPTP}^{\mathrm{QMLM}}_{(\Delta_t)_t},\norm{\cdot}_\diamond,\alpha + \tilde{\Delta}\right)}~\mathrm{d}\alpha.
\end{align}
At this point, we can apply the metric entropy bound from Theorem \ref{thm:metric-entropy-optimization-dependent} to obtain
\begin{align}
    \int\limits_{0}^{2^{-1} - \tilde{\Delta}}\sqrt{\log\left(\mathcal{N}(\mathcal{CPTP}^{\mathrm{QMLM}}_{(\Delta_t)_t},\norm{\cdot}_\diamond,\alpha + \tilde{\Delta}\right)}~\mathrm{d}\alpha
    &\leq \int\limits_{0}^{2^{-1} - \tilde{\Delta}}\sqrt{K\max\limits_{1\leq t\leq T}c_t\log\left(1 + \frac{K}{\alpha} \right)}~\mathrm{d}\alpha\\
    &\leq \sqrt{K\max\limits_{1\leq t\leq T}c_t}\int\limits_{0}^{2^{-1} - \tilde{\Delta}}\sqrt{\log\left(\frac{2K}{\alpha} \right)}~\mathrm{d}\alpha\\
    &\leq \mathcal{O}\left(\sqrt{K\max\limits_{1\leq t\leq T}c_t\log(K)}\right).
\end{align}
Altogether, so far we have shown that, for any fixed choice of $K\in\{0,\ldots,T\}$ and of pairwise distinct $t_1,\ldots,t_K\in\{1,\ldots,T\}$ such that $\sum\limits_{\substack{t=1 \\t\neq t_1,\ldots,t_K}}^T \Delta_{t}<\frac{1}{2}$, with probability $\geq 1-\delta$ over the choice of training data of size $N$, we have
\begin{equation}
    R(\vec{\theta}^\ast) - \hat{R}(\vec{\theta}^\ast)
    \in \mathcal{O}\left(C_\mathrm{loss}\left(\sqrt{\frac{K\max\limits_{1\leq t\leq T}c_t\log(K)}{N}} + \sum_{\substack{t=1 \\t\neq t_1,\ldots,t_K}}^T \Delta_{t} + \sqrt{\frac{2\log(\nicefrac{1}{\delta})}{N}}\right)\right).
\end{equation}
After a union bound over the at most $\binom{T}{K}\leq T^K$ different choices of pairwise distinct $t_1,\ldots,t_K\in\{1,\ldots,T\}$ (with $\sum\limits_{\substack{t=1 \\t\neq t_1,\ldots,t_K}}^T \Delta_{t}<\frac{1}{2}$), we see that, with probability $\geq 1-\delta$ over the choice of training data of size $N$, we have
\begin{align}
    R(\vec{\theta}^\ast) - \hat{R}(\vec{\theta}^\ast)
    &\in \mathcal{O}\left(C_\mathrm{loss}\left(\sqrt{\frac{K\max\limits_{1\leq t\leq T}c_t\log(K)}{N}} + \sum_{\substack{t=1 \\t\neq t_1,\ldots,t_K}}^T \Delta_{t} + \sqrt{\frac{2\log(\nicefrac{T^K}{\delta})}{N}}\right)\right)\\
    &\in \mathcal{O}\left(C_\mathrm{loss}\min\left\{\sqrt{\frac{K\max\limits_{1\leq t\leq T}c_t\log(K)}{N}} + \sqrt{\frac{K\log(T)}{N}} + \sum_{\substack{t=1 \\t\neq t_1,\ldots,t_K}}^T \Delta_{t} + \sqrt{\frac{2\log(\nicefrac{1}{\delta})}{N}}\right\}\right),
\end{align}
where $K\in\{0,\ldots,T\}$ is still fixed and the minimum is over all choices of pairwise distinct $t_1,\ldots,t_K\in\{1,\ldots,T\}$.

Finally, we can take a union bound over at most $T+1$ different values of $K$ and obtain that, with probability $\geq 1-\delta$ over the choice of training data of size $N$, we have
\begin{align}
    R(\vec{\theta}^\ast) - \hat{R}(\vec{\theta}^\ast)
    &\in \mathcal{O}\left(C_\mathrm{loss}\min\left\{\sqrt{\frac{K\max\limits_{1\leq t\leq T}c_t\log(K)}{N}} + \sqrt{\frac{K\log(T)}{N}} + \sum_{\substack{t=1 \\t\neq t_1,\ldots,t_K}}^T \Delta_{t} + \sqrt{\frac{2\log(\nicefrac{T}{\delta})}{N}}\right\}\right)\\
    &\in \mathcal{O}\left(C_\mathrm{loss}\min\left\{\sqrt{\frac{K\max\limits_{1\leq t\leq T}c_t\log(K)}{N}}+ \sqrt{\frac{K\log(T)}{N}} + \sum_{\substack{t=1 \\t\neq t_1,\ldots,t_K}}^T \Delta_{t} + \sqrt{\frac{2\log(\nicefrac{1}{\delta})}{N}}\right\}\right),
\end{align}
where the minimum is over all values of $K$ and over all choices of pairwise distinct $t_1,\ldots,t_K\in\{1,\ldots,T\}$.
\end{proof}

If, in the generalization error bound of Theorem \ref{thm:optimization-dependent-gen-bound}, we disregard the potential improvements gained from the $\mathcal{M}_t$-dependent constants $c_t$ and instead replace all of them by their worst-case value $1024$, we can simplify the bound to
\begin{equation}
    R(\vec{\theta}^\ast)-\hat{R}_S(\vec{\theta}^\ast)
    \in \mathcal{O}\left(C_\mathrm{loss}\min\left\{\sqrt{\frac{K\log(T)}{N}} + \sum_{\substack{t=1 \\t\neq t_1,\ldots,t_K}}^T \Delta_{t} + \sqrt{\frac{\log(\nicefrac{1}{\delta})}{N}}\right\}\right),
\end{equation}
because $K\log(K)\leq K\log(T)$ for all $K=1,\ldots,T$. Moreover, instead of taking a minimum over all such $K$ and over all choices of pairwise distinct $t_1,\ldots,t_K\in\{1,\ldots,T\}$, we can take the minimum only over $K$, and fix the choice $t_k=k$ to obtain
\begin{equation}
    R(\vec{\theta}^\ast)-\hat{R}_S(\vec{\theta}^\ast)
    \in \mathcal{O}\left(C_\mathrm{loss}\min\limits_{K=1,\ldots,T}\left\{\sqrt{\frac{K\log(T)}{N}} + \sum_{t=K+1}^T \Delta_{t} + \sqrt{\frac{\log(\nicefrac{1}{\delta})}{N}}\right\}\right).
\end{equation}
If we again fix a confidence level $\delta$ and consider $C_\mathrm{loss}$ as a fixed constant of the problem, this becomes the bound given in Theorem \ref{thm:gate-sharing-QMLM-optimization} for the case $M=1$. (The case for general $M$ follows from our ``mother theorem,'' see Subsection \ref{SbSbSct:mother-theorem}.)

We can clearly see that, if the optimization has only made substantial changes to few trainable maps, then the generalization error bound in Theorem \ref{thm:optimization-dependent-gen-bound} is dominated by the maps that have undergone more significant changes during optimization. The number of such parameterized maps could be much smaller than the overall number of the trainable CPTP maps $T$. Therefore, this optimization-dependent generalization error bound can significantly outperform the previous bounds, which did not take the optimization procedure into account, w.r.t.~the dependence on $T$.

One consequence of this theorem is that a good choice of initialization for the optimization of a QMLM can not only serve to improve the cost of the optimization itself, but it can also help the generalization behavior. Namely, a particularly good choice of initialization, potentially found through pretraining on an independent data set, can lead to an optimization procedure that does not have to deviate too far from the initialization w.r.t.~some of the trainable maps, which, according to our bound, will be advantageous for generalization.

A second implication of this result for what to take into account in designing an optimization procedure for training a QMLM is the following: Making large steps only on few trainable gates and only negligibly small steps on the remaining ones is, from a generalization perspective, preferable to making steps of comparable, non-negligible sizes on many (or even all) of the trainable gates.

\begin{remark}
We note that in the proof of Theorem \ref{thm:optimization-dependent-gen-bound}, it was not necessary that the fixed CPTP maps $\mathcal{E}^0_1 \in \mathcal{M}_1, \ldots, \mathcal{E}^0_T \in \mathcal{M}_T$ were given as the initialization of the optimization procedure. In fact, we can take these maps to be any fixed ``reference points'' w.r.t.~which we measure distances. Th proof then works without changes, as long as the reference maps are indeed fixed in advance, independently of the training data.
\end{remark}

\subsubsection{Extension to unbiased estimates of measurement statistics}\label{SbSbSct:estimated-training-error}

In practice, we cannot obtain the exact value of $\Tr\left[O^{\mathrm{loss}}_{x_i, y_i} (\mathcal{E}^{\mathrm{QMLM}}_{\vec{\theta}}\otimes \operatorname{id})(\rho(x_i)) \right]$ for a training example $(x_i,y_i)$ if we only perform finitely many measurements. Instead, as a proxy for the training error, we consider an unbiased estimator: For $1\leq \sigma\leq \sigma_{\mathrm{est}}$, with $\sigma_{\mathrm{est}}\in\mathbb{N}$ fixed, we independently pick $i_\sigma$ uniformly at random from $\{1,\ldots,N\}$ and measure the observable $O^{\mathrm{loss}}_{x_{i_\sigma}, y_{i_\sigma}}$ on the output state $\mathcal{E}^{\mathrm{QMLM}}_{\vec{\theta}}(\rho(x_{i_\sigma}))$ to yield a single measurement outcome $o^{\mathrm{loss}}_{\vec{\theta},\sigma} \in \left[-\norm{O^{\mathrm{loss}}_{x_{i_{\sigma}}, y_{i_{\sigma}}}}, \norm{O^{\mathrm{loss}}_{x_{i_{\sigma}}, y_{i_{\sigma}}}}\right]$. As $\mathbb{E}\left[o^{\mathrm{loss}}_{\vec{\theta},\sigma}\right] = \frac{1}{N} \sum_{i=1}^N \Tr\left[O^{\mathrm{loss}}_{x_i, y_i} (\mathcal{E}^{\mathrm{QMLM}}_{\vec{\theta}}\otimes \operatorname{id})(\rho(x_i)) \right]$, where the expectation is taken w.r.t.~the sampling of $i_{\sigma}$ and the randomness in obtaining the measurement outcome. This yields a finite sequence of observations
\begin{equation}
    o^{\mathrm{loss}}_{\vec{\theta}, 1}, \ldots, o^{\mathrm{loss}}_{\vec{\theta}, \sigma_{\mathrm{est}}}, \,\,\, \text{with} \,\,\,  \frac{1}{\sigma_{\mathrm{est}}} \sum_{\sigma=1}^{\sigma_{\mathrm{est}}}  o^{\mathrm{loss}}_{\vec{\theta}, \sigma} 
    \approx \frac{1}{N} \sum_{i=1}^N \Tr\left[O^{\mathrm{loss}}_{x_i, y_i} (\mathcal{E}^{\mathrm{QMLM}}_{\vec{\theta}}\otimes \operatorname{id})(\rho(x_i)) \right].
\end{equation}
In this scenario, where we only obtain a noisy estimate of the training error from $\sigma_{\mathrm{est}}$ measurements, the prediction performance guarantee takes the following form:

\begin{corollary}\label{cor:EstimatedTrainingError}
Let $\mathcal{E}^{\mathrm{QMLM}}_{\vec{\theta}}(\cdot)$ be a quantum machine learning model with a fixed architecture consisting of $T$ parameterized 2-qubit CPTP maps. Let $P$ be a probability distribution over input-output pairs.
Suppose that, given training data $S=\{(x_i,y_i)\}_{i=1}^N$ of size $N$, our optimization yields the parameter setting $\vec{\theta}^\ast = \vec{\theta}^\ast(S)$.

Then, with probability at least $1-\delta$ over the choice of i.i.d.~training data $S$ of size $N$ according to $P$, over the sampling of $i_1,\ldots,i_{\sigma_{\mathrm{est}}}$, and over the $\sigma_{\mathrm{est}}$ obtained measurement outcomes, 
\begin{equation}
     \E_{x, y} \ell (\vec{\theta}^*; x, y) - \frac{1}{\sigma_{\mathrm{est}}} \sum_{\sigma=1}^{\sigma_{\mathrm{est}}}  o^{\mathrm{loss}}_{\vec{\theta}^\ast, \sigma} 
     \in \mathcal{O}\left(C_\mathrm{loss}\left( \sqrt{\frac{T\log(T)}{N}}+ \sqrt{\frac{\log(\nicefrac{1}{\delta})}{N}} + \sqrt{\frac{\log(\nicefrac{1}{\delta})}{\sigma_{\mathrm{est}}}} \right)\right).
\end{equation}
\end{corollary}
\begin{proof}
We first insert a zero in terms of the empirical risk as follows:
\begin{equation}
    \E_{x, y} \ell (\vec{\theta}^*; x, y) - \frac{1}{\sigma_{\mathrm{est}}} \sum_{\sigma=1}^{\sigma_{\mathrm{est}}}  o^{\mathrm{loss}}_{\vec{\theta}^\ast, \sigma}
    = \left(\E_{x, y} \ell (\vec{\theta}^*; x, y) - \frac{1}{N} \sum_{i=1}^{N} \ell (\vec{\theta}^\ast; x_i, y_i) \right) + \left( \frac{1}{N} \sum_{i=1}^{N} \ell (\vec{\theta}^\ast; x_i, y_i)- \frac{1}{\sigma_{\mathrm{est}}} \sum_{\sigma=1}^{\sigma_{\mathrm{est}}}  o^{\mathrm{loss}}_{\vec{\theta}^\ast, \sigma} \right).
\end{equation}
By Theorem \ref{thm:prediction-error-bound-qnn}, with probability at least $1-\frac{\delta}{2}$ over the choice of the training data, the first term on the right-hand side (which is independent of the subsampling and of the obtained measurement outcomes) is bounded as $\in \mathcal{O}\left(C_\mathrm{loss}\left( \sqrt{\frac{T\log(T)}{N}} + \sqrt{\frac{\log(\nicefrac{1}{\delta})}{N}} \right)\right)$. By Hoeffding's inequality, for any fixed $S$ and $\vec{\theta}^\ast$, with probability at least $1-\frac{\delta}{2}$ over the sampling of $i_1,\ldots,i_{\sigma_{\mathrm{est}}}$ and over the $\sigma_{\mathrm{est}}$ obtained measurement outcomes, the second term on the right-hand side is $\leq C_\mathrm{loss} \sqrt{\frac{2\log(\nicefrac{2}{\delta})}{\sigma_{\mathrm{est}}}}$. Therefore, we also have
\begin{align}
    &\hphantom{=}\mathbb{P}\left[\frac{1}{N} \sum_{i=1}^{N} \ell (\vec{\theta}^\ast; x_i, y_i)- \frac{1}{\sigma_{\mathrm{est}}} \sum_{\sigma=1}^{\sigma_{\mathrm{est}}}  o^{\mathrm{loss}}_{\vec{\theta}^\ast, \sigma} > C_\mathrm{loss} \sqrt{\frac{2\log(\nicefrac{2}{\delta})}{\sigma_{\mathrm{est}}}}\right]\\
    &= \mathbb{E}_{S,\vec{\theta}^\ast}\left[\mathbb{P}\left[\frac{1}{N} \sum_{i=1}^{N} \ell (\vec{\theta}^\ast; x_i, y_i)- \frac{1}{\sigma_{\mathrm{est}}} \sum_{\sigma=1}^{\sigma_{\mathrm{est}}}  o^{\mathrm{loss}}_{\vec{\theta}^\ast, \sigma} > C_\mathrm{loss} \sqrt{\frac{2\log(\nicefrac{2}{\delta})}{\sigma_{\mathrm{est}}}}~\Big\rvert~ S,\vec{\theta}^\ast\right]\right]\\
    &\leq \mathbb{E}_{S,\vec{\theta}^\ast}\left[\frac{\delta}{2}\right]\\
    &= \frac{\delta}{2}.
\end{align}
Now, the statement of the Corollary follows via a union bound.
\end{proof}

This shows that we do not need to perform a disproportionately large number of measurements to guarantee that the estimated training error is indeed a good proxy for the prediction error. 
It suffices to choose $\sigma_{\mathrm{est}}$ to be roughly $\nicefrac{N}{T\log(T)}$, along with $N$ being sufficiently larger than $T\log(T)$, to guarantee that the prediction error will not be much higher than the approximate (observed) training error.

\subsubsection{Mother theorem}\label{SbSbSct:mother-theorem}

We can summarize all the previously discussed extensions of Theorem \ref{thm:prediction-error-bound-qnn} in Theorem \ref{thm-mother}, which we restate here for convenience:

\begin{thm-mother}[Mother Theorem]
Let $\mathcal{E}^{\mathrm{QMLM}}_{\vec{\alpha}}(\cdot)$ be a QMLM with a variable structure. Suppose that, for every $k\in\mathbb{N}$, there are at most $G_\tau\in\mathbb{N}$ allowed structures with exactly $\tau$ parameterized $2$-qubit CPTP maps, in which the $t^{\mathrm{th}}$ of these maps is taken from $\mathcal{M}_t$ and used $M_t$ times, and an arbitrary number of non-trainable, global CPTP maps. Also, for each $t\in\mathbb{N}$, let $\mathcal{E}^0_t\in\mathcal{CPTP}\left((\mathbb{C})^{\otimes 2}\right)$ be a fixed reference CPTP map.
Let $P$ be a probability distribution over input-output pairs. Suppose that, given training data $S=\{(x_i,y_i)\}_{i=1}^N$ of size $N$, our optimization of the QMLM over structures and parameters w.r.t.~the loss function $\ell(\vec{\alpha};x_i,y_i)=\Tr\left[O^{\mathrm{loss}}_{x_i,y_i}(\mathcal{E}^{\mathrm{QMLM}}_{\vec{\alpha}}\otimes \operatorname{id})(\rho(x_i))\right]$ yields a (data-dependent) structure with $T=T(N)$ independently parameterized $2$-qubit CPTP maps, in which the $t^{\mathrm{th}}$ of these maps is taken from $\mathcal{M}_t$ and used $M_t$ times, as well as the parameter setting $\vec{\alpha}^\ast = \vec{\alpha}^\ast(S)$.

Then, with probability at least $1-\delta$ over the choice of i.i.d.~training data $S$ of size $N$ according to $P$,
\small
\begin{align}
    &R(\vec{\alpha}^\ast) - \hat{R}_S(\vec{\alpha}^\ast)\nonumber\\
    &\in\mathcal{O}\left(C_\mathrm{loss}\min\left\{\sqrt{\frac{K\max\limits_{1\leq t\leq T}c_t\log(K)}{N}} + \sqrt{\frac{K\log(T)}{N}} + \sqrt{\frac{K\max\limits_{1\leq t\leq T}c_t\log(M_t)}{N}} + \sum_{\substack{t=1 \\t\neq t_1,\ldots,t_K}}^T M_t \Delta_{t} + \sqrt{\frac{\log(G_T)}{N}} +\sqrt{\frac{\log(\nicefrac{1}{\delta})}{N}}\right\}\right),
\end{align}
\normalsize
where $\Delta^T_1,\ldots,\Delta^T_T$ denote the (data-dependent) distance between the trainable maps in the output QMLM to the respective reference maps $\mathcal{E}^0_1,\ldots,\mathcal{E}^0_T$, $C_\mathrm{loss}=\sup_{x,y}\norm{O^{\mathrm{loss}}_{x, y}}$ is the maximum (absolute) value attainable by the loss function, and the minimum is over all $K\in\{0,\ldots,T\}$ and choices of pairwise distinct $t_1,\ldots,t_K\in\{1,\ldots,T\}$.

Moreover, if the loss is not evaluated exactly, but an unbiased estimator is built from $\sigma_{\mathrm{est}}$ subsampled training data points (as in Section \ref{SbSbSct:estimated-training-error}), we only incur an additional error of $\mathcal{O}\left(\sqrt{\nicefrac{\log(\nicefrac{1}{\delta})}{\sigma_{\mathrm{est}}}}\right)$.
\end{thm-mother}

\begin{proof}
To prove this most general version of our results, we combine the previous results and proof strategies.
First, fix $\tau\in\mathbb{N}$ and one of the $G_\tau$ admissible QMLM architectures with exactly $\tau$ trainable $2$-qubit CPTP maps, in which the $t^{\mathrm{th}}$ of these maps is taken from $\mathcal{M}_t$ and used $M_t$ times. With the same strategy as in the proof of Theorem \ref{thm:metric-entropy-optimization-dependent}, if we take $\mathcal{N}_t$ to be a $(\nicefrac{\varepsilon}{K M_t})$-covering net for $\mathcal{M}_t$ w.r.t.~$\norm{\cdot}_\diamond$, and consider the set of $n$-qubit CPTP maps obtained from the QMLM if exactly $K$ of the $\tau$ independently trainable $2$-qubit CPTP maps are taken from the respective $\mathcal{N}_t$, and the remaining $\tau-K$ maps are left at the corresponding reference map, this gives us an $\varepsilon_K$-covering net $\mathcal{N}_\varepsilon$ of the class of $n$-qubit CPTP maps that the QMLM architecture can implement, where
\begin{equation}
    \varepsilon_K
    \coloneqq \varepsilon + \sum_{\substack{t=1 \\t\neq t_1,\ldots,t_K}}^T M_t\Delta_{t}.
\end{equation}
This $\varepsilon_K$-covering net can be taken to have cardinality bounded as
\begin{equation}
    \log(\lvert\mathcal{N}_\varepsilon\rvert)\leq K\log(\tau) + K\max_{1\leq t\leq \tau}c_t \log\left(1 + \nicefrac{K M_t}{\varepsilon}\right).
\end{equation}
If we use this metric entropy bound for the chaining argument presented in the proof of Theorem \ref{thm:optimization-dependent-gen-bound}, we can show that, with probability $\geq 1-\nicefrac{\delta}{2G_\tau \tau^2}$ over the choice of i.i.d.~training data $S$ of size $N$, if $\vec{\theta}^\ast_\tau$ is a (continuous) parameter setting for the $\tau$ parameterized maps obtained through optimization upon data $S$, we have 
\begin{equation}
    R(\vec{\theta}^\ast_\tau)-\hat{R}_S(\vec{\theta}^\ast_\tau)
    \in \mathcal{O}\left(C_\mathrm{loss}\min\left\{\sqrt{\frac{K\max\limits_{1\leq t\leq \tau}c_t\log(KM_t)}{N}} + \sqrt{\frac{K\log(\tau)}{N}} + \sum_{\substack{t=1 \\t\neq t_1,\ldots,t_K}}^T M_t \Delta_{t} + \sqrt{\frac{\log(\nicefrac{2G_\tau \tau^2}{\delta})}{N}}\right\}\right),
\end{equation}
where the minimum is over $K\in\{0,\ldots,\tau\}$ and over choices of pairwise distinct $t_1,\ldots,t_K\in\{1,\ldots,\tau\}$.

We now first take a union bound over the $G_\tau$ admissible structures and then another union bound over $\tau\in\mathbb{N}$ to obtain: With probability $\geq 1-\delta$ over the choice of i.i.d.~training data $S$ of size $N$, if the optimization upon input of $S$ outputs a QMLM architecture with $T=T(N)$ parameterized $2$-qubit CPTP maps and the (discrete and continuous) parameter setting $\vec{\alpha}^\ast = \vec{\alpha}^\ast (S)$, then we have the generalization error bound
\small
\begin{align}
    &R(\vec{\alpha}^\ast)-\hat{R}_S(\vec{\alpha}^\ast)\\
    &\in \mathcal{O}\left(C_\mathrm{loss}\min\left\{\sqrt{\frac{K\max\limits_{1\leq t\leq T}c_t\log(KM_t)}{N}} + \sqrt{\frac{K\log(T)}{N}} + \sum_{\substack{t=1 \\t\neq t_1,\ldots,t_K}}^T M_t \Delta_{t} + \sqrt{\frac{\log(\nicefrac{2G_T T^2}{\delta})}{N}}\right\}\right)\\
    &\in \mathcal{O}\left(C_\mathrm{loss}\min\left\{\sqrt{\frac{K\max\limits_{1\leq t\leq T}c_t\log(K)}{N}} + \sqrt{\frac{K\log(T)}{N}} + \sqrt{\frac{K\max\limits_{1\leq t\leq T}c_t\log(M_t)}{N}} + \sum_{\substack{t=1 \\t\neq t_1,\ldots,t_K}}^T M_t \Delta_{t} + \sqrt{\frac{\log(G_T)}{N}} +\sqrt{\frac{\log(\nicefrac{1}{\delta})}{N}}\right\}\right),
\end{align}
\normalsize
with the minimum as claimed.

To understand the added error in the case in which an unbiased estimate of the empirical risk is used, we now repeat the analysis given in Section \ref{SbSbSct:estimated-training-error}, but use the generalization bound just established instead of the one from Theorem \ref{thm:prediction-error-bound-qnn}, and obtain the claimed bound.
\end{proof}

Also for Theorem \ref{thm-mother}, we shortly explain how this leads to the result stated as Theorem \ref{thm:mother-main-text} in the main text. First, Theorem \ref{thm:mother-main-text} only considers the case $M_t=M$ for all $t$. Second, just like presented in Subsection \ref{SbSbSct:optimization-dependent}, we can bound the constants $c_t$ by their worst-case upper bound of $1024$, and then take a minimum not over all $K$ and all choices of $t_1,\ldots,t_T$, but only over all $K$ with the fixed choice $t_k=k$. With these two simplifications, the bound becomes
\begin{equation}
    R(\vec{\alpha}^\ast) - \hat{R}_S(\vec{\alpha}^\ast)\nonumber
    \in \mathcal{O}\left(C_\mathrm{loss}\min\left\{\sqrt{\frac{K\log(MT)}{N}} + \sum_{\substack{t=1 \\t\neq t_1,\ldots,t_K}}^T M \Delta_{t} + \sqrt{\frac{\log(G_T)}{N}} +\sqrt{\frac{\log(\nicefrac{1}{\delta})}{N}}\right\}\right).
\end{equation}
Finally, once we plug in a constant confidence level $\delta$ and also consider $C_\mathrm{loss}$ as a constant dictated by the problem, we end up with the bound of Theorem \ref{thm:mother-main-text}.

\begin{remark}
In Theorem \ref{thm-mother}, we have chosen a fixed reference map $\mathcal{E}^0_t\in\mathcal{CPTP}\left((\mathbb{C})^{\otimes 2}\right)$ for every $t\in\mathbb{N}$. One could even choose different reference maps for each $k$ and for each of the $G_k$ allowed structures with exactly $k$ parameterized maps.

In Section \ref{SbSbSct:optimization-dependent}, we have taken the initial point of the optimization procedure as reference point for evaluating distances. A similar interpretation is possible in Theorem \ref{thm-mother}, however, the reference points can be more abstract. In principle, the reference maps can be chosen freely, as long as the choice is independent of the training data sample w.r.t.~which the empirical risk is evaluated.
\end{remark}

\begin{remark}\label{rmk:data-reuploading}
    We present our results for the case of a QMLM $\mathcal{E}^{\mathrm{QMLM}}_{\vec{\alpha}}(\cdot)$ acting on a quantum input state $\rho(x)$. If $x$ describes classical data, this presumes an ``encoding-first'' architecture, in which the classical-to-quantum data-encoding $x\mapsto\rho(x)$ is applied first, followed by a trainable quantum circuit.
    As observed in~\cite{gil2020input, schuld2021effect, caro2021encodingdependent}, the expressive power of a QMLM for processing classical data can significantly benefit from allowing for data re-uploading~\cite{perez2020data}. This is achieved by allowing for a more flexible form of QMLM, in which data-encoding and trainable gates can be interleaved. 
    Our results, which focus on the trainable part of the QMLM circuit, directly extend to QMLMs with data re-uploading. 
    
    This can be seen as follows: In our proofs of the metric entropy bounds from Subsection~\ref{SbSct:CoveringNumberBounds}, we already allowed for an interleaving of the trainable gates with arbitrary fixed gates. The same reasoning applies if we replace the fixed gates by encoding gates depending on the classical input data $x$, as long as they are still independent of the trainable parameters.
\end{remark}

\section{Application to quantum phase recognition}\label{Sct:application-quantum-phase-recognition}

As a second application of our prediction error bounds, we demonstrate their implications for quantum phase recognition (QPR) with quantum convolutional neural networks (QCNNs). Here, for each training example $(\ket{\psi_i},y_i)$, the encoded input is simply  $\rho(x_i)=\ketbra{\psi_i}{\psi_i}$, a pure $n$-qubit quantum state that belongs to one of four possible quantum phases of matter. The corresponding output label $y_i\in\{0,1\}^{2}$ tells us to which of the four phases $\rho(x_i)$ belongs. The goal of a quantum machine learning model for this scenario is to accurately predict, given a new input $x$, the corresponding label, and thus the phase, of the state $\psi_i$.

In our language, a QCNN acting on $n$-qubit states, as introduced in \cite{cong2019quantum}, is a QMLM $\mathcal{E}^{\mathrm{QCNN}}_{\vec{\theta}}(\cdot)$ with a particular fixed structure, explained in more detail in Section \ref{sec:numerics}, consisting of $\log(n)$ independently parameterized $2$-qubit maps, each of which is used at most $n$ times. By measuring some of the qubits and then discarding them in pooling layers, the QCNN maps an $n$-qubit input to a $2$-qubit output, on which it then performs a computational basis measurement. The phase prediction that the QCNN makes for an $n$-qubit input state is the one corresponding to the smallest of the four outcome probabilities in the computational basis measurement on the output state. This can be well approximated by running multiple gate-sharing copies of the QCNN in parallel and appropriately post-processing the single measurement outcomes. For simplicity of presentation, we showcase our bounds in the scenario of only one copy of the QCNN. However, this extends to multiple gate-sharing copies according to Corollary \ref{cor:MultipleCopiesQNN}.
Thus, we consider the loss function characterized by the loss observables
\begin{equation}
    O^{\mathrm{loss}}_{x_i, y_i} 
    = O^{\mathrm{loss}}_{y_i} 
    = \ketbra{y_i}{y_i},
\end{equation}
which is independent of $x_i$. This means the loss function is given by 
\begin{equation}
    \ell(\vec{\theta};\psi_i, y_i )
    \coloneqq \bra{y_i}\mathcal{E}^{\mathrm{QCNN}}_{\vec{\theta}}(\ketbra{\psi_i}{\psi_i})\ket{y_i}.
\end{equation}
That is, the QMLM achieves a small value of the loss function on the example $(\ket{\psi_i},y_i)$ if the probability observing $y_i$ when performing a computational basis measurement on the output state, upon input $\ket{\psi_i}$, is small.
Correspondingly, the true risk is
\begin{equation}
    R(\vec{\theta})
    = \mathbb{E}_{(\psi_i,y_i)\sim P}[\bra{y}\mathcal{E}^{\mathrm{QCNN}}_{\vec{\theta}}(\ketbra{\psi}{\psi})\ket{y}]
\end{equation}
and the empirical risk on training data $S=\{(\ket{\psi_i},y_i)\}_{i=1}^N$ is
\begin{equation}
    \hat{R}_S(\vec{\theta})  
    = \frac{1}{N} \sum_{i=1}^N \bra{y_i}\mathcal{E}^{\mathrm{QCNN}}_{\vec{\theta}}(\ketbra{\psi_i}{\psi_i})\ket{y_i}.
\end{equation}

With the scenario established, we can now apply the prediction error bound proved in Corollary \ref{cor:MultipleCopiesQNN}. Here, it takes the form: Suppose that, given training data $S$ of size $N$, our optimization of the parameters in the QCNN yields a parameter setting $\vec{\theta}^\ast=\vec{\theta}^\ast(S)$. Then, with probability $\geq 1-\delta$ over the choice of training data,
\begin{equation}
    \mathbb{E}_{(\psi_i,y_i)\sim P}[\bra{y}\mathcal{E}^{\mathrm{QCNN}}_{\vec{\theta}}(\ketbra{\psi}{\psi})\ket{y}]
    \leq \frac{1}{N} \sum_{i=1}^N \bra{y_i}\mathcal{E}^{\mathrm{QCNN}}_{\vec{\theta}}(\ketbra{\psi_i}{\psi_i})\ket{y_i} + \mathcal{O}\left( \sqrt{\frac{\log(n)^2 + \log(\nicefrac{1}{\delta})}{N}} \right).
\end{equation}
Therefore, a small training error guarantees a small prediction error already for training data size $N\in\OC(\mathrm{poly}(\log(n)))$. In other words, when using a QCNN for QPR, a good generalization error is already guaranteed for training data of size poly-logarithmic in $n$, the number of qubits. Thereby, our results provide a rigorous explanation for the good generalization behavior of QCNNs even for small training data size that was observed numerically in \cite{cong2019quantum}.

\section{Application to unitary compiling}\label{Sct:application-unitary-compiling}

The second application of our generalization guarantees to be presented here is that of learning unitaries in the sense of (quantum-assisted) unitary compiling \cite{khatri2019quantum}. Unitary compiling is the task of finding a circuit representation of a target unitary, given black-box access to that unitary.

From a learning perspective, this motivates the following problem: For each training example $(x_i,y_i)$, the input is a pure $n$-qubit state $\rho(x_i)=\ketbra{\psi_i}{\psi_i}$, and the corresponding label is the pure $n$-qubit state $\ketbra{\phi_i}{\phi_i}=U\ketbra{\psi_i}{\psi_i}U^\dagger$ obtained by unitarily evolving the input state according to the (unknown) target unitary $U$. We consider the loss function given induced by the trace distance via
\begin{equation}
    \ell (\vec{\alpha}; \ket{\psi}, \ket{\phi})
    \coloneqq \norm{\ketbra{\phi}{\phi} - \mathcal{U}^\mathrm{QMLM}_{\vec{\alpha}}(\ketbra{\psi}{\psi})}_1^2.
\end{equation}
where $\mathcal{U}^\mathrm{QMLM}_{\vec{\alpha}}(\cdot)=U_{\vec{\alpha}}(\cdot)U_{\vec{\alpha}}^\dagger$ is a (unitary) quantum machine learning model, and we take $\ket{\phi}=U\ket{\psi}$.

As we are considering a trace distance between pure states, we can rewrite the loss function in terms of the fidelity (i.e., the overlap) as
\begin{equation}
    \ell (\vec{\alpha}; \ket{\psi}, \ket{\phi})
    = 1 - \lvert\braket{\phi}{U_{\vec{\alpha}}\psi}\vert^2
    = 1 - \Tr\left[\ketbra{\phi}{\phi}\cdot \mathcal{U}^\mathrm{QMLM}_{\vec{\alpha}}(\ketbra{\psi}{\psi})\right].
\end{equation}
Hence, we see that this loss function is encompassed by our scenario, because we can write
\begin{equation}
    \ell (\vec{\alpha}; \ket{\psi}, \ket{\phi})
    = \Tr\left[O^{\mathrm{loss}}_{\psi,\phi} \cdot \mathcal{U}^\mathrm{QMLM}_{\vec{\alpha}}(\ketbra{\psi}{\psi}) \right],
\end{equation}
with loss observables $O^{\mathrm{loss}}_{\psi,\phi}= \mathds{1} - \ketbra{\phi}{\phi}$ (depending only on the quantum output, but not on the input).

With (the above rewriting of) this loss function, the expected loss, when the expectation is w.r.t.~drawing the input states independently at random from the Haar measure, becomes connected to the Hilbert-Schmidt inner product between the target unitary and the unitary implemented by the circuit. This, in turn, can be given an operational interpretation, as detailed in \cite{khatri2019quantum}.

We solve this learning problem using a QMLM with a variable structure. (See Sections \ref{sec:numerics} and \ref{sec:Methods} for more details on how this is implemented.) In this scenario, Corollary \ref{cor:VariableGateNumberImproved} implies that, if we optimize over both (discrete) structures and (continuous) parameters and obtain an output structure $\vec{k}^\ast$ with $T$ parameterized gates with a parameter setting $\vec{\alpha}^\ast=(\vec{\theta}^\ast,\vec{k}^\ast)$, then, with probability $\geq 1-\delta$ over the choice of training data of size $N$, which is drawn i.i.d.~from some distribution $P$ over pure $n$-qubit states, we are guaranteed that
\begin{equation}
    \mathbb{E}_{\ket{\psi}\sim P}\left[\norm{U\ketbra{\psi}{\psi}U^\dagger - U_{\vec{\alpha}^\ast}\ketbra{\psi}{\psi}U_{\vec{\alpha}^\ast}^\dagger}_1^2\right]
    \leq \frac{1}{N}\sum_{i=1}^N\norm{U\ketbra{\psi_i}{\psi_i}U^\dagger - U_{\vec{\alpha}^\ast}\ketbra{\psi_i}{\psi_i}U_{\vec{\alpha}^\ast}^\dagger}_1^2 + \tilde{\mathcal{O}}\left(\sqrt{\frac{T}{N}} + \sqrt{\frac{\log(\nicefrac{1}{\delta})}{N}}\right),
\end{equation}
assuming that the number of allowed structures with $T$ gates scales at most exponentially in $T$. Here, the $\tilde{\mathcal{O}}$ hides terms logarithmic in $T$.

Consequently, we know that, with high probability, the trace distance between the state obtained by applying the learned unitary on a new unseen input state (drawn at random from the data-generating distribution) and the true output state will be small if we can achieve a small average trace distance over the $N$ randomly sampled states, where $N$ scales roughly linearly in $T$. For many unitary gates of interest, namely those that can be efficiently implemented, we thus expect $T$, and thereby also $N$, to scale polynomially in $n$, the number of qubits. This is a substantial improvement over the training data sizes used in previous approaches to unitary compiling, which were often taken to be exponential in $n$ such as to uniquely determine the unknown target unitary \cite{qFactor,cincio2018learning,cincio2021machine}. This improvement comes at the cost of not compiling the target unitary exactly, but only with a certain (small) accuracy and success probability. Nevertheless, for many applications, paying this cost is worthwhile, given the significant savings in training data size guaranteed by our results. 

As a concrete example, the QFT discussed in Section \ref{SbSct:unitary-compiling} can be exactly implemented with $T\in\mathcal{O}(n^2)$ gates. In this case, our theory implies that $N\in\mathcal{O}(n^2)$ training data points are, with high probability, sufficient for good generalization. As discussed in \cite{nam2020approximate}, approximate implementations of the QFT are possible with a lower number of gates, namely with $T\in\mathcal{O}(n\log(n))$. Potentially, one can combine this insight with our result to obtain a similar improvement in the upper bound on the sufficient training data size.

\end{appendices}

\end{document}